%% file: Which_Voting_Rules_Are_More_Resilient_to_Coalitional_Manipulation.tex
\documentclass[format=acmsmall, review=false]{acmart}
\usepackage{sscw-26}
\usepackage{booktabs} % For formal tables
\usepackage[ruled]{algorithm2e} % For algorithms

\SetAlFnt{\small}
\SetAlCapFnt{\small}
\SetAlCapNameFnt{\small}
\SetAlCapHSkip{0pt}
\IncMargin{-\parindent}

\input{commands.tex}

\etocdepthtag.toc{main} % tag "begin main" (optional but better)

\setcitestyle{authoryear}

\title[Which Voting Rules Are More Resilient to Coalitional Manipulation?]{Which Voting Rules Are More Resilient to\\Coalitional Manipulation?}

\author{François Durand (Nokia Bell Labs France)}

\begin{abstract}
Which voting rules are more resilient to coalitional manipulation?
We find that a deliberately minimal model, capturing only the degree of advantage of one preference ranking over the others, can predict their relative vulnerability remarkably well.

Extending prior work on three rules, we systematically analyze all standard ordinal voting rules under the Perturbed Culture model, a variant of Impartial Culture parameterized by the extra weight assigned to one ranking.
Each rule exhibits a sharp phase transition: manipulation succeeds with high probability below a critical concentration threshold, and fails above it.
This structure reveals natural families of rules: seemingly distinct methods such as Maximin, Ranked Pairs, Schulze, and Young share identical thresholds, while Baldwin, Nanson, Kemeny, and Dodgson form another.
These groupings are driven by new, strengthened notions of Condorcet winners.
In addition, we identify a third family based on a previously introduced Condorcet notion: Black, Slater, and Copeland.

Empirically, the model displays strong predictive power. Tested on real-world datasets (Netflix and FairVote), it accurately ranks rules by vulnerability, predicts how this ranking evolves with the number of candidates, and explains why empirically similar clusters persist despite large absolute differences in manipulation rates, with a more nuanced picture for Bucklin and veto-based rules. Thus, an extremely parsimonious model with no tuning captures the comparative vulnerability of voting rules: which rules to prefer depends largely on the number of candidates alone.

Presentation video: \textcolor{blue}{\url{https://www.youtube.com/watch?v=hY4233TGUGw}}.
\end{abstract}

\begin{document}

% Title page for title and abstract only.
\begin{titlepage}

\maketitle

% Optionally include a table of contents
\vspace{1cm}
%\setcounter{tocdepth}{2} % adjust to 1 if desired
%\tableofcontents

\begingroup
  \etocsettagdepth{main}{all}  % on affiche tout le main
  \etocsettagdepth{app}{none}  % on masque les appendices
  \etocsettocstyle{\section*{Contents}}{} % ou \chapter* selon ta classe
  \tableofcontents
\endgroup

\end{titlepage}

% Paper body

\section{Introduction}

Which voting rules are more resilient to coalitional manipulation? This question has long motivated research in social choice theory, yet remains difficult to answer in general terms. The Gibbard--Satterthwaite theorem establishes that no non-trivial voting rule can be entirely immune to manipulation, but it leaves open the possibility of meaningful comparisons between rules. This paper develops a principled framework for such comparisons, based on a deliberately simple probabilistic model that captures the essential structure of voting rule vulnerability.

\subsection{Motivation}

A voting rule is \emph{coalitionally manipulable} (CM) in a given profile if a subset of voters could obtain a preferred outcome by misreporting their preferences. 
This property can be interpreted \emph{ex ante} as a vulnerability to strategic voting, or \emph{ex post} as a source of regret for sincere voters, potentially undermining trust in electoral outcomes \cite{eggers2024susceptibility,durand2015towards}.
While the Gibbard--Satterthwaite theorem implies that any non-trivial voting rule is susceptible to this phenomenon \cite{gibbard1973manipulation,satterthwaite1975strategyproofness}, it remains possible to compare the vulnerability of voting rules, in particular through their \emph{CM rate}, defined as the theoretical or empirical proportion of profiles in which the rule is coalitionally manipulable.

To explain the low empirical vulnerability of \emph{Instant-Runoff Voting} (IRV), \citet{durand2025irv} relies on the \emph{Perturbed Culture} model introduced by~\citet{williamson1967social}, in which one ranking is favored by a \emph{concentration parameter}~$\theta$ while all others are equally probable. Comparing IRV with Plurality and Plurality with Runoff, this study identifies a \emph{phase transition} characterized by a critical concentration parameter $\theta_c(f,m)$, depending on the voting rule~$f$ and the number of candidates~$m$. In the subcritical regime $\theta<\theta_c(f,m)$, the CM rate converges to~$1$ for large electorates, whereas in the supercritical regime $\theta>\theta_c(f,m)$ it converges to~$0$. The parameter $\theta_c(f,m)$ thus measures how quickly a rule becomes resistant to coalitional manipulation as preference concentration increases. For IRV and its variants, the introduction of the \emph{Super Condorcet Winner} (SCW) yields $\theta_c(f,m)=0$, showing resilience even under arbitrarily small preference concentration.

The goal of the present study is to determine whether the Perturbed Culture model can also shed light on the comparative performance of different voting rules.
While IRV is particularly resilient, one may wish to use alternative rules that satisfy axiomatic properties violated by IRV, such as monotonicity.
This raises a natural question: among standard voting rules satisfying certain axiomatic properties, which ones should be preferred in order to minimize vulnerability to coalitional manipulation? 
Our objective is to provide a principled way to rank voting rules with respect to this criterion, thereby addressing this question beyond the special case of IRV.

Rather than aiming for a realistic model fitted to empirical data through multiple parameters, we deliberately focus on a minimal model. The ability of such a simple framework to reproduce observed qualitative phenomena provides evidence that it captures key underlying mechanisms.

\subsection{Contributions}

This paper makes three main contributions.

\paragraph{Theoretical framework}

We substantially extend the work of \citet{durand2025irv} through a systematic analysis of all standard ordinal voting rules. We introduce two strengthened notions of Condorcet winners:
\begin{itemize}
\item The \emph{Pair-Safe Condorcet Winner}, characterizing Maximin, Ranked Pairs, Schulze, and Young;
\item The \emph{Set-Safe Condorcet Winner}, characterizing Baldwin, Nanson, Kemeny, and Simplified Dodgson.
\end{itemize}
Together with the existing notion of \emph{Resistant Condorcet Winner} \cite{durand2016condorcet}, which we show underlies the behavior of Black, Slater, and Copeland, these notions form a natural hierarchy, while the remaining voting rules are analyzed on a case-by-case basis. For each rule, we establish the existence of a phase transition and determine the critical concentration parameter $\theta_c(f, m)$.

\paragraph{Empirical validation} 

We test our theoretical predictions on the Netflix and
FairVote datasets using an enhanced version of the SVVAMP package \cite{durand2016svvamp,svvamp}. Our contributions include five new voting rules (Kemeny, Slater, Young, Dodgson, and Simplified Dodgson) and improved manipulation algorithms for five others (Ranked Pairs, Baldwin, Nanson, Copeland, and Kim--Roush). The key empirical findings are:
\begin{itemize}
\item Rules belonging to the same theoretical family exhibit strikingly similar CM rates;
\item The theoretical bounds from our Condorcet notions are almost tight in practice;
\item The model accurately predicts the relative ranking of voting rules by vulnerability, largely independent of the dataset.
\end{itemize}

\paragraph{Predictive power without fitting}

The model uses no tunable parameters beyond the number of candidates~$m$, yet it accurately predicts the relative vulnerability of voting rules across datasets with very different absolute CM rates.

\subsection{Related Work}

Several studies analyze coalitional manipulability from a theoretical perspective
\cite{favardin2002bordacopeland,lepelley1994vulnerability,lepelley1999kimroush,lepelley2003homogeneity}.
In particular, \citet{kim1996manipulability} show that many rules have a limiting CM rate equal to~$1$ under Impartial Culture, which prevents this criterion from meaningfully discriminating between them.

\citet{xia2023impact} identifies a form of phase transition for coalitional manipulation phenomena.
In that work, this transition is studied with respect to a budget~$B$ for the number of manipulators, whereas our analysis focuses on the concentration of voters’ preferences.
Moreover, for many standard voting rules, it is shown that in Impartial Culture, the CM rate with $n$ voters and budget~$B$ is $\Theta(\min\{ \frac{B}{\sqrt{n}}, 1 \})$.
While this result provides valuable insight into the role of the number of manipulators,
it cannot be used to compare voting rules.

Empirically, several studies have compared CM rates across voting rules
\cite{chamberlin1984observed,tideman2006collective,green2014strategic,green2016statistical}.
Analyses based on both artificial cultures and real-world datasets
\cite{durand2015towards,durand2023coalitional}
have highlighted similarities in behavior among rules such as
Maximin, Ranked Pairs, and Schulze.
We extend these results by identifying additional families of voting rules,
first at the theoretical level and then empirically,
thanks to the implementation of more precise manipulation algorithms.

Finally, a large body of work studies the algorithmic complexity of individual or coalitional manipulation, which we do not address here (see \citet[Chapter~6]{brandt2016handbook} for an overview).

\subsection{Limitations}

This study has two main limitations.
First, it excludes non-ordinal rules such as Approval, Range Voting, Majority Judgment, or STAR. 
Second, while the model accurately captures the behavior of most rules, it does not explain the near-identical empirical CM rates of Kim--Roush and Veto, and its predictions are less precise for Bucklin and Coombs (see Section~\ref{sec:simu_variation_m}).

\subsection{Roadmap}

Section~\ref{sec:framework} introduces the framework and notation. Section~\ref{sec:theoretical_results} presents the theoretical results, and Section~\ref{sec:numerical_results} reports the numerical experiments. Section~\ref{sec:future_work} concludes with perspectives for future work.

\section{Framework}\label{sec:framework}

This section introduces the framework and notation used throughout the paper.
Section~\ref{sec:voting_theory_and_model} recalls standard notions from voting theory and defines the Perturbed Culture model used in this work.
Section~\ref{sec:zoology} reviews the voting rules considered.
Finally, Section~\ref{sec:known_results} summarizes the main results established by \citet{durand2025irv}, which serve as a starting point for our analysis.

\subsection{Voting Theory and Model}\label{sec:voting_theory_and_model}

We introduce the basic notions of voting theory and the probabilistic model considered in this work.
Our notation mostly follows that of \citet{durand2025irv}.  

A \emph{discrete profile}~$P$ consists of a finite, non-empty set of candidates $\mathcal{C}(P)$, with $m(P)=|\mathcal{C}(P)|$; a finite, non-empty set of voters $\mathcal{V}(P)$, with $n(P)=|\mathcal{V}(P)|$; and, for each voter $v \in \mathcal{V}(P)$, a preference ranking $P_v$ over $\mathcal{C}(P)$.
For a ranking~$r$, let $w(r,P)$ denote its \emph{weight}, that is, the number of voters whose preference ranking is~$r$.
The \emph{total weight} of~$P$ is then $w(P)=\sum_r w(r,P)=n(P)$.

A \emph{continuous profile}~$P$ consists of a finite, non-empty set of candidates $\mathcal{C}(P)$, with $m(P)=|\mathcal{C}(P)|$; a total weight $w(P)\in(0,\infty)$; and, for each ranking~$r$ over $\mathcal{C}(P)$, a weight $w(r,P)\in[0,\infty)$, satisfying $\sum_r w(r,P)=w(P)$.

For any profile~$P$, discrete or continuous, the \emph{normalized profile}~$\bar{P}$ is the continuous profile defined by
\( w(r,\bar{P}) = \frac{w(r,P)}{w(P)}. \)

We use the following notation for restricted profiles.  
For $K \subseteq \mathcal{C}(P)$, let $P_{K}$ denote the restriction of~$P$ to the candidates in~$K$.
For a candidate~$c$ and a position~$k$, let $P^{r(c)=k}$ (resp.\ $P^{r(c)\leq k}$) denote the sub-profile of voters ranking~$c$ in position~$k$ (resp.\ among their top~$k$ positions).
For $c,d \in \mathcal{C}(P)$, let $P^{c \succ d}$ denote the restriction of~$P$ to voters preferring~$c$ to~$d$, and similarly for multiple comparisons (e.g., $P^{c \succ d \text{ and } c \succ e}$).
These notations can be combined to restrict simultaneously to a subset of candidates and a subset of voters.  

A \emph{voting rule}~$f$ maps any discrete profile~$P$ to a candidate in $\mathcal{C}(P)$.  
Almost all rules considered here extend naturally to continuous profiles and are \emph{homogeneous}, in the sense that their outcome~$f(P)$ depends only on the normalized profile~$\bar{P}$.
The only exceptions are Young, Dodgson, and Simplified Dodgson (see Section~\ref{sec:def_dodgson_young}).

For candidates~$c,d \in \mathcal{C}(P)$, let
\(
W(c,d,P)=w(P^{c \succ d})
\)
denote the total weight of voters preferring~$c$ to~$d$ in profile~$P$.
The matrix~$W(P)$ with entries~$W(c,d,P)$ is the \emph{weighted majority matrix} of~$P$.
The \emph{unweighted majority matrix}~$M(P)$ is obtained from~$W(P)$ by setting
$M(c,d,P)$ equal to $1$ (resp.\ $\tfrac{1}{2}$, $0$) if $W(c,d,P)$ is greater than
(resp.\ equal to, less than) $W(d,c,P)$; by convention, diagonal entries are set to~$0$.
A candidate~$c$ is the \emph{Condorcet winner} of a profile~$P$ if $M(c,d,P)=1$ for every other candidate~$d$.
A voting rule~$f$ is \emph{Condorcet-consistent} if it elects the Condorcet winner whenever one exists.
The \emph{Smith set} of a profile~$P$ is the smallest non-empty subset
$K \subseteq \mathcal{C}(P)$ such that $M(c,d,P)=1$ for all
$c \in K$ and $d \in \mathcal{C}(P)\setminus K$.
In particular, if a Condorcet winner exists, it is the unique member of the Smith set.

For a discrete profile~$P$, we say that a voting rule~$f$ is \emph{coalitionally manipulable} (CM) in~$P$ (or that $P$ is CM in~$f$) if there exists a \emph{target profile}~$Q$ with the same candidates and voters such that $f(Q)\neq f(P)$, and every voter who changes their ballot prefers $f(Q)$ to $f(P)$ according to their true ranking~$P_v$.
For a continuous profile~$P$, assuming that~$f$ is defined on continuous profiles, we say that $f$ is coalitionally manipulable (CM) in~$P$ if there exists a target profile~$Q$ with the same candidates and total weight such that $f(Q)\neq f(P)$, and for every ranking~$r$, whenever $w(r,Q)<w(r,P)$ the ranking~$r$ places $f(Q)$ above $f(P)$. Throughout the paper, we use the terms ``coalitionally manipulable'' (CM) and ``susceptible to coalitional manipulation'' interchangeably. Likewise, we use ``non-coalitionally manipulable'' (non-CM) interchangeably with ``immune to coalitional manipulation''.

\medskip
We now introduce the probabilistic model and the quantitative measure of coalitional manipulability central to our analysis, following \citet{durand2025irv}.

Given integers $n,m>0$ and a \emph{concentration parameter} $\theta\in(0,1]$, the \emph{Perturbed Culture} generates a random discrete profile~$P$ with $\mathcal{C}(P)=\{1,\ldots,m\}$ and $\mathcal{V}(P)=\{1,\ldots,n\}$.  
Each voter independently adopts the reference ranking $(1,\ldots,m)$ with probability~$\theta$, and a uniformly random ranking otherwise.
We exclude the case $\theta=0$, which corresponds to the classical model of \emph{Impartial Culture}.

We denote by $\hat{P}$ the \emph{expected normalized profile}, or simply the \emph{expected profile}, defined as the continuous profile in which each ranking has weight $\tfrac{1-\theta}{m!}$, except for the reference ranking $(1,\ldots,m)$, which has weight $\theta+\tfrac{1-\theta}{m!}$. For concision, we leave the dependence of $\hat{P}$ on $\theta$ and $m$ implicit. The expected profile $\hat{P}$ serves as the asymptotic limit of normalized profiles as the number of voters grows, and plays a central role in the proofs.

For a dataset of profiles, the \emph{CM rate} of a voting rule~$f$ is the fraction
of profiles that are CM in~$f$; for a probabilistic model, it is the probability
that a random profile is CM.
We write $\rho(f,m,n,\theta)$ for the CM rate of~$f$ in the Perturbed Culture
with parameters~$m,n,\theta$.

\subsection{Voting Rules}\label{sec:zoology}

\newcommand{\defrule}[3]{\smallskip\emph{#1 (#2).} #3}

We now review the voting rules considered in this paper.
To keep the presentation concise, rules are grouped into broad categories, although many could naturally belong to more than one.
All rules are assumed to use a tie-breaking mechanism; unless otherwise specified (Slater and Copeland), our theoretical results do not depend on it, and numerical simulations break ties in favor of candidates with smaller indices.
Readers familiar with standard voting rules may safely jump directly to Section~\ref{sec:known_results}. Some additional voting rules are considered in the technical appendix.

\subsubsection{Score-based rules}\label{sec:score_based_rules}

Elect the candidate~$c$ with maximal score $s_f(c, P)$.

\defrule{Plurality}{Plu}{$s_\Plu(c, P) = w(P^{r(c) = 1})$.}

\defrule{Borda}{Bor}{$s_\Bor(c, P) = \sum_{k=1}^{m(P)} (m(P) - k)\,w(P^{r(c) = k})$.}

\defrule{Veto}{Vet}{$s_\Vet(c, P) = -\,w(P^{r(c) = m(P)})$.}

\defrule{Maximin}{Max}{$s_\Max(c, P) = \min_{d \in \mathcal{C}(P) \setminus \{c\}} W(c, d, P)$.}

\defrule{Copeland}{Cop}{$s_\Cop(c,P)=|\{d\in\mathcal{C}(P):M(c,d,P)=1\}|
+ \alpha\,|\{d\in\mathcal{C}(P)\setminus\{c\}:M(c,d,P)=\tfrac{1}{2}\}|$, 
with $\alpha\in[0,1]$. Our theoretical results in the main body do not depend on $\alpha$; in simulations, we set $\alpha=\tfrac{1}{2}$.}

\defrule{Bucklin}{Buc}{$s_\Buc(c,P) = (-\mu(c,P),\, w(P^{r(c)\leq \mu(c,P)}))$, 
where $\mu(c,P)$ is the median rank of~$c$. Scores are compared lexicographically.}

\subsubsection{Penalty-based rules (candidate level)}\label{sec:def_dodgson_young}
Elect the candidate~$c$ with minimal penalty $p_f(c,P)$.

\defrule{Young}{You}{$p_\You(c,P)$ is the minimal number of voters to remove so that $c$ becomes the Condorcet winner. If this is impossible, $p_\You(c,P)=n(P)+1$.}

\defrule{Dodgson}{Dod}{$p_\Dod(c,P)$ is the minimal number of adjacent swaps in voters’ rankings required to make $c$ the Condorcet winner.}

\defrule{Simplified Dodgson}{SD}{$p_\SD(c,P)=\sum_{d\in\mathcal{C}(P)\setminus\{c\}}\max(0, \lfloor\tfrac{n(P)}{2}\rfloor+1-W(c,d,P))$.  
This corresponds to the minimal number of pairwise comparisons by individual voters that must be changed for $c$ to become a Condorcet winner, without requiring ballots to remain transitive as in the Dodgson rule.\footnote{We adopt the terminology of \citet{tideman2006collective}, although the two rules coincide only when the number of voters is odd; in general, our definition matches the rule~$V$ of \citet{caragiannis2014socially}. Our results apply to both.} We include this rule mainly to provide insight into the behavior of the Dodgson rule (Section~\ref{sec:sscw}).}

\subsubsection{Penalty-based rules (ranking level)}\label{sec:def_kemeny_slater}
Find a ranking with minimal penalty $p_f(r,P)$, and elect its top candidate.

\defrule{Kemeny}{Kem}{$p_\Kem(r,P)=\sum_{(c,d)\in\mathcal{C}(P)^2:\,c\succ_r d}W(d,c,P)$.}

\defrule{Slater}{Sla}{$p_\Sla(r,P)=\sum_{(c,d)\in\mathcal{C}(P)^2:\,c\succ_r d}M(d,c,P)$.}

\subsubsection{Elimination rules}

Eliminate one or more candidates in successive rounds until a single winner remains.

\defrule{IRV, Baldwin, Coombs}{IRV, Bal, Coo}{Iteratively eliminate the candidate with the lowest plurality, Borda, or veto score, respectively.}

\defrule{Nanson, Kim--Roush}{Nan, KR}{At each round, eliminate all candidates whose Borda or veto score is below the average, respectively.\footnote{For Kim--Roush, the original definition uses ``strictly below'' as here~\cite{kim1996manipulability}, whereas for Nanson it is ``at most equal''~\cite{niou1987note}. This distinction does not affect our results.}\footnote{We do not include the plurality-based analogue, \emph{IRV-Average}, since it exhibits essentially the same behavior as IRV with respect to coalitional manipulation~\cite{durand2023coalitional,durand2025irv}.}}

\defrule{Plurality with runoff}{PR}{Keep the top two candidates by plurality, then elect the one winning their pairwise contest.}

\subsubsection{Other Condorcet rules}\label{sec:other_condorcet_rules}

This last category gathers standard Condorcet-consistent rules, complementing those already mentioned: Maximin, Copeland, Young, Dodgson, Kemeny, and Slater.

\defrule{Black}{Bla}{Elect the Condorcet winner if one exists; otherwise elect the Borda winner.}

\defrule{Ranked Pairs}{RP}{Order all pairs of candidates by decreasing $W(c,d,P)$.  
Iteratively lock each pair into a directed graph unless this creates a cycle.  
Elect the unique source of the final graph \cite{tideman1987clones}.}

\defrule{Schulze}{Sch}{Define the \emph{strength} of $c$ over $d$ as the width of the widest path from $c$ to $d$ in the weighted majority graph induced by~$W(P)$:
\[
\Strength(c,d,P) = \max_{\text{paths $p$ from $c$ to $d$}} \;\min_{(i,j)\in p} W(i,j,P).
\]
Elect a candidate $c$ such that $\Strength(c,d,P) \geq \Strength(d,c,P)$ for all $d$ \cite{schulze2011method}.}

\subsection{Known Results: Phase Transitions in the Perturbed Culture}\label{sec:known_results}

The Perturbed Culture model was introduced by \citet{williamson1967social} as a conceptual
tool to assess the practical relevance of the Condorcet paradox.
Under Impartial Culture, the probability that no Condorcet winner exists converges to a
positive limit as the electorate grows, suggesting a non-negligible prevalence of this
phenomenon in practice.
By contrast, under the Perturbed Culture, as soon as the concentration parameter~$\theta$
is strictly positive, however small, this probability converges to~$0$, indicating that the Condorcet paradox may be exceptional.
Although both models are highly stylized and make no claim of realism, this latter
prediction turns out to be more consistent with empirical observations, as Condorcet winners
are extremely frequent in real-world datasets \cite{tideman2006collective,durand2023coalitional}.

\citet{durand2025irv} uses the Perturbed Culture model to study coalitional
manipulability for three voting rules: Plurality, Plurality with Runoff, and IRV.
In the case of Plurality, for $m\geq 2$, it is shown that there exists a \emph{critical concentration parameter} $\theta_c(\Plu,m)=\frac{m-2}{3m-2}$ such that
\begin{itemize}
\item if $\theta<\theta_c(\Plu,m)$, then $\lim_{n\to\infty}\rho(\Plu,m,n,\theta)=1$;
\item if $\theta>\theta_c(\Plu,m)$, then $\lim_{n\to\infty}\rho(\Plu,m,n,\theta)=0$.
\end{itemize}
This result exhibits a \emph{phase transition}, that is, an abrupt change in the
asymptotic behavior of the CM rate when the concentration parameter~$\theta$ crosses the
critical threshold. 
This phenomenon is illustrated in Figure~\ref{fig:cm_rate_of_theta_and_n_v_plu}, which displays the CM rate of Plurality as a function of~$\theta$ for increasing values of~$n$.
As $n$ grows, the curve takes on a sigmoidal shape and converges to a step function.
Moreover, the convergence is shown to be exponentially fast, implying that the
asymptotic behavior predicted by the theorem quickly becomes relevant as the electorate grows.
For Plurality with Runoff, similar results hold, but with a smaller critical
parameter, $\theta_c(\PR,m)=\frac{m-3}{5m-3}$ for $m\geq 3$.

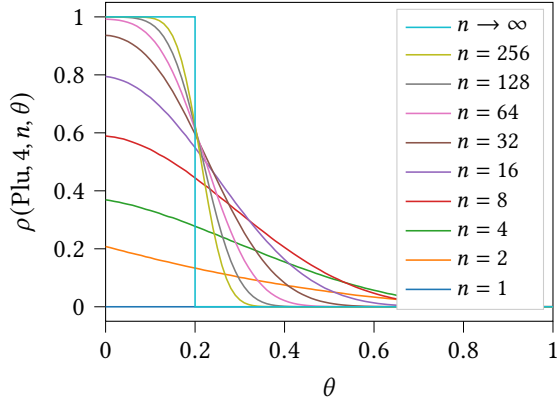
\begin{figure}
	\centering
	\input{cm_rate_of_theta_and_n_v_plu.tex}
	\caption{CM rate of Plurality as a function of \( \theta \) for different values of \( n \) with \( m = 4 \). Curves for finite \( n \) are based on Monte Carlo simulations with 1,000,000 profiles per point. The limiting curve as \( n \to \infty \) is given by theory. Source: \citet{durand2025irv}.}
	\label{fig:cm_rate_of_theta_and_n_v_plu}
\end{figure}

To address the case of IRV, \citet{durand2025irv} introduces the notion of a
\emph{Super Condorcet Winner} (SCW), defined as a candidate~$c$ such that, for every
subset $K\subseteq\mathcal{C}(P)$ containing~$c$ with $|K|\geq 2$,
\begin{equation}\label{eq:def_scw}
s_\Plu(c,P_K)>\frac{w(P)}{|K|},
\end{equation}
that is, $c$ has a plurality score strictly above average in all these restricted profiles.
Whenever such a candidate exists, it is elected by IRV and the profile is immune to coalitional manipulation.
As in the case of Condorcet winners, for any concentration parameter~$\theta>0$, an SCW
exists with high probability, namely with probability tending to~$1$ as the electorate
grows, implying that IRV is immune to coalitional manipulation.
As a consequence, the critical concentration parameter for IRV is
$\theta_c(\IRV,m)=0$.
Moreover, although the existence of an SCW is only a sufficient condition, it is also observed in empirical datasets that it explains most cases in which IRV is immune to coalitional manipulation.

In all cases, the proof strategy follows the same general pattern.
One first analyzes the expected profile, which is susceptible to coalitional manipulation when the concentration parameter~$\theta$ is small enough and immune to coalitional manipulation otherwise.
This property is then shown to be stable in a neighborhood of the expected profile, with technical subtleties arising in some cases from the need to control stability both around the sincere profile and around potential manipulated profiles.
Finally, the result is extended to discrete profiles in the limit $n\to\infty$ using the weak law of large numbers.
Altogether, this approach shows that a simple and mathematically tractable model suffices to explain qualitative phenomena such as the low vulnerability of IRV observed in empirical data.

The paper also shows that a deliberately constructed voting rule may fail to exhibit a phase transition with a well-defined critical concentration parameter $\theta_c(f,m)$.
Nevertheless, it is always possible to define a \emph{lower critical concentration parameter} $\theta_{\ell}(f,m) \in [0, 1]$ and an \emph{upper critical concentration parameter} $\theta_u(f,m) \in [0, 1]$ as the largest and smallest values, respectively, such that the following holds for every $\theta \in (0, 1]$:
\begin{itemize}
\item if $\theta<\theta_\ell(f,m)$, then $\lim_{n\to\infty}\rho(f,m,n,\theta)=1$;
\item if $\theta>\theta_u(f,m)$, then $\lim_{n\to\infty}\rho(f,m,n,\theta)=0$.
\end{itemize}
The critical concentration parameter $\theta_c(f, m)$ is then simply defined as their common value when it exists.

\section{Theoretical Results}\label{sec:theoretical_results}

\citet{durand2025irv} established the existence of a phase transition for three voting rules, while observing that one can deliberately construct exotic voting rules to serve as counterexamples.
From a theoretical perspective, our main objective is to continue this research program by establishing the existence of a phase transition for each voting rule studied in this paper, and by identifying the corresponding critical concentration parameter.

In the same spirit as the IRV analysis based on the Super Condorcet Winner (SCW), we rely on several strengthened notions of the Condorcet winner to capture the behavior of other rules.
We introduce the \emph{Pair-Safe Condorcet Winner} (PSCW) in Section~\ref{sec:pscw} to analyze Maximin, Ranked Pairs, Schulze, and Young.
We then define the \emph{Set-Safe Condorcet Winner} (SSCW) in Section~\ref{sec:sscw} for Baldwin, Nanson, Kemeny, and Simplified Dodgson.
In Section~\ref{sec:rcw}, we rely on the existing notion of a \emph{Resistant Condorcet Winner} (RCW)~\cite{durand2016condorcet} to study Black, Slater, and Copeland.
These notions form a hierarchy of implications:
\[
\text{RCW} \;\Rightarrow\; \text{SSCW} \;\Rightarrow\; \text{PSCW} \;\Rightarrow\; \text{CW},
\]
whereas the Super Condorcet Winner (SCW) only implies CW and is logically independent of the others.
When dealing with a Condorcet notion, we will also use the notation $\theta_c(\cdot, m)$
to denote the critical concentration parameter above which such a winner exists with high probability, and below which it does not.
For example, $\theta_c(\textrm{SCW}, m) = 0$.

Finally, in Section~\ref{sec:main_theorem}, we briefly discuss the remaining rules and present our main theorem, which shows that every voting rule considered in this paper undergoes a phase transition and characterizes the associated critical concentration parameter.

For the sake of clarity and concision, some proofs are deferred to the technical appendix.

\subsection{Pair-Safe Condorcet Winner (PSCW)}\label{sec:pscw}

To build intuition, we start with Maximin.
Let $c$ be the winner, and consider a manipulation attempt in favor of some candidate~$d$.
Manipulators cannot increase the pairwise score of $d$ against $c$, which is $w(P^{d \succ c})$.
Therefore, for the manipulation to succeed, they must reduce the score of $c$ against some candidate~$e$ to at most this value.  
Note that if $c$ is not a Condorcet winner, then candidate $e$ may coincide with $d$; otherwise, $e$ must be a third candidate.
In the extreme case where all manipulators rank $e$ above $c$, the score of $c$ against $e$ drops to $w(P^{c \succ d \text{ and } c \succ e})$, which comes only from sincere voters preferring $c$ to $e$.  
Manipulation therefore requires
\(
w(P^{c \succ d \text{ and } c \succ e}) \leq w(P^{d \succ c}).
\)
The negation of this inequality yields a sufficient condition under which Maximin is immune to coalitional manipulation, motivating the following definition.

\begin{definition}
A candidate $c$ is a \emph{Pair-Safe Condorcet Winner} (PSCW) if, for every pair of other candidates $(d,e)$ (not necessarily distinct),
\begin{equation}\label{eq:def_pscw}
w(P^{c \succ d \text{ and } c \succ e}) > w(P^{d \succ c}).
\end{equation}
Intuitively, for any opponent~$d$, manipulators supporting~$d$ cannot make the pairwise contest of $c$ against some candidate~$e$ appear as unfavorable as the comparison of $d$ against~$c$.
\end{definition}

Taking $d=e$ in Equation~\eqref{eq:def_pscw} shows that every PSCW is a Condorcet winner.  
The converse does not hold: a candidate may be a Condorcet winner without being a PSCW, as illustrated in Table~\ref{tab:ex_scw_no_pscw_rules_not_CM}. 
This profile will be reused throughout the paper; the technical appendix collects the claims it supports for easy reference.

\newcommand{\tabExScwNoPscwRulesNotCM}[2]{% Argument: #1 for \label, #2 main/app
\begin{table}[ht]
\centering
\ifthenelse{\equal{#2}{main}}{%
  \caption{A profile where candidate~$A$ is the CW, but not a PSCW (see Appendix~\ref{sec:appendix_analysis_table_1}).}%
}{%
  \caption{A profile where (1) $A$ is the SCW; (2) $A$ is not a PSCW; (3) \Max, \RP, \Sch, \SC, \Vie, \You, \Bal, \Dod, \Kem, \Nan, and \SD\ are immune to CM. Adding a candidate $A'$ below $A$ in all rankings yields the same result for \Cop\ and \Sla.}%
}
\label{#1}
\(\begin{array}{ccc}
	\hline
	5 & 4 & 2 \\ \hline
	B & A & C \\
	A & C & A \\
	C & B & B \\ \hline
\end{array}\)
% Cf. draft Young p. 16
\end{table}%
}
\tabExScwNoPscwRulesNotCM{tab:ex_scw_no_pscw_rules_not_CM}{main}

To extend our analysis beyond Maximin, we introduce the notion of a \emph{Maximin-like} voting rule, which clearly includes Ranked Pairs and Schulze.

\begin{definition}
A rule $f$ is \emph{Maximin-like} if for every profile $P$ and pair of candidates $(c,d)$,
\[
\min_{e \in \mathcal{C}(P)\setminus\{c\}} W(c,e,P) > W(d,c,P)
\;\;\;\Rightarrow\;\;\;
f(P)\neq d.
\]
In words, if the pairwise score of $c$ against every opponent exceeds the score of $d$ against $c$, then $d$ cannot win.
\end{definition}

\begin{restatable}{proposition}{ThmMaximinLikeWhichRules}
\label{thm:maximin_like_which_rules}
Maximin, Ranked Pairs, and Schulze are Maximin-like voting rules.
\end{restatable}

By the same reasoning as for Maximin, we obtain the following result.
\begin{restatable}{theorem}{ThmMaximinLikeRulesProperty}
\label{thm:maximin_like_rules_property}
Let $f$ be a Maximin-like voting rule.  
If a candidate $c$ is the PSCW of a profile $P$, then $f(P)=c$ and $P$ is non-CM.
\end{restatable}

Young fits into this picture as a particular case.

\begin{restatable}{proposition}{ThmYoungNotMaximinLike}
\label{thm:young_not_maximin_like}
The Young rule is not Maximin-like.
However, it satisfies the conclusion of Theorem~\ref{thm:maximin_like_rules_property}:
if $c$ is the PSCW of a profile $P$, then $\You(P)=c$ and $P$ is non-CM.
\end{restatable}

The first statement is illustrated by Table~\ref{tab:ex_you_is_not_maximin_like}, while the second relies on more subtle arguments and is proved in Appendix~\ref{sec:young_and_pscw}.

\newcommand{\tabExYouIsNotMaximinLike}[2]{% Argument: #1 for \label, #2 main/app
\begin{table}[ht]
\centering
\ifthenelse{\equal{#2}{main}}{%
  \caption{\You\ is not Maximin-like. Although $\min_e W(A,e,P) > W(B,A,P)$, candidate~$B$ is elected. Each column denotes a voter block with uniform permutations of $C_1, C_2, C_3$ (see Appendix~\ref{sec:appendix_analysis_table_2}).}%
}{%
  \caption{\You\ and \Vie\ are not Maximin-like. Although $\min_e W(A,e,P) > W(B,A,P)$, both elect $B$. Each column denotes a voter block with uniform permutations of $C_1, C_2, C_3$.}%
}
\label{#1}% Cf. draft Young p. 3
\(\begin{array}{c c c}
	\hline
	30        & 72        & 72        \\ \hline
	A         & C_\bullet & B         \\
	B         & C_\bullet & C_\bullet \\
	C_\bullet & A         & C_\bullet \\
	C_\bullet & B         & A         \\
	C_\bullet & C_\bullet & C_\bullet \\ \hline
\end{array}\)
\end{table}%
}
\tabExYouIsNotMaximinLike{tab:ex_you_is_not_maximin_like}{main}

For Maximin, Ranked Pairs, Schulze, and Young, the existence of a PSCW is only a sufficient condition for being immune to coalitional manipulation.
Indeed, the profile in Table~\ref{tab:ex_scw_no_pscw_rules_not_CM} admits no PSCW, yet all these rules remain immune to coalitional manipulation; the verification is provided in the technical appendix.

We now examine how the previous results apply to the Perturbed Culture model.

\begin{restatable}{theorem}{ThmThetaCriticalPSCW}
\label{thm:theta_critical_pscw}
In the Perturbed Culture model, the critical concentration parameter for the existence of a Pair-Safe Condorcet Winner is
\[
\theta_c(\PSCW, m) = \frac{1}{7}.
\]
\end{restatable}

The proof, which is given in Appendix~\ref{sec:critical_thetas_rule_related_to_pscw}, proceeds by determining whether a PSCW exists in a neighborhood of the expected profile, and then applying the weak law of large numbers.
A direct consequence is the following.

\begin{restatable}{corollary}{ThmThetaUIfProtectedByPSCW}
\label{thm:theta_u_if_protected_by_pscw}
Let $f$ be a voting rule that is non-CM whenever a PSCW exists.
Then 
\[
\theta_u(f, m) \le \frac{1}{7}.
\]
\end{restatable}

This applies in particular to Maximin, Ranked Pairs, Schulze, and Young.
In Appendix~\ref{sec:critical_thetas_rule_related_to_pscw}, we further show that, conversely, for $\theta < \tfrac{1}{7}$, these rules are coalitionally manipulable with high probability, which implies that $\theta_\ell(f, m) \ge \tfrac{1}{7}$.
Altogether, this establishes that $\theta_\ell(f, m) = \theta_u(f, m) = \tfrac{1}{7}$ for these rules; this result will be integrated into the main theorem (Theorem~\ref{thm:critical_thetas}).

\subsection{Set-Safe Condorcet Winner (SSCW)}\label{sec:sscw}

To build intuition, we first consider Baldwin and Nanson.
For a manipulation in favor of some candidate $d$ to succeed, the current winner $c$ must be eliminated while $d$ remains in contention.
This requires the existence of a subset $S$ of candidates containing both $c$ and $d$ such that, after manipulation, the Borda score of $c$ within $S$ is no greater than the average (a reasoning similar to that underlying the notion of Super Condorcet Winner for IRV).
By contraposition, this yields a sufficient condition for Baldwin and Nanson to be immune to coalitional manipulation, motivating the definition below in its original form, given in Equation~\eqref{eq:def_sscw_s}.

\begin{definition}
A candidate $c$ is a \emph{Set-Safe Condorcet Winner} (SSCW) if any (and hence all) of the following equivalent conditions are satisfied.
\begin{itemize}
\item For every candidate $d \neq c$ and every subset $S \subseteq \mathcal{C}(P)$ containing $c$ and $d$,  
	\begin{equation}\label{eq:def_sscw_s}
	\sum_{e \in S \setminus \{c\}} w(P^{c \succ d \text{ and } c \succ e}) \;>\; \frac{|S|-1}{2}\, w(P).
	\end{equation}
\item For every candidate $d \neq c$ and every subset $T \subseteq \mathcal{C}(P)\setminus\{c\}$ containing $d$,  
	\begin{equation}\label{eq:def_sscw_t}
	\sum_{e \in T} w(P^{c \succ d \text{ and } c \succ e}) \;>\; \frac{|T|}{2}\, w(P).
	\end{equation}
\item For every candidate $d \neq c$ and every subset $T \subseteq \mathcal{C}(P)\setminus\{c\}$ containing $d$,  
	\begin{equation}\label{eq:def_sscw_kemeny_style}
	\Big( w(P^{c \succ d}) - \tfrac{w(P)}{2} \Big) 
	+ \sum_{e \in T \setminus \{d\}} \Big( w(P^{c \succ d \text{ and } c \succ e}) - \tfrac{w(P)}{2} \Big) > 0.
	\end{equation}
\item For every candidate $d \neq c$,
	\begin{equation}\label{eq:def_sscw_computable}
	\Big( w(P^{c \succ d}) - \tfrac{w(P)}{2} \Big) 
	+ \sum_{e \notin \{c, d\}} 
	\min\Big(0,\, w(P^{c \succ d \text{ and } c \succ e}) - \tfrac{w(P)}{2}\Big) > 0.
	\end{equation}
\end{itemize}
\end{definition}

In summary, Equations~\eqref{eq:def_sscw_s} and~\eqref{eq:def_sscw_t} are convenient for Baldwin and Nanson, Equation~\eqref{eq:def_sscw_kemeny_style} for Kemeny, and Equation~\eqref{eq:def_sscw_computable} for Simplified Dodgson.

Equation~\eqref{eq:def_sscw_s} reflects our initial reasoning, while Equation~\eqref{eq:def_sscw_t} is merely a reformulation. In both cases, note that the summation includes the case $e=d$. Equation~\eqref{eq:def_sscw_t} can be equivalently rewritten as Equation~\eqref{eq:def_sscw_kemeny_style}, which can be interpreted as follows: the margin of $c$ over $d$ (typically a win) cannot be offset by negative margins against other opponents~$e$ that may arise after manipulation.

The definitions given in Equations~\eqref{eq:def_sscw_s}, \eqref{eq:def_sscw_t}, and~\eqref{eq:def_sscw_kemeny_style} are computationally costly, as they require considering all subsets of opponents. In Equation~\eqref{eq:def_sscw_kemeny_style}, for a fixed~$d$, the worst case is obtained by selecting precisely those candidates~$e$ that can defeat~$c$ after manipulation. This observation leads to Equation~\eqref{eq:def_sscw_computable}, an equivalent formulation that can be tested in polynomial time.

By considering sets $T$ of size one or two in Equation~\eqref{eq:def_sscw_t}, and using the identity $w(P) = w(P^{c \succ d}) + w(P^{d \succ c})$, Equation~\eqref{eq:def_pscw} is recovered, showing that every SSCW is also a PSCW.
This observation also motivates the terminology: \emph{set-safe} refers to arbitrary sets of opponents~$T$, whereas \emph{pair-safe} restricts attention to pairs (possibly identical).
In the same spirit, the usual Condorcet winner can be seen as \emph{single-opponent-safe}, as it corresponds to the case where $T$ contains only one opponent.
The implication is strict: in Table~\ref{tab:ex_pscw_no_sscw}, candidate~$A$ is the PSCW (and hence a CW) but not an SSCW.

\newcommand{\tabExPscwNoSscw}[2]{% Argument: #1 for \label, #2 main/app
\begin{table}[ht]
\centering
\ifthenelse{\equal{#2}{main}}{%
  \caption{A profile where candidate~$A$ is the PSCW, but not an SSCW (see Appendix~\ref{sec:appendix_analysis_table_3}).}%
}{%
  \caption{A profile where (1) $A$ is the PSCW; (2) $A$ is not an SSCW; (3) \Vie\ is CM in favor of $B$.}%
}
\(\begin{array}{c c c c}
	\hline
	2 & 6 & 5 & 6 \\ \hline
	A & B & C & D \\
	B & A & A & A \\
	C & C & B & B \\
	D & D & D & C \\ \hline
\end{array}\)
\label{#1}
\end{table}%
}
\tabExPscwNoSscw{tab:ex_pscw_no_sscw}{main}

For Kemeny, consider a manipulation attempt in favor of candidate~$d$.
The resulting winning ranking would then have to take the form \((d, e_1, \ldots, e_k, c, f_1, \ldots, f_\ell)\).
If Equation~\eqref{eq:def_sscw_kemeny_style} holds for the opponent set \(T = \{d, e_1, \ldots, e_k\}\), however, then moving $c$ to the top of the ranking strictly reduces the Kemeny penalty, and the manipulation fails.
Hence, the existence of an SSCW also makes Kemeny immune to coalitional manipulation.
Intuitively, the condition ensures that manipulators cannot interpose other candidates $e_1, \ldots, e_k$ between $d$ and $c$ in the winning ranking so as to prevent $c$ from resurfacing above $d$. 

For Simplified Dodgson, in Equation~\eqref{eq:def_sscw_computable}, the first term can be interpreted (up to rounding) as the number of points candidate~$d$ must recover against~$c$, while the second term is the opposite of the points that $c$ must recover in its defeats after manipulation.
If the inequality holds, the score of $d$ after manipulation remains below that of~$c$.
Hence, the existence of an SSCW also makes Simplified Dodgson immune to coalitional manipulation.

\begin{restatable}{proposition}{ThmBaldwinNansonKemenySSCW}
\label{thm:baldwin_nanson_kemeny_sscw}
Let $f$ be one of Baldwin, Nanson, Kemeny, or Simplified Dodgson.
If a candidate $c$ is the SSCW of a profile $P$, then $f(P)=c$ and $P$ is non-CM.
\end{restatable}

For Baldwin, Nanson, Kemeny, and Simplified Dodgson, the presence of an SSCW is thus a sufficient condition for immunity to coalitional manipulation, but not a necessary one, as illustrated by the profile already introduced in Table~\ref{tab:ex_scw_no_pscw_rules_not_CM}.

As for the PSCW, analyzing a neighborhood of the expected profile yields the following two results (see Appendix~\ref{sec:appendix_baldwin_family}).

\begin{restatable}{theorem}{ThmThetaCriticalSSCW}
\label{thm:theta_critical_sscw}
In the Perturbed Culture model, the critical concentration parameter for the existence of a Set-Safe Condorcet Winner is
\[
\theta_c(\SSCW, m) = \frac{m-2}{4m-5}.
\]
\end{restatable}

\begin{restatable}{corollary}{ThmThetaUIfProtectedBySSCW}
\label{thm:theta_u_if_protected_by_sscw}
Let $f$ be a voting rule that is non-CM whenever an SSCW exists.
Then 
\[
\theta_u(f, m) \le \frac{m-2}{4m-5}.
\]
\end{restatable}

This applies in particular to Baldwin, Nanson, Kemeny, and Simplified Dodgson. Conversely, Appendix~\ref{sec:appendix_baldwin_family} shows that, for $\theta < \tfrac{m-2}{4m-5}$, these rules are susceptible to coalitional manipulation with high probability. Altogether, this establishes that $\theta_c(f, m) = \tfrac{m-2}{4m-5}$, a result that will be incorporated into the main theorem (Theorem~\ref{thm:critical_thetas}).

Finally, Dodgson stands apart within this family, since the presence of an SSCW does not guarantee immunity to coalitional manipulation (Table~\ref{tab:ex_sscw_no_rcw_dod_cm}). On the other hand, there also exist profiles with no SSCW in which Dodgson is still immune to coalitional manipulation (Table~\ref{tab:ex_scw_no_pscw_rules_not_CM}).
Dodgson would be protected by the SSCW condition, like Simplified Dodgson, if the number of swaps required always equaled the number of points to be regained.
In practice, however, so-called ``useless swaps'' may be needed.\footnote{Useless swaps also explain why it is NP-hard to determine the winner for the Dodgson rule \cite{bartholdi1989voting}.}  
For example, suppose that $c$ already defeats $d$ in pairwise comparison but loses to $e$, and consider a ranking $c \succ d \succ e$.
Moving $e$ above $c$ then requires two swaps, yet this changes the outcome of $c$ versus $e$ by only one point.
Such situations, however, do not arise in the proofs under the Perturbed Culture given in Appendix~\ref{sec:appendix_baldwin_family}.
As a result, Dodgson nevertheless shares the same critical concentration parameter as Simplified Dodgson, as we will state in the main theorem (Theorem~\ref{thm:critical_thetas}).

\newcommand{\tabExSscwNoRcwDodCM}[2]{% Argument: #1 for \label, #2 main/app
\begin{table}[ht]
\centering
\ifthenelse{\equal{#2}{main}}{%
  \caption{A profile where candidate~$A$ is the SSCW, yet Dodgson is CM for~$B$ (see Appendix~\ref{sec:appendix_analysis_table_4}).}%
}{%
  \caption{A profile where (1) $A$ is the SSCW; (2) $A$ is not an RCW; (3) \Dod\ is CM for $B$.}%
}
\label{#1}
\(\begin{array}{c c c c}
	\hline
	40 & 36 & 12 & 12 \\ \hline
	A  & B  & C  & C  \\
	B  & A  & D  & D' \\
	C  & C  & A  & A  \\
	D  & D  & B  & B  \\
	D' & D' & D' & D  \\ \hline
\end{array}\)
\end{table}
}
\tabExSscwNoRcwDodCM{tab:ex_sscw_no_rcw_dod_cm}{main}

\subsection{Resistant Condorcet Winner (RCW)}\label{sec:rcw}

We finally consider the existing notion of Resistant Condorcet Winner (RCW), which completes our hierarchy of Condorcet notions.

\begin{definition}\cite{durand2016condorcet}
A candidate~$c$ is a \emph{Resistant Condorcet Winner} (RCW) if, for every pair of other candidates $(d,e)$ (possibly identical),\footnote{The case $d=e$ only matters when $m=2$; for $m\geq 3$, it can be omitted since it is implied by the case $d\neq e$.}
\begin{equation}\label{eq:def_rcw}
w(P^{c \succ d \text{ and } c \succ e}) > \frac{w(P)}{2}.
\end{equation}
Intuitively, no coalition preferring $d$ to $c$ can make $c$ appear as defeated by another candidate~$e$.

Equivalently, in any Condorcet-consistent rule, $c$ is elected and the profile is immune to coalitional manipulation.
\end{definition}

Whenever a candidate is an RCW, Equation~\eqref{eq:def_sscw_kemeny_style} implies that it is also an SSCW.
This implication is strict: in Table~\ref{tab:ex_sscw_no_rcw_dod_cm}, candidate~$A$ is the SSCW, but not an RCW.

\begin{restatable}{theorem}{ThmThetaCriticalRCW}
\label{thm:theta_critical_rcw}
In the Perturbed Culture model, the critical concentration parameter for the existence of a Resistant Condorcet Winner is
\[
\theta_c(\RCW, m) = \frac{1}{4}.
\]
\end{restatable}

\begin{restatable}{corollary}{ThmThetaUIfProtectedByRCW}
\label{thm:theta_u_if_protected_by_rcw}
Let $f$ be a Condorcet-consistent voting rule.
Then 
\[
\theta_u(f, m) \le \frac{1}{4}.
\]
\end{restatable}

In Appendix~\ref{sec:appendix_black_family}, it is further shown that for Black, Slater ($m \geq 4$), and Copeland ($m \geq 5$), if $\theta < \tfrac{1}{4}$, then a manipulation exists with high probability.
Hence, these rules attain $\theta_c(f, m) = \tfrac{1}{4}$, the largest—and therefore the worst—critical concentration parameter among Condorcet-consistent rules.
This result will be incorporated into the main theorem (Theorem~\ref{thm:critical_thetas}).

For Slater with $m=3$ and Copeland with $m \in \{3,4\}$, the result depends on the tie-breaking rule.
When $m=3$, these rules are only required to elect the Condorcet winner whenever one exists; depending on the tie-breaking rule, they may coincide, for instance, with \emph{Condorcet-IRV} (IRV with a precondition to elect the Condorcet winner when one exists), yielding $\theta_c(f,3)=0$ \cite{durand2025irv}, or with the Black rule, yielding $\theta_c(f,3)=\tfrac{1}{4}$. In Appendix~\ref{sec:appendix_black_family}, we provide a detailed analysis of Copeland with $m=4$, showing that the critical concentration parameters may take values in the same interval, depending on the tie-breaking rule and on the value of the parameter~$\alpha$ of the rule.

\medskip

Now that we have established the chain of strict implications
\[
RCW \Rightarrow SSCW \Rightarrow PSCW \Rightarrow CW,
\]
we compare these notions with the Super Condorcet Winner (SCW), another strengthening of the Condorcet winner.
Table~\ref{tab:ex_rcw_no_scw} shows that a candidate may be an RCW (and thus also an SSCW and a PSCW) without being an SCW.
Conversely, in Table~\ref{tab:ex_scw_no_pscw_rules_not_CM}, candidate~$A$ is the SCW, but not a PSCW (and thus neither an SSCW nor an RCW).
This confirms that SCW is logically independent of RCW, SSCW, and PSCW.

\newcommand{\tabExRcwNoScw}[2]{% Argument: #1 for \label, #2 main/app
\begin{table}[ht]
\centering
\ifthenelse{\equal{#2}{main}}{%
  \caption{A profile where candidate~$A$ is the RCW but not an SCW. The example remains valid for any completion of the truncated rankings (see Appendix~\ref{sec:appendix_analysis_table_5}).}%
}{%
  \caption{A profile where candidate~$A$ is the RCW but not an SCW. The example remains valid for any completion of the truncated rankings.}%
}
\label{#1}
\(\begin{array}{c c c c c}
	\hline
	1      & 2      & 2      & 2      & 2      \\ \hline
	A      & B_1    & B_2    & B_3    & B_4    \\
	       & A      & A      & A      & A      \\
	\vdots & \vdots & \vdots & \vdots & \vdots \\ \hline
\end{array}\)
\end{table}%
}
\tabExRcwNoScw{tab:ex_rcw_no_scw}{main}

\subsection{Critical Concentration Parameters}\label{sec:main_theorem}

The analysis of Sections~\ref{sec:pscw}, \ref{sec:sscw}, and~\ref{sec:rcw} determines the critical concentration parameters for a large collection of voting rules, complementing those studied by \citet{durand2025irv}.
At this point, it remains to address Borda, Bucklin, and veto-based rules, namely Veto, Coombs, and Kim--Roush.

This is done on a case-by-case basis in Appendix~\ref{sec:appendix_proof_of_theoretical_results}.
For each rule, the argument follows the general strategy of \citet{durand2025irv}.
We show that there exists a threshold $\theta_c(f,m)$ such that any profile in a neighborhood of the expected profile~$\hat{P}$ is coalitionally manipulable below this value and immune to coalitional manipulation above it.
The result is then lifted to large finite electorates using the weak law of large numbers.
Some rules raise specific subtleties, such as Bucklin and Veto. In all cases, the convergence is exponentially fast, by the same argument as \citet{durand2025irv} (see Section~\ref{sec:known_results} and Appendix~\ref{sec:appendix_manipulation_lemmas}).

We are now ready to state our main theorem.

\begin{restatable}{theorem}{ThmCriticalThetas}
\label{thm:critical_thetas}
Each voting rule~$f$ defined in Section~\ref{sec:zoology} admits a critical concentration parameter $\theta_c(f,m)$, given in Table~\ref{tab:theo_results} for $m \geq 3$ ($m \geq 4$ for Slater and $m \geq 5$ for Copeland).
\end{restatable}

For $m \leq 2$, all voting rules considered here coincide with Plurality and are therefore immune to coalitional manipulation, so that $\theta_c(f,m)=0$.
As noted in Section~\ref{sec:rcw}, for Slater with $m=3$ and for Copeland with $m \in \{3,4\}$, the critical concentration parameter may range from $0$ to $\tfrac{1}{4}$ depending on the tie-breaking rule.

Thus, although pathological voting rules may fail to admit a critical concentration parameter, Theorem~\ref{thm:critical_thetas} shows that the phase transition identified by \citet{durand2025irv} extends to essentially all standard ordinal voting rules and identifies their critical thresholds.

\begin{table}[ht]
\caption{Critical concentration parameters $\theta_c(\cdot,m)$ for $m \geq 3$ and their limits $\theta_c(\cdot,\infty)$. When a Condorcet notion is present, the corresponding rules share the same critical concentration parameter.}
\label{tab:theo_results}
\begin{minipage}{\columnwidth}
\renewcommand{\thefootnote}{\emph{\alph{footnote}}}
\begin{center}
\renewcommand{\arraystretch}{1.3}
\begin{tabular}{l c c c}
	\toprule
	\textbf{Rule $f$}                                    & \textbf{Notion} & $\boldsymbol{\theta_c(\cdot, m)}$ & $\boldsymbol{\theta_c(\cdot, \infty)}$ \\ \midrule
	IRV \footnotemark[1]                                 & SCW             & $0$                           & $0$                                \\ 
	Maximin, Ranked Pairs, Schulze, Young                & PSCW            & $\frac{1}{7}$                 & $\frac{1}{7}$                      \\ 
	Plurality with Runoff \footnotemark[1]               &                 & $\frac{m-3}{5m-3}$            & $\frac{1}{5}$                      \\ 
	Baldwin, Nanson, Kemeny, Dodgson, Simplified Dodgson & SSCW            & $\frac{m-2}{4m-5}$            & $\frac{1}{4}$                      \\ 
	Black, Slater ($m \geq 4$), Copeland ($m \geq 5$)    & RCW             & $\frac{1}{4}$                 & $\frac{1}{4}$                      \\ 
	Plurality \footnotemark[1]                           &                 & $\frac{m-2}{3m-2}$            & $\frac{1}{3}$                      \\ 
	Coombs                                               &                 & $\frac{m-1}{3m-1}$            & $\frac{1}{3}$                      \\ 
	Bucklin                                              &                 & $\frac{m-2}{2m-2}$            & $\frac{1}{2}$                      \\ 
	Borda                                                &                 & $\frac{m-2}{m+1}$             & $1$                                \\ 
	Kim--Roush                                            &                 & $\frac{m-2}{m}$               & $1$                                \\ 
	Veto                                                 &                 & $1$                           & $1$                                \\ \bottomrule
\end{tabular}
\end{center}
\footnotetext[1]{The results on IRV, Plurality with Runoff, and Plurality are due to~\citet{durand2025irv}.}
\end{minipage}
\end{table}

Table~\ref{tab:theo_results} reveals the existence of three clusters of voting rules sharing the same critical concentration parameter $\theta_c(f,m)$, in addition to the \emph{IRV family}, consisting of IRV and its variants, already identified by \citet{durand2025irv}.
These clusters can be referred to as the \emph{Maximin family} (Maximin, Ranked Pairs, Schulze, and Young), the \emph{Baldwin family} (Baldwin, Nanson, Kemeny, Dodgson, and Simplified Dodgson), and the \emph{Black family} (Black, Slater, and Copeland).

That some rules are grouped together is not entirely surprising, as several of them share closely related mechanisms—for instance Maximin, Ranked Pairs, and Schulze; Baldwin and Nanson; or Slater and Copeland.
However, other aspects of the classification are less obvious a priori.
In particular, it is not immediate that Young should belong to the Maximin family, or that Dodgson and Kemeny should be separated from it and instead cluster with Baldwin and Nanson.
For Young, as discussed above, the connection with the Pair-Safe Condorcet Winner is non-trivial.
For Kemeny and Dodgson, a closer inspection reveals that the number of swaps in the Dodgson rule, or the Kemeny score itself, can be related to scores in the weighted majority matrix, and therefore to the Borda scores used in Baldwin and Nanson.
Finally, the position of the Black rule is particularly intriguing: it differs both from other Condorcet-consistent rules based on Borda elimination (Baldwin and Nanson) and from Borda itself.

Thus, the proposed classification is not immediately apparent from the axiomatic definitions of the rules.
Whether it translates into similar behavior with respect to coalitional manipulation in real-world datasets is an empirical question, which we investigate next.

\section{Numerical Results}\label{sec:numerical_results}

The goal of this section is to test whether our simplified theoretical model retains explanatory power when confronted with real preference data.
All code used in our experiments will be made available as a GitHub repository once this paper is published.
We use two datasets, both introduced by \citet{durand2023coalitional}.
\begin{itemize}
\item The \emph{Netflix dataset} contains $11{,}215$ profiles obtained by perturbing $2{,}243$ base profiles, spanning $m=3$ to $m=11$ candidates and $1{,}000$ to $91{,}880$ voters.
Preferences are cardinal and complete, and uniform random noise is used to break ties between equal ratings.
Its size and coverage make it particularly suitable for detailed analyses that depend on~$m$.
\item The \emph{FairVote dataset} contains $10{,}044$ profiles obtained by perturbing $162$ base profiles, spanning $m=3$ to $m=11$ candidates and $1{,}560$ to $299{,}107$ voters.
Preferences are ordinal and generally truncated, and uniform random noise is used to complete the truncated rankings.
Although less rich, especially for larger values of~$m$, this dataset serves as a useful complement, as it describes real-world political elections.
\end{itemize}
Additional results on the more heterogeneous PrefLib dataset \cite{mattei2013preflib} are provided in the code repository, confirming the robustness of our findings.

Our numerical evaluations rely on \mbox{\textsc{SVVAMP}}, a Python package dedicated
to studying the manipulability of voting rules \cite{durand2016svvamp,svvamp}.
For each rule, the implemented manipulation algorithm classifies a profile as
CM, non-CM, or \emph{undecided}, meaning that the chosen algorithm
cannot certify either outcome for that profile (algorithmic uncertainty).
Compared to the version used by \citet{durand2023coalitional}, the present work relies on a substantially improved version of the package.
Namely, we implemented the notions of PSCW and SSCW, added five voting rules (Kemeny, Slater, Young, Dodgson, and Simplified Dodgson), and significantly improved the manipulation algorithms for five others (Ranked Pairs, Baldwin, Nanson, Copeland, and Kim--Roush).

We proceed in three steps.
First, we examine the overall CM rates.
Second, we compare how these rates vary with~$m$ against the theoretical predictions given by~$\theta_c(f,m)$.
Finally, we focus on a fixed value of~$m$ to evaluate how well the theory predicts which rules are more resilient to coalitional manipulation.

\subsection{Overall CM Rates}\label{sec:overall_cm_rates}

\renewcommand{\axisWidth}{14cm}
\renewcommand{\axisHeight}{8.1cm}
\begin{figure}[t]
\begin{subfigure}[t]{\linewidth}
  \centering
  \input{netflix_selected_rules_cm_rate_bar_plot_n_c_ge=5.tex}
  \caption{Netflix dataset.}
  \label{fig:netflix_cm_rate_bar_plot}
\end{subfigure}

\vspace{4mm}
\begin{subfigure}[t]{\linewidth}
  \centering
  \input{fairvote_selected_rules_cm_rate_bar_plot_n_c_ge=5.tex}
  \caption{FairVote dataset.}
  \label{fig:fairvote_cm_rate_bar_plot}
\end{subfigure}

\caption{CM rates for $m \geq 5$. Solid bars show the fraction of coalitionally manipulable profiles, with thin vertical lines indicating algorithmic uncertainty. Colors group rules by critical concentration parameter $\theta_c(f,m)$ (Table~\ref{tab:theo_results}). Dashed horizontal lines indicate, from top to bottom, the fraction of profiles without a Resistant, Set-Safe, Pair-Safe, or Super Condorcet winner.}
\label{fig:cm_rate_bar_plot}
\end{figure}
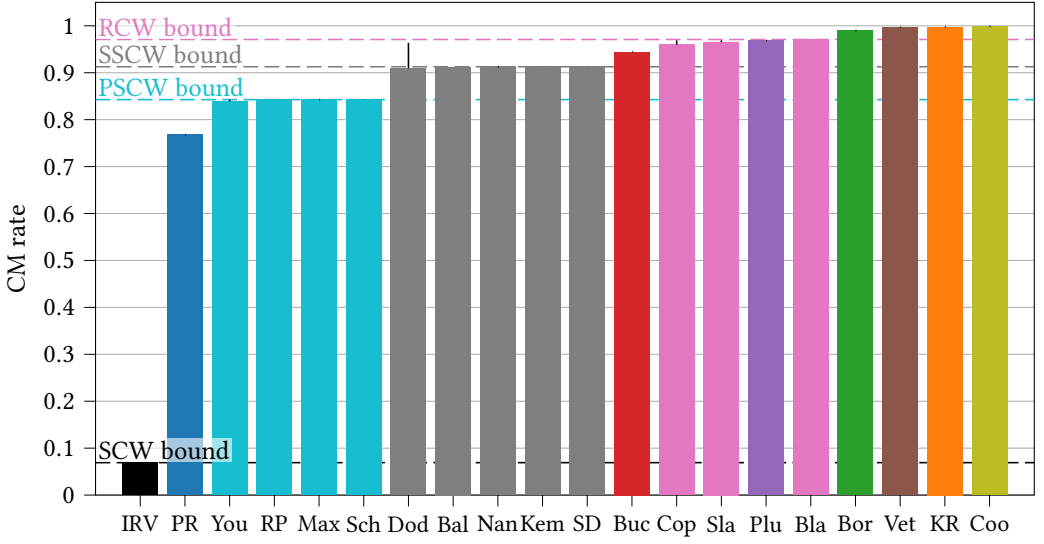
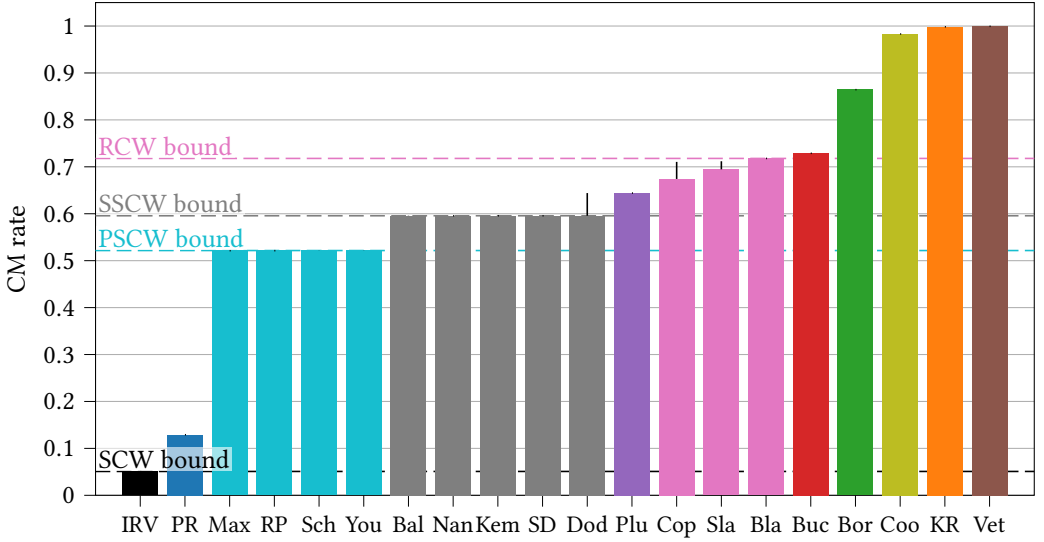
\renewcommand{\axisWidth}{\axisWidthMemory}
\renewcommand{\axisHeight}{\axisHeightMemory}

Figure~\ref{fig:cm_rate_bar_plot} reports CM rates for the two datasets, restricted to profiles with at least five candidates, ensuring that our theoretical results apply to all voting rules. Additional figures for all values of~$m$ are provided in Appendix~\ref{sec:appendix_additional_figures}; conclusions are similar, except for Slater and Copeland as expected.

Compared to earlier work \cite{durand2023coalitional}, our computations reduce uncertainty to very low levels for almost all rules, including NP-hard ones such as Kemeny, Slater, Dodgson, and Young \cite{bartholdi1989voting,rothe2003exact,conitzer2006computing}. Dodgson remains the least precise, with about 5\% undecided cases in both datasets, followed by Copeland in the FairVote dataset (4\%).

As already noted by~\citet{durand2023coalitional}, the absolute levels differ markedly across the two datasets, with profiles from the FairVote dataset being on average substantially less susceptible to coalitional manipulation. However, several common conclusions emerge.

First, the theoretical families identified in Table~\ref{tab:theo_results} also appear empirically: rules within a given family display very similar CM rates.
Within the Maximin and Baldwin families, differences remain below 0.5\% in both datasets.
The only deviation arises from the upper bound of the uncertainty interval for Dodgson,
while its lower bound remains consistent with the other rules in the Baldwin family.
Within the Black family, the difference is at most 1\% in the Netflix dataset, and ranges between 1\% and 4\% in the FairVote dataset, depending on uncertainty.

Second, just as the notion of SCW explains most non-CM cases for IRV \cite{durand2025irv}, the notions of PSCW, SSCW, and RCW provide tight explanations for specific families. For example, the CM rates of the Maximin family closely align with the absence of a PSCW. In the particular case of Dodgson, we recall that despite sharing the same critical concentration parameter as the rest of the Baldwin family, it is theoretically unaffected by the SSCW bound (Section~\ref{sec:sscw}). Unfortunately, algorithmic uncertainty prevents us from determining whether counterexamples in which Dodgson is susceptible to coalitional manipulation despite the existence of an SSCW are frequent.

Third, an unexpected pattern emerges: in both datasets, Kim--Roush and Veto exhibit very similar, near-maximal CM rates. We return to this observation below.

\subsection{Variation with the Number of Candidates}\label{sec:simu_variation_m}

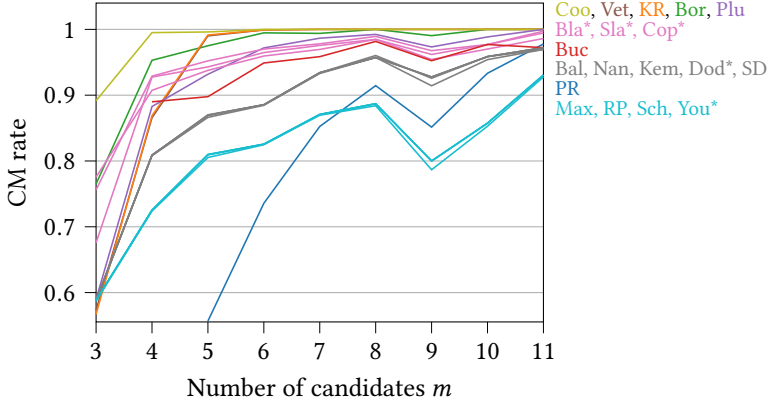
\begin{figure}[t]
\hspace{-1.85mm}\begin{subfigure}[t]{\linewidth}
  \centering
  \input{plot_theoretical_results.tex}
  \caption{Critical concentration parameter $\theta_c(f,m)$ as a function of the number of candidates $m$ (Theorem~\ref{thm:critical_thetas}). ${}^\dagger$~For Slater, valid only if $m \geq 4$; for Copeland, only if $m \geq 5$.}
  \label{plot_theoretical_results}
\end{subfigure}

\bigskip\bigskip
\begin{subfigure}[t]{\linewidth}
  \centering
  \input{netflix_selected_rules_nb_candidates_rate_line_plot.tex}
  \caption{CM rate in the Netflix dataset as a function of the number of candidates $m$. IRV remains below $12\%$ for all values of $m$. Curves show the fraction of profiles certified as coalitionally manipulable. ${}^\ast$~Rules marked with a star have algorithmic uncertainty above $1\%$ for some values of~$m$.}
  \label{netflix_selected_rules_nb_candidates_rate_line_plot}
\end{subfigure}

\caption{Coalitional manipulability as a function of the number of candidates~$m$: theoretical critical concentration parameters (top) and empirical CM rates in the Netflix dataset (bottom).}
\label{fig:cm_as_function_of_m}
\end{figure}

We now examine how coalitional manipulability varies with the number of candidates~$m$, combining theoretical predictions (Fig.~\ref{plot_theoretical_results}) with empirical results (Fig.~\ref{netflix_selected_rules_nb_candidates_rate_line_plot}). We focus on the Netflix dataset, which contains a relatively large number of base profiles even for larger~$m$ (e.g., 72 base profiles for $m=11$). We also include profiles with three or four candidates, noting that tie-breaking then significantly affects Slater and/or Copeland. IRV is not visible in Fig.~\ref{netflix_selected_rules_nb_candidates_rate_line_plot}, as its CM rates remain below~12\%.

Both Figures exhibit common patterns.
Rules within each of the Maximin, Baldwin, and Black families display very similar CM rates across all values of~$m$.
Plurality with Runoff coincides with IRV at $m=3$ (by definition), and then crosses the Maximin family as predicted, with a slightly shifted crossing point ($m=9$ in theory and $m \in [7,8]$ in practice).
The Maximin family, the Baldwin family, and Plurality coincide at $m=3$ but stratify thereafter in that order, with Plurality eventually performing worse than the Black family.
The Black family consistently performs worse than the Baldwin family, while Borda starts at the same level as the Black family but degrades further as $m$ increases.
Overall, the theory provides a faithful qualitative description of the behavior of IRV, the Maximin family, Plurality with Runoff, the Baldwin family, the Black family, Plurality, and Borda.

Some discrepancies nevertheless emerge.
In Figure~\ref{netflix_selected_rules_nb_candidates_rate_line_plot}, Kim--Roush and Veto are almost indistinguishable empirically, a behavior reminiscent of their identical asymptotic CM rate under Impartial Culture~\cite{kim1996manipulability}.
For larger values of~$m$, both rules behave as expected and are more vulnerable than Borda.
For smaller~$m$, however, they appear less manipulable than predicted when compared to other rules.
Coombs is consistently more manipulable than predicted,
whereas Bucklin performs worse than expected at $m=3$ but better for larger~$m$.

For veto-based rules, these deviations likely stem from peculiarities of the Perturbed Culture at the bottom of the rankings: candidate~1 is not favored, receiving as many bottom votes as candidates $\{2,\ldots,m-1\}$, while only candidate~$m$ is disadvantaged. For Bucklin, a similar effect occurs beyond the first rank: at rank~2 the model favors candidate~2, at rank~3 candidate~3, and so on. Alternative models such as Mallows may be better suited to study these rules, though at higher technical cost.

\subsection{Comparison of Voting Rules for Fixed $m$}\label{sec:expe_cm_rate_vs_theta_c}

\begin{figure}[t]
\hspace{-5mm}\begin{subfigure}[t]{0.52\linewidth}
  \centering
  \input{netflix_selected_rules_scatter_rank_theta_c_rank_cm_rate_n_c=5_revised.tex}
  \captionsetup{margin={-1.3cm,0cm}}
  \caption{Netflix dataset.\hspace{-1cm}~}
  \label{fig:netflix_selected_rules_scatter_rank_theta_c_rank_cm_rate}
\end{subfigure}%
\begin{subfigure}[t]{0.46\linewidth}
  \centering
  \input{fairvote_selected_rules_scatter_rank_theta_c_rank_cm_rate_n_c=5_revised.tex}
  \captionsetup{margin={-1.cm,0cm}}
  \caption{FairVote dataset.}
  \label{fig:fairvote_selected_rules_scatter_rank_theta_c_rank_cm_rate}
\end{subfigure}

\caption{Agreement between theoretical and empirical rankings of voting rules by coalitional manipulability for $m=5$. The $x$-axis ranks rules according to the theoretical critical concentration parameter $\theta_c(f,5)$, while the $y$-axis ranks them according to their empirical CM rate. Vertical bars indicate uncertainty.}
\label{fig:scatter_rank_theta_c_rank_cm_rate}
\end{figure}
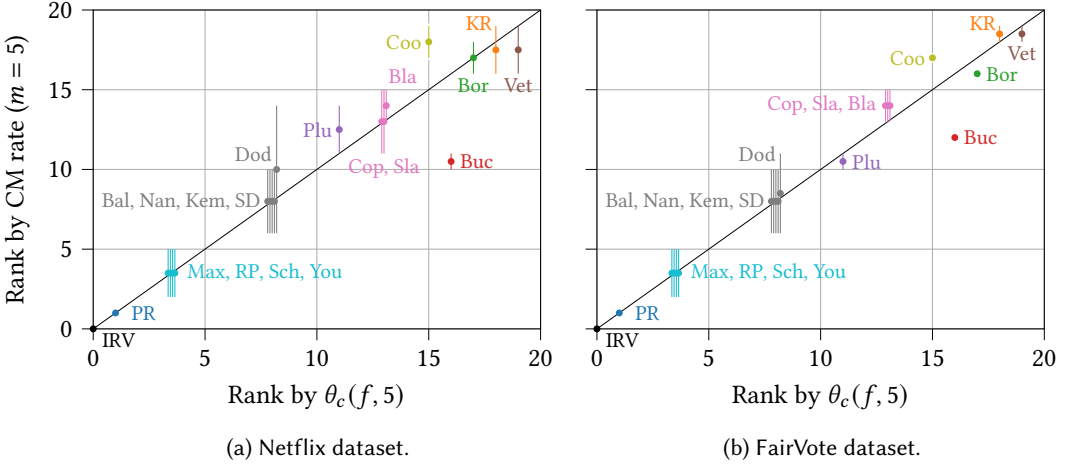

For a fixed number of candidates~$m$, the critical concentration parameter~$\theta_c(f,m)$ induces a theoretical ordering of voting rules.
A natural question is whether this theoretical ordering has predictive power for the ordering induced by their empirical CM rates.
Figure~\ref{fig:netflix_selected_rules_scatter_rank_theta_c_rank_cm_rate} confronts these two rankings for the Netflix dataset, while Figure~\ref{fig:fairvote_selected_rules_scatter_rank_theta_c_rank_cm_rate} does so for the FairVote dataset.
We fix $m=5$, the smallest value for which all our theoretical results apply, including Slater and Copeland.
Analogous plots for larger values of~$m$ are provided in the code repository.

By convention, the minimal rank is set to~0.
For the theoretical ranking based on $\theta_c(f,5)$, ties correspond to exact equalities and are assigned the average rank (e.g., the four rules tied for ranks~2 to~5 each receive rank~3.5).
For empirical CM rates, uncertainty arises from two sources: algorithmic limitations and differences below~1\%, which are treated as non-significant; both effects are reflected in the error bars.
Finally, for readability, rules sharing the same value of~$\theta_c(f,5)$ are displayed with a slight horizontal offset.

In both figures, the agreement is strong: the ranking predicted by the critical concentration parameter closely matches the empirical ranking.
The main exceptions are Bucklin and Coombs, in line with the discrepancies already identified in the previous subsection.
For Veto and Kim--Roush, as already observed for larger values of $m$, the theory correctly predicts their position in the ranking at $m=5$, but not their near-identical empirical CM rates.
Finally, the ordinal conclusions are remarkably similar across the two datasets, despite large differences in absolute CM rates (Section~\ref{sec:overall_cm_rates}).

Overall, this experiment shows that the model has genuine predictive power over the relative vulnerability of voting rules to coalitional manipulation, despite the absence of any parameter fitting: the only input provided to the model is the number of candidates~$m$.
The key intuition is that elections typically feature a candidate who is globally stronger than the others in the preferences of the voters, a phenomenon captured in the model by the concentration parameter.
When this advantage is sufficiently large, a voting rule becomes immune to coalitional manipulation; what counts as “sufficiently large” depends on the rule and is captured in the theoretical model by its critical concentration parameter.
This explains why, for a fixed~$m$, we recover essentially the same ordering of voting rules across datasets, even when their absolute CM rates differ substantially.
It is very plausible that other probabilistic models built on the same basic ingredients, such as the Mallows model, would lead to similar conclusions.
The strength of the Perturbed Culture lies in its ability to deliver these insights while remaining particularly tractable mathematically.

\section{Future work}\label{sec:future_work}

A first direction for future work is to test the robustness of our results under alternative preference models, such as Mallows, Bradley--Terry, Plackett--Luce, spatial models, or mixtures thereof, and to examine whether some of these models better explain the behavior of Bucklin, Veto, Coombs, and Kim--Roush.
Another natural step is to design probabilistic models that are also suitable for non-ordinal rules, such as Range Voting.
It would also be valuable to study the critical regime $\theta=\theta_c(f,m)$ in more detail, in particular to determine whether the bounds derived from our strengthened Condorcet notions remain tight at the phase transition.
Beyond coalitional manipulation, phase transitions in parametric cultures could further shed light on other voting paradoxes, such as violations of monotonicity, participation, or independence of irrelevant alternatives.

% In the interest of anonymization, please do not include acknowledgements in your submission.
%
%\begin{acks}
%
%\end{acks}

% Bibliography
\bibliographystyle{ACM-Reference-Format}
\bibliography{sample}

% Appendix

\newpage
\appendix
\etocdepthtag.toc{app} % tag "appendices"

\section*{\Large Technical Appendix}

In this technical appendix, we elaborate on some points of the paper that are omitted from the main text for the sake of conciseness and clarity. We also include two additional Condorcet-consistent voting rules, defined as follows.

\defrule{Split Cycle}{SC}{With the same notation as for the Schulze rule (Section~\ref{sec:other_condorcet_rules}), elect a candidate~$c$ such that $\Strength(c,d,P) \geq W(d,c,P)$ for all $d$ \cite{holliday2020split}.}

\defrule{Viennot}{Vie}{At each round, select the two candidates with the lowest plurality scores and eliminate the one losing their pairwise contest \cite{durand2015towards,durand2023coalitional}.}

\medskip
We will see that these two rules naturally belong to the Maximin family, both theoretically and experimentally.
Split Cycle relies on a mechanism very close to that of the Schulze rule; we chose
to exclude it from the main paper to avoid redundancy, but include it here for
the sake of completeness.
The Viennot rule is less studied, but provides a particularly interesting case in our framework: although the existence of a PSCW is neither necessary nor sufficient for immunity to coalitional manipulation, the rule shares the same critical concentration parameter as the Maximin family.
The fact that it exhibits very similar empirical CM rates (Section~\ref{sec:appendix_additional_figures}) shows that this empirical
clustering of voting rules cannot be attributed solely to the PSCW notion, but also reflects the predictive power of the critical concentration parameter itself.

\medskip
Before proceeding, we briefly recall the main notation used throughout the appendices.
For any (discrete or continuous) profile~$P$, $w(P)$ denotes its total weight.  
We also write $w(P^{\psi})$ for the weight of the subpopulation of voters whose rankings satisfy a condition~$\psi$; typical examples include
$w(P^{c \succ d})$ and $w(P^{c \succ d \text{ and } c \succ e})$.
The notation \(s_f\) denotes a notion of score for rule~$f$, whereas $p_f$ denotes a notion of penalty.
For any pair of candidates $(c,d)$, $W(c,d,P) = w(P^{c \succ d})$ denotes the total weight of voters preferring~$c$ to~$d$; the matrix $W(P)$ with entries $W(c,d,P)$ is the \emph{weighted majority matrix} of~$P$.

For reference, we recall where the strengthened notions of Condorcet winner used in the paper are defined: the SCW in Equation~\eqref{eq:def_scw}; the PSCW in Equation~\eqref{eq:def_pscw}; the SSCW in Equations \eqref{eq:def_sscw_s}, \eqref{eq:def_sscw_t}, \eqref{eq:def_sscw_kemeny_style}, and \eqref{eq:def_sscw_computable}; and the RCW in Equation~\eqref{eq:def_rcw}.

Finally, we introduce the notation $b(\epsilon)$ (read “bounded by epsilon”) to denote a real number whose absolute value is at most~$\epsilon$.

\medskip
Appendix~\ref{sec:appendix_analysis_of_the_examples} provides a detailed analysis
of the examples used throughout the paper.
Appendix~\ref{sec:appendix_theo_preliminaries} presents the theoretical
preliminaries, while Appendix~\ref{sec:appendix_proof_of_theoretical_results} contains the
proofs of the paper’s main results.
Finally, Appendix~\ref{sec:appendix_additional_figures} reports additional
figures showing global CM rates, including the cases $m \in \{3,4\}$ and the two
additional voting rules.

\begingroup
  \etocsettagdepth{main}{none} % hide the main
  \etocsettagdepth{app}{all}   % display the appendices
  \etocsettocstyle{\section*{Contents of the Technical Appendix}}{}
  \tableofcontents
\endgroup

\section{Examples}\label{sec:appendix_analysis_of_the_examples}

\newcounter{tableBackup}
\newcommand{\tableForceNumber}[2]{%
  \setcounter{tableBackup}{\value{table}}%
  \setcounter{table}{\getrefnumber{#1}}%
  \addtocounter{table}{-1}%
  #2%
  \setcounter{table}{\value{tableBackup}}%
}

This section provides a detailed verification of all examples discussed in the
main body of the paper.
Each table is recalled with an extended caption summarizing its relevant
properties, which are then checked explicitly.
This analysis also serves to illustrate the strengthened Condorcet notions discussed in the paper.

\subsection{Analysis of Table~\ref{tab:ex_scw_no_pscw_rules_not_CM}}\label{sec:appendix_analysis_table_1}

\tableForceNumber{tab:ex_scw_no_pscw_rules_not_CM}{%
\tabExScwNoPscwRulesNotCM{tab:ex_scw_no_pscw_rules_not_CM_remind}{app}%
}

The profile in Table~\ref{tab:ex_scw_no_pscw_rules_not_CM_remind} consists of $n(P) = w(P) = 11$ voters.

Candidate $A$ is the SCW, as defined in Equation~\eqref{eq:def_scw}, since
\[
\begin{aligned}
s_\Plu(A, P_{\{A,B,C\}}) &= 4 > 11/3,\\
s_\Plu(A, P_{\{A,B\}})   &= 6 > 11/2,\\
s_\Plu(A, P_{\{A,C\}}) &= 9 > 11/2.
\end{aligned}
\]

Candidate $A$ is not a PSCW, as defined in Equation~\eqref{eq:def_pscw}, since
\[
w(P^{A \succ B \text{ and } A \succ C}) = 4
\quad\text{and}\quad
w(P^{B \succ A}) = 5,
\]
hence \( w(P^{A \succ B \text{ and } A \succ C}) \le w(P^{B \succ A}) \).
Intuitively, if all voters preferring $B$ to $A$ demote $A$ below $C$ in their ranking, only the voters in $P^{A \succ B \text{ and } A \succ C}$ still support $A$ in the pairwise comparison against $C$, which then becomes a defeat more severe than that of $B$ against $A$.

If a candidate $A'$ is added just below $A$ in every ranking, it is straightforward to verify that $A$ remains an SCW but is still not a PSCW.

Since all the rules mentioned here are Condorcet-consistent, $A$ is elected. We now show that the profile is immune to coalitional manipulation for all these rules.

Consider a manipulation attempt in favor of $B$. Only the voters in the first column can participate in such a manipulation. However, they cannot prevent $B$ from being a Condorcet loser, i.e., a candidate who suffers pairwise defeats against all others. Since $m=3$ and $n$ is odd, this implies that another candidate must be the Condorcet winner and is therefore elected. For Copeland and Slater, the example includes the additional candidate $A'$; it then suffices to note that these rules can never elect a Condorcet loser, which prevents $B$ from winning.

We now consider a manipulation attempt in favor of $C$, which may involve the voters in the third column. Table~\ref{tab:ex_scw_no_pscw_rules_not_CM_manipulated_wmm} reports bounds on the entries of the weighted majority matrix of a potential target profile~$Q$.
When displaying a majority matrix, we always omit the diagonal coefficients for legibility.
\begin{table}[ht]
\centering
\caption{Bounds on the weighted majority matrix of the target profile $Q$ in a manipulation attempt for $C$ from Table~\ref{tab:ex_scw_no_pscw_rules_not_CM_remind}.}
\(\begin{array}{c|c|c|c|}
	  & A    & B    & C     \\ \hline
	A &      & [4, 6] & [9, 11] \\ \hline
	B & [5,7] &      & [5,7]  \\ \hline
	C & [0,2] & [4,6] &       \\ \hline
\end{array}\)
\label{tab:ex_scw_no_pscw_rules_not_CM_manipulated_wmm}
\end{table}%

We now show that, for each rule listed in the caption of Table~\ref{tab:ex_scw_no_pscw_rules_not_CM_remind}, this manipulation attempt fails, as candidate $C$ cannot be elected in $Q$.

\paragraph{Max, RP, Sch, SC}
$\min\{W(A,B,Q),\, W(A,C,Q)\} > W(C,A,Q)$. Since all four rules are Maximin-like, $C$ cannot win in $Q$.

\paragraph{Vie}
The manipulators cannot prevent $C$ from being selected for the first-round duel, so $C$ must face $A$ and then $B$, or vice versa, in pairwise comparisons. However, $C$ necessarily loses its pairwise comparison against $A$, which prevents it from being elected.

\paragraph{You}
We bound the penalties after manipulation. For $A$, we can keep at least the voters of the second column and three others, that is, 7 voters in total; hence $p_\You(A, Q) \le 11 - 7 = 4$. For $C$, to win against $A$, at most the two voters of the last column and one additional voter can stay, yielding $p_\You(C, Q) \ge 11 - 3 = 8$. Therefore, $p_\You(C, Q) > p_\You(A, Q)$.

\paragraph{Bal, Nan}
We have $s_\Bor(A, Q) \ge 13$, $s_\Bor(B, Q) \ge 10$, and finally $s_\Bor(C, Q) \le 8$. Hence, $C$ is eliminated in the first round.

\paragraph{Dod, SD}
We bound the penalties after manipulation. Candidate $C$ needs at least four swaps to win against $A$, hence $p_f(C, Q) \ge 4$. Candidate $A$ already beats $C$, and two swaps among the voters in the left column are sufficient for $A$ to win against $B$, so $p_f(A, Q) \le 2$. Therefore, $p_f(C, Q) > p_f(A, Q)$.

\paragraph{Kem}
If $C$ were the winner, moving it to the bottom of the ranking would change the penalty by an amount 
\(
W(C,A,Q) + W(C,B,Q) - W(A,C,Q) - W(B,C,Q) \le 2 + 6 - 9 - 5 = -6 < 0.
\)
There would thus exist a strictly better ranking, contradicting the optimality
of the winning ranking.

\paragraph{Cop, Sla}
As with all Condorcet-consistent voting rules, the existence of an RCW implies that Copeland and Slater are immune to coalitional manipulation. This example will show that the converse is false: these rules are immune to CM, even though the winner~$A$ is not an RCW, and in fact not even a PSCW.

Recall that we consider the version of the example including the additional candidate $A'$. Candidate $C$ has at most one victory (against $B$), and thus cannot win under Copeland, since another candidate has at least two victories.  
For Slater with $m=4$, a candidate with at most one victory can likewise never win, as we now show.  
If a candidate has three victories, it is the Condorcet winner and therefore wins. Otherwise, the vector of numbers of victories (Copeland scores) is either $(2,2,2,0)$ or $(2,2,1,1)$.  
The first case is straightforward, as there is then a Condorcet loser. Let us therefore examine the second one.  
Label the two candidates with two victories $a$ and $b$ such that $a$ beats $b$. Consequently, the two victories of $b$ are against the other two candidates, denoted $c$ and~$d$.  
Without loss of generality, assume that $c$ beats $d$; then the only victory of $d$ must be against $a$.  
It is easy to verify that the order $(a \succ b \succ c \succ d)$ is the unique ranking with minimal Slater penalty, equal to~1.
Hence, none of the candidates with a single victory can be the Slater winner.
Finally, it remains to check that no manipulation in favor of $A'$ is possible, which is trivial since no voter prefers $A'$ to $A$.

\subsection{Analysis of Table~\ref{tab:ex_you_is_not_maximin_like}}\label{sec:appendix_analysis_table_2}

\tableForceNumber{tab:ex_you_is_not_maximin_like}{%
\tabExYouIsNotMaximinLike{tab:ex_you_is_not_maximin_like_remind}{app}%
}

\begin{table}[ht]
\centering
\caption{Weighted majority matrix of the profile in Table~\ref{tab:ex_you_is_not_maximin_like_remind}.}
\(\begin{array}{c|c|c|c|c|c|}
	    & A  & B   & C_1 & C_2 & C_3 \\ \hline
	A   &    & 102 & 78  & 78  & 78  \\ \hline
	B   & 72 &     & 126 & 126 & 126 \\ \hline
	C_1 & 96 & 48  &     & 87  & 87  \\ \hline
	C_2 & 96 & 48  & 87  &     & 87  \\ \hline
	C_3 & 96 & 48  & 87  & 87  &     \\ \hline
\end{array}\)
\label{tab:ex_you_is_not_maximin_like_wmm}
\end{table}%

We now turn to the example of Table~\ref{tab:ex_you_is_not_maximin_like_remind}, whose weighted majority matrix is given in Table~\ref{tab:ex_you_is_not_maximin_like_wmm}. We have $\min_e W(A, e, P) = 78$ and $W(B, A, P) = 72$, hence $\min_e W(A, e, P) > W(B, A, P)$.

\paragraph{You}
For candidate $A$, it is most effective to remove voters from the second column. 
To make $A$ win against each candidate $C_i$, more than $96 - 78 = 18$ points
must be removed in each of the three corresponding pairwise comparisons, that
is, more than $54$ points in total.  
However, each ballot removed deprives the $C_i$’s of two points while also
costing $A$ one point against a $C_i$.
Hence, more than 54 ballots must be removed, so $p_\You(A, P) > 54$.
For candidate $B$, it suffices to remove the voters of the first column and one additional arbitrary voter, giving $p_\You(B, P) = 31$.  
For each candidate $C_i$, it is necessary to win against $B$, so $p_\You(C_i, P) > 126 - 48 = 78$.  
Therefore, $B$ is elected.

\paragraph{Vie}
The first-round plurality scores are $\{A: 30, B: 72, C_1: 24, C_2: 24, C_3: 24\}$. One candidate $C_k$ loses to some $C_j$.  
At the second round, the plurality scores are $\{A: 30, B: 72, C_i: 36, C_j: 36\}$.  
Candidates $A$ and, say, $C_j$ are selected, and $A$ is eliminated.  
Since $B$ is the Condorcet winner in the remaining profile, it is elected.

\subsection{Analysis of Table~\ref{tab:ex_pscw_no_sscw}}\label{sec:appendix_analysis_table_3}

\tableForceNumber{tab:ex_pscw_no_sscw}{%
\tabExPscwNoSscw{tab:ex_pscw_no_sscw_remind}{app}%
}

In the example of Table~\ref{tab:ex_pscw_no_sscw_remind}, we have
\[
\begin{aligned}
w(P^{A \succ B \text{ and } A \succ C}) &= 8, \qquad\qquad w(P^{B \succ A}) = 6,\\
w(P^{A \succ B \text{ and } A \succ D}) &= 7, \qquad\qquad w(P^{C \succ A}) = 5,\\
w(P^{A \succ C \text{ and } A \succ D}) &= 8, \qquad\qquad w(P^{D \succ A}) = 6,\\
\end{aligned}
\]
hence $w(P^{A \succ B \text{ and } A \succ C}) > w(P^{B \succ A})$, and similarly for any pair of candidates distinct from~$A$.  
Therefore, $A$ is the PSCW, as defined in Equation~\eqref{eq:def_pscw}.

On the other hand, evaluating Equation~\eqref{eq:def_sscw_computable} for the candidate of interest $c = A$ and the opponent $d = B$, we obtain
\[
\begin{aligned}
& w(P^{A \succ B}) - \tfrac{19}{2} + w(P^{A \succ B \text{ and } A \succ C}) - \tfrac{19}{2} + w(P^{A \succ B \text{ and } A \succ D}) - \tfrac{19}{2} \\
&= 13 - \tfrac{19}{2} + 8 - \tfrac{19}{2} + 7 - \tfrac{19}{2} \\
&= - \tfrac{1}{2},
\end{aligned}
\]
hence $A$ is not an SSCW.
Intuitively, in a manipulation attempt in favor of $B$, the manipulators can make
the cumulative defeats of $A$ against $C$ and $D$ outweigh the defeat that $B$
suffers against $A$.

In Viennot, candidate~$A$ wins since it is the Condorcet winner.  
If the voters in the second column move $A$ to the bottom of their ranking, the first round still selects $A$ and $C$, but $A$ is then eliminated.  
Candidate~$B$ becomes the Condorcet winner in the remaining profile and is therefore elected.

\subsection{Analysis of Table~\ref{tab:ex_sscw_no_rcw_dod_cm}}\label{sec:appendix_analysis_table_4}

\tableForceNumber{tab:ex_sscw_no_rcw_dod_cm}{%
\tabExSscwNoRcwDodCM{tab:ex_sscw_no_rcw_dod_cm_remind}{app}%
}

We now examine the example of Table~\ref{tab:ex_sscw_no_rcw_dod_cm_remind}. Evaluating Equation~\eqref{eq:def_sscw_computable} for the candidate of interest $c = A$ and the opponent $d = B$, we obtain
\[
\begin{aligned}
& w(P^{A \succ B}) - 50 + \sum_e \min(0,\, w(P^{A \succ B \text{ and } A \succ e}) - 50) \\
&= (64 - 50) + \min(0, 40 - 50) + \min(0, 52 - 50) + \min(0, 52 - 50) \\
&= 4 > 0,
\end{aligned}
\]
which satisfies the SSCW condition. Similarly, the same condition can be verified for opponents $C$ and $D$.  
Therefore, $A$ is an SSCW.

However, we have
\[
w(P^{A \succ B \text{ and } A \succ C}) = 40 \le 50,
\]
hence $A$ is not an RCW.

Let us show that Dodgson is susceptible to coalitional manipulation in favor of $B$.  
Consider a target profile $Q$ where all voters in the second column simply move $A$ to the bottom of their ranking.  
Candidate $B$ already beats $C$, $D$, and $D'$, and overturning the defeat of $B$ against $A$ can be achieved with 15 swaps in the first column. Hence $p_\Dod(B, Q) = 15$.  
Candidate $A$ has only 40 points against $C$, thus requiring at least $51 - 40 = 11$ useful swaps. The most effective operation is to perform them in one of the two last columns, but each useful swap must first be preceded by a useless swap between $A$ and $D$ or $D'$, so $p_\Dod(A, Q) \ge 22$.  
Candidate $C$ has only 24 points against $B$, hence $p_\Dod(C, Q) \ge 27$.  
Candidate $D$ has only 12 points against $B$, hence $p_\Dod(D, Q) \ge 39$, and similarly for $D'$.  
Therefore, $B$ is elected.

\subsection{Analysis of Table~\ref{tab:ex_rcw_no_scw}}\label{sec:appendix_analysis_table_5}

\tableForceNumber{tab:ex_rcw_no_scw}{%
\tabExRcwNoScw{tab:ex_rcw_no_scw_remind}{app}%
}

In the example of Table~\ref{tab:ex_rcw_no_scw_remind}, for any $i\neq j$ we have
$w(P^{A \succ B_i \text{ and } A \succ B_j}) = 5 > 9/2$; hence candidate~$A$ is
the RCW, as defined in Equation~\eqref{eq:def_rcw}. 
On the other hand, $s_\Plu(A, P) = 1 \le 9/5$, so candidate~$A$ violates the SCW condition
defined in Equation~\eqref{eq:def_scw} for the set $K = \mathcal{C}(P)$.

\section{Theoretical Preliminaries}\label{sec:appendix_theo_preliminaries}

Section~\ref{sec:appendix_manipulation_lemmas} states a collection of lemmas that are used repeatedly to establish the asymptotic behavior of the different voting rules.
Section~\ref{sec:analysis_sincere_profile} analyzes the expected profile.
Section~\ref{sec:analysis_sincere_voters} then restricts this profile to voters who prefer candidate~1 to candidate~2 and therefore remain sincere under a manipulation in favor of~2 against~1.
Finally, Section~\ref{sec:analysis_um_2} studies a simple manipulation strategy by voters preferring~2 to~1, which suffices to establish the sub-critical regime for a significant number of voting rules. We assume throughout that $m \ge 3$, as the case $m=2$ is trivial for all rules in our analysis.

\subsection{Manipulation Lemmas}\label{sec:appendix_manipulation_lemmas}

The proofs of \citet{durand2025irv} rely on three lemmas.
We restate them here with minor rephrasing, in order to facilitate their later generalization.
Recall that $\hat{P}$ denotes the \emph{expected normalized profile} (or simply the \emph{expected profile}), in which each ranking has weight $\tfrac{1-\theta}{m!}$, except for the reference ranking $(1 \succ \cdots \succ m)$, which has weight $\theta + \tfrac{1-\theta}{m!}$.

\begin{lemma}[Non-CM \cite{durand2025irv}]\label{lem_not_CM_in_neighborhood}
	Assume there exists a neighborhood~$\mathcal{N}$ of the expected normalized profile~$\hat{P}$ such that, for any profile~$P$, if $P$ lies in~$\mathcal{N}$, then the homogeneous rule~$f$ is non-CM.
	
	Then \( \lim_{n \to \infty} \rho(f, m, n, \theta) = 0\).
\end{lemma}

Before stating the second lemma, we recall the notion of \emph{unison manipulation} \cite{walsh2010empirical,durand2023coalitional,durand2025irv}.  
A voting rule~\( f \) is said to be \emph{unison-manipulable} (UM) in a profile~\( P \) (or equivalently, \( P \) is UM under~\( f \)) if a manipulation can succeed even when all interested voters cast the same ballot (see Section~\ref{sec:analysis_um_2} for an example).

\begin{lemma}[UM \cite{durand2025irv}]\label{lem_um_in_neighborhood}
	Assume there exists a neighborhood~$\mathcal{N}$ of the expected normalized profile~$\hat{P}$ such that, for any profile~$P$, if $P$ lies in~$\mathcal{N}$, then the homogeneous rule~$f$ is UM.
	
	Then \( \lim_{n\to\infty} \rho(f, m, n, \theta) = 1 \).
\end{lemma}

Before introducing the third lemma, we recall the notion of \(\delta\)-stable coalitional manipulability \cite{durand2025irv}, for any \(\delta > 0\).
A rule~\(f\) is said to be \emph{\(\delta\)-stable-CM} in a continuous profile~\(P\) (or equivalently, \( P \) is \(\delta\)-stable-CM in~\( f \)) if there exists a continuous profile~\(Q\) such that:
\begin{itemize}
	\item \(f\) is CM from \(P\) to \(Q\), and
	\item for any profile~\(Q'\) with \(d_\infty(Q, Q') < \delta\), we have \(f(Q') = f(Q)\).
\end{itemize}
Here, \(d_\infty(Q, Q')\) denotes the \(\ell^\infty\) distance between profiles, viewed as vectors of weights. Intuitively, $f$ is CM from $P$ to $Q$ with an outcome that is stable close enough to~$Q$.

\begin{lemma}[\(\delta\)-stable-CM \cite{durand2025irv}]\label{lem_delta_stable_cm}
	Assume there exist \(\delta > 0\) and a neighborhood~$\mathcal{N}$ of the expected normalized profile~\(\hat{P}\) such that, for any profile~$P$, if $P$ lies in~$\mathcal{N}$, then the homogeneous rule~\(f\) is \(\delta\)-stable-CM.
	
	Then \(\lim_{n\to\infty} \rho(f, m, n, \theta) = 1\).
\end{lemma}

The general proof strategy is as follows, for some value~$\theta^*$ conjectured to be the critical concentration parameter:
\begin{itemize}
\item show that candidate~1 wins in a neighborhood of~$\hat{P}$;
\item show that for $\theta > \theta^*$ the rule is non-CM in a neighborhood of~$\hat{P}$, allowing us to apply Lemma~\ref{lem_not_CM_in_neighborhood};
\item show that for $\theta < \theta^*$ the rule is UM or \(\delta\)-stable-CM in a neighborhood of~$\hat{P}$, allowing us to apply Lemma~\ref{lem_um_in_neighborhood} or~\ref{lem_delta_stable_cm}.
\end{itemize}
We may then conclude that the rule admits the critical parameter $\theta_c(f, m) = \theta^*$. Variants of this strategy are sometimes required, for instance when candidate~1 does not win in a neighborhood of~$\hat{P}$ (see Section~\ref{sec:proof_bucklin}).

If a rule is defined only on discrete profiles, as is the case for Young, Dodgson, and Simplified Dodgson, we proceed similarly.  
However, we cannot work directly with~$\hat{P}$ or with continuous profiles in its neighborhood.  
Instead, we consider discrete profiles~$P$ whose normalized version~$\bar{P}$ lies close to~$\hat{P}$.  
We then use the following generalized lemmas, which are proved in exactly the same way using the weak law of large numbers.  
Differences with the original statements are highlighted in bold.

\begin{lemma}[Non-CM, generalized version]\label{lem_not_CM_in_neighborhood_generalized}
Assume there exists a neighborhood~$\mathcal{N}$ of the expected normalized profile~$\hat{P}$ such that, for any profile~$P$, if \textbf{its normalized version~$\bar{P}$ lies in~$\mathcal{N}$ and $n(P)$ is large enough},  
then the \textbf{(not necessarily homogeneous)} rule~$f$ is non-CM.  

Then \( \lim_{n \to \infty} \rho(f, m, n, \theta) = 0\).
\end{lemma}

\begin{lemma}[UM, generalized version]\label{lem_um_in_neighborhood_generalized}
Assume there exists a neighborhood~$\mathcal{N}$ of the expected normalized profile~$\hat{P}$ such that, for any profile~$P$, if \textbf{its normalized version~$\bar{P}$ lies in~$\mathcal{N}$ and $n(P)$ is large enough}, then the \textbf{(not necessarily homogeneous)} rule~$f$ is~UM.

Then \( \lim_{n\to\infty} \rho(f, m, n, \theta) = 1 \).
\end{lemma}

In all cases, the convergence is exponentially fast, by the same concentration argument as in \citet{durand2025irv}, relying on Hoeffding’s inequality.

\subsection{Expected Profile}\label{sec:analysis_sincere_profile}

In this section, we analyze the expected profile~$\hat{P}$, which will be used in many subsequent proofs.  
Its weighted majority matrix $W(\hat{P})$ is shown in Table~\ref{tab:wmm_original}.  
Candidate~1 is the Condorcet winner; since this conclusion relies on strict inequalities, it remains valid in a neighborhood of~$\hat{P}$, and therefore also holds for the random profile~$P$ with high probability by the weak law of large numbers.  

\begin{table}[ht]
	\centering
	\caption{Weighted majority matrix $W(\hat{P})$ of the expected profile. Indices $j$ and $k$ denote generic candidates, with $2 < j < k$, and $k$ defined if $m \geq 4$.}\label{tab:wmm_original}
	\(\begin{array}{c|c|c|c|c|}
		& 1                     & 2                              & j                              & k                              \\ \hline
		1 &                      & \frac{1}{2}(1-\theta) + \theta & \frac{1}{2}(1-\theta) + \theta & \frac{1}{2}(1-\theta) + \theta \\ \hline
		2 & \frac{1}{2}(1-\theta) &                               & \frac{1}{2}(1-\theta) + \theta & \frac{1}{2}(1-\theta) + \theta \\ \hline
		j & \frac{1}{2}(1-\theta) & \frac{1}{2}(1-\theta)          &                               & \frac{1}{2}(1-\theta) + \theta \\ \hline
		k & \frac{1}{2}(1-\theta) & \frac{1}{2}(1-\theta)          & \frac{1}{2}(1-\theta)          &                               \\ \hline
	\end{array}\)
\end{table}

For two distinct opponents of candidate~1, denoted $c$ and~$d$, note that the $(1-\theta)$ purely random voters (the “IC part”) treat these candidates symmetrically, whereas the remaining $\theta$ voters (the “Dirac part”) all rank candidate~1 first. Hence,
\[
w(\hat{P}^{1 \succ c \text{ and } 1 \succ d}) = \tfrac{1}{3}(1-\theta) + \theta,
\]
a useful quantity that appears in several strengthened notions of Condorcet winner.

\subsection{Expected Profile Restricted to $1 \succ 2$}\label{sec:analysis_sincere_voters}

For most of the rules considered in this paper, candidate~1 is the winner in the expected profile~$\hat{P}$.  
Candidate~2 then naturally emerges as the main challenger, and we will frequently examine manipulation attempts in favor of~2.  
As a preliminary step, we study the contribution of the sincere voters, that is, those who rank candidate~1 above~2 in~$\hat{P}$.  
The corresponding weighted majority matrix~$W(\hat{P}^{1 \succ 2})$ is given in Table~\ref{tab:wmm_sincere_voters} for later reference.

\begin{table}[ht]
\centering
	\caption{Weighted majority matrix $W(\hat{P}^{1\succ 2})$ of the expected profile restricted to voters who rank candidate~1 above~2. Indices $j$ and $k$ denote generic candidates, with $2 < j < k$, and $k$ defined if $m \geq 4$.}\label{tab:wmm_sincere_voters}
	\(\begin{array}{c|c|c|c|c|}
		& 1                     & 2                              & j                              & k                              \\ \hline
		1 &                       & \frac{1}{2}(1-\theta) + \theta & \frac{1}{3}(1-\theta) + \theta & \frac{1}{3}(1-\theta) + \theta \\ \hline
		2 & 0                     &                                & \frac{1}{6}(1-\theta) + \theta & \frac{1}{6}(1-\theta) + \theta \\ \hline
		j & \frac{1}{6}(1-\theta) & \frac{1}{3}(1-\theta)          &                                & \frac{1}{4}(1-\theta) + \theta \\ \hline
		k & \frac{1}{6}(1-\theta) & \frac{1}{3}(1-\theta)          & \frac{1}{4}(1-\theta)          &                                \\ \hline
	\end{array}\)
\end{table}

\subsection{A Simple Manipulation Strategy}\label{sec:analysis_um_2}

When candidate~1 is the original winner, a simple manipulation strategy in favor of candidate~2 consists in casting the ballot $(2 \succ \cdots \succ m \succ 1)$.  
Let~$Q$ denote the profile obtained from~$\hat{P}$ when all voters who sincerely prefer candidate~2 to candidate~1 cast this ballot.  
The corresponding weighted majority matrix~$W(Q)$ is shown in Table~\ref{tab:wmm_um_2} and will be used in several proofs.

\begin{table}[ht]
\centering
	\caption{Weighted majority matrix $W(Q)$, with $Q$ derived from $\hat{P}$ when all voters sincerely preferring candidate~2 to~1 alter their ballots to $(2 \succ \ldots \succ m \succ 1)$. Indices $j$ and $k$ denote generic candidates, with $2 < j < k$, and $k$ defined if $m \geq 4$.}\label{tab:wmm_um_2}
	\(\begin{array}{c|c|c|c|c|}
		& 1                      & 2                               & j                               & k                               \\ \hline
		1 &                        & \frac{1}{2} (1-\theta) + \theta & \frac{1}{3} (1-\theta) + \theta & \frac{1}{3} (1-\theta) + \theta \\ \hline
		2 & \frac{1}{2} (1-\theta) &                                 & \frac{2}{3} (1-\theta) + \theta & \frac{2}{3} (1-\theta) + \theta \\ \hline
		j & \frac{2}{3} (1-\theta) & \frac{1}{3} (1-\theta)          &                                 & \frac{3}{4} (1-\theta) + \theta \\ \hline
		k & \frac{2}{3} (1-\theta) & \frac{1}{3} (1-\theta)          & \frac{1}{4} (1-\theta)          &                                 \\ \hline
	\end{array}\)
\end{table}

\section{Proofs of Theoretical Results}\label{sec:appendix_proof_of_theoretical_results}

This appendix provides proofs of all theoretical results stated in the paper, and also establishes the corresponding results for Split Cycle and the Viennot rule.
As in the previous section, we always assume $m \geq 3$.

\subsection{Maximin family: Maximin, Ranked Pairs, Schulze, Split Cycle, Young, Viennot}

We analyze here the Maximin family.
We first study the properties of Split Cycle and the Young rule in Sections~\ref{sec:split_cycle_is_maximin_like} and~\ref{sec:young_and_pscw}.
In Section~\ref{sec:critical_thetas_rule_related_to_pscw}, we then determine the common critical concentration parameter of the family, except for the Viennot rule, which is not structurally related to the Pair-Safe Condorcet Winner and is therefore treated separately in Section~\ref{sec:viennot}.

\subsubsection{Split Cycle is Maximin-like}\label{sec:split_cycle_is_maximin_like}

As for Maximin, Ranked Pairs, and Schulze (Proposition~\ref{thm:maximin_like_which_rules}), it is straightforward to verify that Split Cycle is Maximin-like.
Consequently, by Theorem~\ref{thm:maximin_like_rules_property}, if a PSCW exists, then Split Cycle is immune to coalitional manipulation.
Here again, the converse does not hold, as illustrated by Table~\ref{tab:ex_scw_no_pscw_rules_not_CM}, analyzed in Section~\ref{sec:appendix_analysis_table_1}: the profile admits no PSCW, yet Split Cycle remains immune to coalitional manipulation.

%\ThmMaximinLikeWhichRules*
%\ThmMaximinLikeRulesProperty*

\subsubsection{The Young rule and the PSCW}\label{sec:young_and_pscw}

In the main body, the Young rule was defined in the most common way, through the notion of penalty.  
For the purpose of the result below, it is convenient to express it in terms of a \emph{Young score}:
\[
s_\You(c, P) =
\begin{cases}
n(P) - p_\You(c, P), & \text{if $c$ is not a CW,}\\[4pt]
2 \min_{d \neq c}\big(w(P^{c \succ d})\big) - 1, & \text{if $c$ is a CW.}
\end{cases}
\]

The first expression conveys the main intuition: rather than counting the minimal number of voters that must be removed so that $c$ becomes a Condorcet winner, it represents the maximal number of voters that can be kept.  

The second expression can be interpreted through the following thought experiment.
In addition to the voters in~$P$, consider an infinite pool of virtual voters who all rank~$c$ last. We then ask for the maximal number of voters that can be selected so that~$c$ is a Condorcet winner, either from~$P$ or from this additional pool.
This thought experiment also applies to the first case. 

Finally, recall that if $c$ can never be a Condorcet winner, we defined by convention $p_\You(c, P) = n(P) + 1$.
This falls under the first case above and yields $s_\You(c, P) = -1$, which can be interpreted as follows: not only can no voter be selected from~$P$ or from the additional pool, but one would even need to add a hypothetical \emph{anti-voter} from the virtual pool, that is, a voter who ranks~$c$ first.

\medskip

This definition is convenient because it provides simple bounds.  
For any opponent $d$, candidate $c$ must defeat $d$ in pairwise comparison; hence
\begin{equation}\label{eq:young_score_upper_bound}
s_\You(c, P) \le 2\, w(P^{c \succ d}) - 1.
\end{equation}
On the other hand—and this is where the second case of our definition proves useful—considering the rankings where $c$ is placed first yields a lower bound on the Young score:
\begin{equation}\label{eq:young_score_lower_bound}
s_\You(c, P) \ge 2\, w(P^{r(c) = 1}) - 1.
\end{equation}
We will use variants of this bound in the proof below.

\ThmYoungNotMaximinLike*

We already proved that the Young rule is not Maximin-like (Table~\ref{tab:ex_you_is_not_maximin_like}, analyzed in Section~\ref{sec:appendix_analysis_table_2}). It remains to prove the second part of the proposition.

\begin{proof}[Proof of Proposition~\ref{thm:young_not_maximin_like}]
We actually prove a stronger statement: for any candidate~$d \neq c$, voters who prefer $d$ to $c$ cannot even make $d$ obtain a higher Young score than~$c$.  
For a proof by contradiction, assume that there exists a target profile $Q$ in which this occurs.

If we modify the profile $Q$ by moving $d$ up and $c$ down in the ranking of every manipulator, this can only increase the Young score of $d$ and decrease that of $c$.  
Hence, without loss of generality, we may assume that in $Q$ all manipulators rank $d$ first and $c$ last.  
This assumption also ensures that each voter keeps $c$ and $d$ in the same order as in $P$, making it easy to identify sincere voters and manipulators in $Q$ based on their relative ordering of $c$ and $d$.

Let $e \in \mathcal{C}(P) \setminus \{c, d\}$.  
Applying the PSCW condition~\eqref{eq:def_pscw} with opponents $(e, d)$, we obtain
\[
w(P^{c \succ d \text{ and } c \succ e}) > w(P^{e \succ c}) > w(P^{c \succ d \text{ and } e \succ c}).
\]
We now transfer this inequality to profile $Q$.  
The left-hand term represents sincere voters, so its weight cannot decrease:
\[
w(Q^{c \succ d \text{ and } c \succ e}) \ge w(P^{c \succ d \text{ and } c \succ e}).
\]
The right-hand term also corresponds to sincere voters; its weight cannot decrease for the same reason.  
Moreover, since all manipulators rank $c$ last, it cannot increase either. Hence,
\[
w(Q^{c \succ d \text{ and } e \succ c}) = w(P^{c \succ d \text{ and } e \succ c}).
\]
Combining these three relations yields a useful relation between two subsets that partition the sincere voters:
\[
w(Q^{c \succ d \text{ and } c \succ e}) > w(Q^{c \succ d \text{ and } e \succ c}).
\]

Let us now compute the Young score of $c$.  
To defeat $e$, we can of course include all voters in $P^{c \succ d \text{ and } c \succ e}$.  
From the inequality above, we also know that we can include all voters in $P^{c \succ d \text{ and } e \succ c}$.  
Hence, we may freely add manipulators—who rank $c$ last—and possibly some virtual voters from the pool.  
The number of manipulators and virtual voters that can be added depends only on the weakest pairwise contest of $c$ against some third candidate~$e$ in~$Q$.  
Let $e$ denote the candidate that yields this weakest contest.  
We then have
\[
s_\You(c, Q) = 2\, w(Q^{c \succ e}) - 1.
\]
Since the only voters ranking $c$ above $e$ are sincere, it follows that
\[
s_\You(c, Q) = 2\, w(Q^{c \succ d \text{ and } c \succ e}) - 1.
\]

For the score of~$d$, we use the upper bound~\eqref{eq:young_score_upper_bound}:
\[
s_\You(d, Q) \le 2\, w(Q^{d \succ c}) - 1.
\]

Finally, applying the PSCW property with opponents $(d, e)$, we obtain
\[
w(Q^{c \succ d \text{ and } c \succ e}) > w(Q^{d \succ c}),
\]
which implies that
\[
s_\You(c, Q) > s_\You(d, Q),
\]
yielding the desired contradiction.
\end{proof}

\subsubsection{Critical concentration parameters}\label{sec:critical_thetas_rule_related_to_pscw}

We now analyze the asymptotic behavior of the PSCW notion and of the Maximin family under the Perturbed Culture model, excluding the Viennot rule, whose analysis requires a substantially different proof and is therefore deferred to Section~\ref{sec:viennot}.

\ThmThetaCriticalPSCW*

\begin{proof}[Proof of Theorem~\ref{thm:theta_critical_pscw}]
Since $\theta > 0$, candidate~1 is the Condorcet winner in~$\hat{P}$. For $c=d$, the PSCW condition~\eqref{eq:def_pscw} reduces to the Condorcet condition $w(\hat{P}^{1\succ c})>w(\hat{P}^{c\succ 1})$, so it suffices to consider two distinct opponents~$c$ and~$d$.
We have
\[
w(\hat{P}^{1 \succ c \text{ and } 1 \succ d}) = \theta + \tfrac{1}{3}(1-\theta), 
\quad \textrm{and} \quad
w(\hat{P}^{c \succ 1}) = \tfrac{1}{2}(1-\theta).
\]
If $\theta > \frac{1}{7}$, then $w(\hat{P}^{1 \succ c \text{ and } 1 \succ d}) > w(\hat{P}^{c \succ 1})$, hence candidate~1 is a PSCW in~$\hat{P}$.
This property also holds in a neighborhood of~$\hat{P}$, since it relies on strict inequalities involving quantities that vary continuously with the profile.
By the weak law of large numbers, with high probability the random profile~$P$ has its normalized version~$\bar{P}$ in this neighborhood, and therefore admits a PSCW. If $\theta < \frac{1}{7}$, we conclude similarly that a PSCW fails to exist with high probability.
\end{proof}

In the subsequent proofs, unless stated otherwise, we will restrict attention to the normalized profile and leave the end of the argument implicit—namely, that the reasoning relies on strict inequalities involving continuous quantities and therefore holds in a neighborhood.

\medskip
An immediate consequence of Theorem~\ref{thm:theta_critical_pscw} is Corollary~\ref{thm:theta_u_if_protected_by_pscw}.

\ThmThetaUIfProtectedByPSCW*

We can finally prove the existence of a phase transition and determine the corresponding critical concentration parameter for the rules of the Maximin family, except the Viennot rule.

\begin{theorem}\label{thm:critical_theta_maximin_family_except_viennot}
For Maximin, Ranked Pairs, Schulze, Split Cycle, and Young, the critical concentration parameter is
\[
\theta_c(f, m) = \frac{1}{7}.
\]
\end{theorem}

\begin{proof}[Proof of Theorem~\ref{thm:critical_theta_maximin_family_except_viennot}]
By Corollary~\ref{thm:theta_u_if_protected_by_pscw}, what remains to prove is 
that for $\theta < \frac{1}{7}$, the rule $f$ is CM w.h.p.

In the expected sincere profile~$\hat{P}$, candidate~1 is the Condorcet winner, and this is also true for all profiles $P$ whose normalized version $\bar{P}$ is sufficiently close to $\hat{P}$. Since $f$ is Condorcet-consistent, candidate~1 is declared the winner with high probability.

Consider the unison manipulation attempt~$Q$ in favor of candidate~2 with ballot $(2 \succ \ldots \succ m \succ 1)$, whose weighted majority matrix is given in Table~\ref{tab:wmm_um_2}.  
If $\theta < \tfrac{1}{7}$, we have
\[
\tfrac{1}{2}(1 - \theta) > \tfrac{1}{3}(1 - \theta) + \theta,
\]
which implies that $W(2, 1, Q) > W(1, j, Q)$ for any third candidate~$j$.  
Applying Maximin, Ranked Pairs, Schulze, and Split Cycle to Table~\ref{tab:wmm_um_2}, a brief rule-by-rule inspection shows that candidate~2 is elected. By Lemma~\ref{lem_um_in_neighborhood}, we conclude that these rules are CM with high probability.

We now turn to the case of Young.  
Let us temporarily assume that there exists a profile~$P$ whose normalized version is equal to~$\hat{P}$.  
Applying the same unison manipulation as above to~$P$ yields the profile~$n(P)Q$.  
The upper bound~\eqref{eq:young_score_upper_bound} gives
\[
s_\You(1, n(P)Q) \le 2\, w(n(P)Q^{1 \succ j}) - 1 = 2n(P) \left( \tfrac{1}{3}(1-\theta) + \theta \right) - 1,
\]
and the lower bound~\eqref{eq:young_score_lower_bound} gives
\[
s_\You(2, n(P)Q) \ge 2\, w(n(P)Q^{r(2) = 1}) - 1 = 2 n(P) \left( \tfrac{1}{2}(1-\theta) \right) - 1.
\]
If $\theta < \frac{1}{7}$, these bounds imply that $s_\You(2, n(P)Q) > s_\You(1, n(P)Q)$.  
It is also straightforward to verify that all other candidates have lower scores than~2.  
If~$P$ has its normalized version sufficiently close to~$\hat{P}$, all the strict inequalities involved still hold, and~$P$ remains CM.  
We then conclude by Lemma~\ref{lem_um_in_neighborhood_generalized}.
\end{proof}

\subsubsection{The Viennot rule}\label{sec:viennot}

Like the Young rule, but unlike the other rules in the Maximin family, the Viennot rule is not Maximin-like (Table~\ref{tab:ex_you_is_not_maximin_like}, analyzed in Section~\ref{sec:appendix_analysis_table_2}). Moreover, unlike the Young rule, the existence of a PSCW does not even guarantee immunity to coalitional manipulation (Table~\ref{tab:ex_pscw_no_sscw}, analyzed in Section~\ref{sec:appendix_analysis_table_3}). Conversely, the Viennot rule may still be immune to coalitional manipulation without a PSCW (Table~\ref{tab:ex_scw_no_pscw_rules_not_CM}, analyzed in Section~\ref{sec:appendix_analysis_table_1}). 
Nevertheless, its critical concentration parameter is equal to~$\tfrac{1}{7}$, but for a different reason. Intuitively, when manipulating in favor of candidate~2, this value corresponds to the threshold below which manipulators can force an elimination duel between candidates~1 and~3. It then becomes relatively easy to eliminate candidate~1 in that pairwise comparison. Once candidate~1 is eliminated, the path is clear for candidate~2.

\begin{theorem}\label{thm:critical_theta_viennot}
For Viennot, the critical concentration parameter is
\[
\theta_c(\Vie, m) = \frac{1}{7}.
\]
\end{theorem}

\begin{proof}[Proof of Theorem~\ref{thm:critical_theta_viennot}]
In the expected sincere profile~$\hat{P}$, candidate~1 is the Condorcet winner and is therefore declared the winner.

Assume that $\Vie$ is CM in $\hat{P}$ towards some target profile~$Q$ in favor of a candidate~$c$.  
Since candidate~1 is the Condorcet winner, even after manipulation we have $W(1, c, Q) > W(c, 1, Q)$.  
Hence, to eliminate candidate~1, a duel must be organized between~1 and some third candidate $d \notin \{1, c\}$.  
Let $k$, with $3 \le k \le m$, be the number of candidates present at this round, denoted by $\{1, c, d, j_1, \ldots, j_{k-3}\}$.  
Every candidate except~$d$ must then have at least the same plurality score as candidate~1.  
Therefore, the number of manipulators must be large enough to compensate for the vote deficits of these candidates, yielding
\[
\tfrac{1}{2}(1-\theta) \ge 
\Big[\tfrac{1}{k}(1-\theta) + \theta\Big] 
+ (k-3)\Big[\tfrac{1}{k}(1-\theta) + \theta - \tfrac{1}{2k}(1-\theta)\Big],
\]
which simplifies to $\theta \le \tfrac{1}{2k(k - 2) + 1}$.  
Since this expression is maximized for $k = 3$, we must in particular have $\theta \le \tfrac{1}{7}$.  
By contraposition, if $\theta > \tfrac{1}{7}$, Viennot is non-CM in~$\hat{P}$, and therefore in any neighborhood of it. 
By Lemma~\ref{lem_not_CM_in_neighborhood}, Viennot is therefore non-CM with high probability.

Now assume $\theta < \tfrac{1}{7}$ and consider the unison manipulation attempt~$Q$ in favor of candidate~2 with ballot $(2 \succ \ldots \succ m \succ 1)$, described in Section~\ref{sec:analysis_um_2}.  
Assume a round where candidates~1,~2, and $k-2$ other candidates remain, with $3 \leq k \leq m$.  
Candidate~1 then has a plurality score of $\tfrac{1}{k}(1-\theta) + \theta$, candidate~2 has a score of $\tfrac{1}{2}(1-\theta)$, and any other candidate has a score of $\tfrac{1}{2k}(1-\theta)$, which is always smaller than those of candidates~1 and~2.  
If several candidates other than~1 and~2 are still present, they all have the lowest plurality scores, and one of them is eliminated.  
(The specific candidates selected and eliminated may vary depending on the initial profile within the neighborhood of~$\hat{P}$, but the outcome remains the same.) When only one candidate $j \notin \{1, 2\}$ remains (i.e., when $k=3$), and since $\theta < \tfrac{1}{7}$, we have $\tfrac{1}{3}(1-\theta) + \theta < \tfrac{1}{2}(1-\theta)$.  
Hence, candidates~1 and~$j$ are selected for the elimination duel.  
Table~\ref{tab:wmm_um_2} then shows that candidate~1 is eliminated.  
The final counting round thus involves candidates~2 and~$j$, and Table~\ref{tab:wmm_um_2} again shows that~$j$ is eliminated.  
Candidate~2 is therefore declared the winner.  
Consequently, Viennot is UM in~$\hat{P}$ and in a neighborhood of it.  
By the UM Lemma~\ref{lem_um_in_neighborhood}, Viennot is CM with high probability.
\end{proof}

\subsection{Baldwin family: Baldwin, Nanson, Kemeny, Dodgson, Simplified Dodgson}\label{sec:appendix_baldwin_family}

We now turn to the Baldwin family.  
We recall that the rules in this family are immune to coalitional manipulation whenever a Set-Safe Condorcet Winner exists (Proposition~\ref{thm:baldwin_nanson_kemeny_sscw}), except Dodgson.

%\ThmBaldwinNansonKemenySSCW*

\ThmThetaCriticalSSCW*

\begin{proof}[Proof of Theorem~\ref{thm:theta_critical_sscw}]
For any two distinct opponents $d$ and~$e$, we have
\[
w(\hat{P}^{1 \succ d \text{ and } 1 \succ e}) = \tfrac{1}{3}(1-\theta) + \theta.
\]
If $\theta > \tfrac{1}{4}$, this quantity is higher than $\tfrac{1}{2}$, hence $1$ is a RCW, hence an SSCW, in $\hat{P}$. Otherwise, this quantity is at most $\frac{1}{2}$. Intuitively, this implies that in a manipulation attempt in favor of~$d$, sincere voters cannot guarantee a pairwise victory against~$e$.
Evaluating the SSCW condition~\eqref{eq:def_sscw_computable} for candidate~1 and any opponent~$d$ in the expected profile~$\hat{P}$ then gives
\[
\begin{aligned}
& \Big(w(\hat{P}^{1 \succ d}) - \tfrac{1}{2}\Big) 
+ \sum_{e \notin \{1, d\}} 
\min\Big(0,\, w(\hat{P}^{1 \succ d \text{ and } 1 \succ e}) - \tfrac{1}{2}\Big) \\
&= \left(\tfrac{1}{2}(1-\theta) + \theta - \tfrac{1}{2}\right) 
+ \sum_{e \notin \{1, d\}} 
\min\left(0,\, \tfrac{1}{3}(1-\theta) + \theta - \tfrac{1}{2}\right) \\ 
&= \frac{(4m - 5)\theta - (m - 2)}{6}.
\end{aligned}
\]
If $\theta$ is greater (resp. lower) than $\frac{m-2}{4m-5}$, then this quantity is positive (resp. negative), hence candidate~1 is (resp. is not) the SSCW.
This property also holds in a neighborhood of~$\hat{P}$, and the weak law of large numbers ensures that a random profile~$P$ has its normalized version~$\bar{P}$ in this neighborhood with high probability.
\end{proof}

An immediate consequence of Theorem~\ref{thm:theta_critical_sscw} is Corollary~\ref{thm:theta_u_if_protected_by_sscw}.

\ThmThetaUIfProtectedBySSCW*

We can now establish the phase transition phenomenon for the rules of the Baldwin family.
Although Dodgson is not rendered immune to coalitional manipulation by the existence of an SSCW, its similarity with Simplified Dodgson allows us to reach the same conclusion.

\begin{theorem}\label{thm:critical_theta_baldwin_family}
For Baldwin, Nanson, Kemeny, Dodgson, and Simplified Dodgson, the critical concentration parameter is
\[
\theta_c(f, m) = \frac{m - 2}{4m - 5}.
\]
\end{theorem}

\begin{proof}[Proof of Theorem~\ref{thm:critical_theta_baldwin_family}]
With high probability, candidate~1 is the Condorcet winner and is therefore elected in the random profile~$P$, since the rule is Condorcet-consistent.

\paragraph{Supercritical regime for Baldwin, Nanson, Kemeny, and Simplified Dodgson}

By Corollary~\ref{thm:theta_u_if_protected_by_sscw}, we already know that for $\theta > \tfrac{m - 2}{4m - 5}$, all these rules are non-CM with high probability.

\paragraph{Supercritical regime for Dodgson} 

For Dodgson, if $\theta > \tfrac{1}{4}$, candidate~1 is the RCW and therefore the profile is non-CM with high probability.  
Let us now assume $\theta \in \big(\tfrac{m - 2}{4m - 5}, \tfrac{1}{4}\big]$.  
Temporarily assume that there exists a profile~$P$ whose normalized version coincides with the expected profile~$\hat{P}$.  
We examine a manipulation attempt to a target profile~$Q$ in favor of some candidate~$c$, assuming that manipulators try to maximize the score of~$c$ and minimize that of~1, disregarding the scores of all other candidates.  
We will show that even in this case, candidate~1 still has a higher score than~$c$.  
To ease reading, one may keep in mind the illustrative case $c = 2$ and refer to Section~\ref{sec:analysis_sincere_voters} for the profile restricted to sincere voters, as well as Section~\ref{sec:analysis_um_2}, which provides an example of such a manipulation.

We first examine the score of candidate~1.  
Observe that~1 still wins its pairwise comparison against~$c$ and only loses to the other opponents, with a uniform score.
Consider the following algorithm: as long as~1 loses to opponents $d \neq c$, pick a sincere voter who does not rank~1 first and swap~1 with the candidate immediately above it (which cannot be~2, since the voter is sincere).
Apply the same swap to all rankings obtained by an arbitrary circular permutation of the candidates $(d_1, \ldots, d_{m-2})$.
Each step of this procedure increases the score of candidate~1 against every opponent $d \neq c$ by~1 point. If we were to continue these swaps as long as such voters exist, candidate~1 would eventually reach a score
\(
n(P)\!\big(\tfrac{1}{2}(1-\theta)+\theta\big)
\)
in each pairwise comparison against any $d \neq c$, hence a victory.  
Therefore, there is a step at which the procedure stops, and at that moment candidate~1 wins all its pairwise comparisons by exactly one point, and no useless swap was needed. It remains to bound the number of swaps required. 

The score of candidate~1 against any candidate $d \neq c$ is
\(
n(P)\big(\tfrac{1}{3}(1-\theta) + \theta\big).
\)
To turn this into a victory, candidate~1 needs a number of swaps at most
\[
\frac{n(P)}{2} - n(P)\left(\tfrac{1}{3}(1-\theta) + \theta\right) + 1
= n(P)\,\frac{1 - 4\theta}{6} + 1.
\]
Since there are $(m-2)$ such pairwise contests to recover, we obtain
\[
p_\Dod(1, Q) \le (m-2)\,n(P)\,\frac{1 - 4\theta}{6} + (m-2).
\]

In the manipulated profile, candidate~$c$ must compensate for its defeat against candidate~1, hence
\[
p_\Dod(c, Q) \ge \tfrac{1}{2} n(P)\,\theta.
\]
We then obtain
\[
p_\Dod(c, Q) - p_\Dod(1, Q)
\ge \frac{n(P)}{6}\big[(4m - 5)\theta - (m - 2)\big] - (m - 2).
\]
For $n(P)$ large enough, this quantity is positive, and the manipulation therefore fails.  
Since this remains true for profiles~$P$ whose normalized version is sufficiently close to~$\hat{P}$, Lemma~\ref{lem_not_CM_in_neighborhood_generalized} implies that Dodgson is non-CM with high probability.

\bigskip
Assume now that $\theta < \tfrac{m-2}{4m-5}$. We shall prove that, under this condition, all the rules of the Baldwin family are coalitionally manipulable with high probability.

\paragraph{Subcritical regime for Baldwin and Nanson}
Let $\zeta > 0$ (its value will be specified later) and consider a profile~$P$ such that $d_\infty(P, \hat{P}) \leq \frac{\zeta}{m!}$. 
If $\zeta < \theta$, candidate~1 remains the Condorcet winner in~$P$ and is therefore elected.  
Construct a new profile~$Q$ obtained from~$P$ by modifying the ballots of voters who prefer candidate~2 to candidate~1 as follows:
\begin{itemize}
	\item A fraction $\tfrac{1-3\theta}{2(1-\theta)}$ of them vote $(2 \succ m \succ \cdots \succ 3 \succ 1)$;
	\item A fraction $\tfrac{1-3\theta}{2(1-\theta)}$ vote $(2 \succ 3 \succ \cdots \succ m \succ 1)$;
	\item A fraction $\tfrac{2\theta}{1-\theta}$ vote $(m \succ \cdots \succ 1)$.
\end{itemize}
Let $\delta > 0$ and consider a profile~$Q'$ satisfying $d_\infty(Q, Q') < \tfrac{\delta}{m!}$. 
The weighted majority matrix (WMM) of~$Q'$ is then as shown in Table~\ref{tab:baldwin_nanson_manipulated_profile}, with all entries given up to an error of at most $\zeta + \delta$.

\begin{table}[ht]
\centering
	\caption{Weighted majority matrix $W(Q')$ in the proof of the subcritical regime for Baldwin and Nanson, up to errors of at most $\zeta + \delta$. Indices $j$ and $k$ denote generic candidates, with $2 < j < k$, and $k$ defined if $m \geq 4$.}\label{tab:baldwin_nanson_manipulated_profile}
	\(\begin{array}{c|c|c|c|c|}
	& 1                     & 2                              & j                              & k                              \\ \hline
	1 &                       & \frac{1}{2}(1-\theta) + \theta & \frac{1}{3}(1-\theta) + \theta & \frac{1}{3}(1-\theta) + \theta \\ \hline
	2 & \frac{1}{2}(1-\theta) &                                & \frac{2}{3}(1-\theta)          & \frac{2}{3}(1-\theta)          \\ \hline
	j & \frac{2}{3}(1-\theta) & \frac{1}{3}(1-\theta) + \theta &                                & \frac{1}{2}                    \\ \hline
	k & \frac{2}{3}(1-\theta) & \frac{1}{3}(1-\theta) + \theta & \frac{1}{2}                    &                                \\ \hline
\end{array}\)
\end{table}

At the first round, the Borda scores are:
\[
\begin{aligned}
&s_\Bor(1, Q') = \tfrac{1}{2}(1 - \theta) + \theta 
    + (m-2) \left[ \tfrac{1}{3}(1 - \theta) + \theta \right]
    + b\big((m-2)(\zeta + \delta)\big),\\
&s_\Bor(2, Q') = \tfrac{1}{2}(1 - \theta)
    + (m-2) \left[ \tfrac{2}{3}(1 - \theta) \right]
    + b\big((m-2)(\zeta + \delta)\big),\\
&s_\Bor(j, Q') = \frac{m-1}{2}
    + b\big((m-2)(\zeta + \delta)\big)
    \quad \text{for all } j \ge 3,\\
\end{aligned}
\]
where we recall that $b(\cdot)$ denotes a real number bounded by the given quantity.
In particular,
\[
s_\Bor(2, Q') - s_\Bor(1, Q') 
    = \frac{(m-2) - (4m-5)\theta}{3}
    + b\big(2(m-2)(\zeta + \delta)\big).
\]
If all error terms are zero, we then deduce that candidate~2 obtains a score above the average, candidate~$j$ a score exactly equal to the average, and candidate~1 a score below the average.  
Returning to the general case of~$Q'$, for $\zeta = \delta$ small enough, these inequalities remain valid: candidate~2 stays above the average, candidate~1 below it, and candidate~1 has the lowest score overall.  
Consequently, under both Baldwin and Nanson, candidate~1 is eliminated while candidate~2 is not.  
If additional rounds occur, candidate~2 is the Condorcet winner in the restricted profile and therefore wins.  
By Lemma~\ref{lem_delta_stable_cm}, we conclude that Baldwin and Nanson are coalitionally manipulable with high probability.

\paragraph{Subcritical regime for Kemeny}
We now consider the unison manipulation described in Section~\ref{sec:analysis_um_2}. 
Let $r$ be an arbitrary ranking. To simplify the analysis, we define the \emph{reduced Kemeny penalty} of $r$ as:
\[
\tilde{p}_\Kem(r,P)=\sum_{(c,d)\in\mathcal{C}(P)^2:\,c\succ_r d}W(d,c,P) - \min\big(W(d,c,P), W(c,d,P)\big),
\] 
which is equal to $p_\Kem(r,P)$ up to an additive constant. The advantage is that in the sum, we only need to take into account the pairwise comparisons for which $r$ disagrees with the majority.
\begin{itemize}
\item For $r = (2 \succ \ldots \succ m \succ 1)$, the only pairwise comparison inconsistent with~$r$ is between candidates~1 and~2. Hence $\tilde{p}_\Kem(r) = \theta$.
\item If the top candidate of~$r$ is some $j \notin \{1,2\}$, then considering the defeat of~$j$ against~2 yields
\[
\tilde{p}_\Kem(r) \ge \tfrac{1}{3}(1-\theta) + \theta > \theta.
\]
\item If the top candidate of~$r$ is~1, then considering the pairs $(1,j)$ for $j \notin \{1,2\}$ we obtain
\[
\tilde{p}_\Kem(r) \ge (m-2)\big[\tfrac{1}{3}(1-\theta) - \theta\big].
\]
Since $\theta < \tfrac{m-2}{4m-5}$, this value exceeds~$\theta$.
\end{itemize}
Therefore, candidate~2 is the winner.  
By Lemma~\ref{lem_um_in_neighborhood}, Kemeny is coalitionally manipulable with high probability.

\paragraph{Subcritical regime for Dodgson and Simplified Dodgson}
We again rely on the unison manipulation described in Section~\ref{sec:analysis_um_2}.  
Assume temporarily that the normalized version of~$P$ coincides with the expected profile~$\hat{P}$.  
Applying the unison manipulation to~$P$ yields the profile $n(P)Q$. We then have:
\[
\begin{aligned}
&p_f(1, n(P) Q) \ge (m-2)\, n(P)\, \frac{1 - 4\theta}{6},\\
&p_f(2, n(P) Q) \le \tfrac{1}{2}\, n(P)\, \theta + 1,\\
\end{aligned}
\]
which leads to
\[
p_f(2, n(P) Q) - p_f(1, n(P) Q)
    \le \frac{n(P)}{6}\big[(4m - 5)\theta - (m - 2)\big] + 1.
\]
For $n(P)$ large enough, this difference is negative, hence $p_f(2, n(P) Q) < p_f(1, n(P) Q)$.  
Moreover, for any $j \notin \{1,2\}$,
\[
p_f(j, n(P) Q) \ge n(P) \left[ \tfrac{1}{6}(1 - \theta) + \tfrac{1}{2}\theta \right],
\]
which, for sufficiently large $n(P)$, also exceeds $p_f(2, n(P) Q)$.  
Therefore, candidate~2 is the winner.  
By Lemma~\ref{lem_um_in_neighborhood_generalized}, we conclude that $f$ (either Dodgson or Simplified Dodgson) is coalitionally manipulable with high probability.
\end{proof}

\subsection{Black family: Black, Slater, Copeland}\label{sec:appendix_black_family}

We continue with the Black family, related to the notion of Resistant Condorcet Winner.

\ThmThetaCriticalRCW*

\begin{proof}[Proof of Theorem~\ref{thm:theta_critical_rcw}]
Since $\theta > 0$, candidate~1 is the Condorcet winner in the expected profile~$\hat{P}$.  
It thus suffices to verify the RCW condition~\eqref{eq:def_rcw} for distinct opponents $d$ and~$e$.  
We have
\[
w(\hat{P}^{1 \succ d \text{ and } 1 \succ e}) - \tfrac{1}{2}
    = \tfrac{1}{3}(1 - \theta) + \theta  - \tfrac{1}{2} 
    = \tfrac{1}{6}(4 \theta - 1).
\]
If $\theta$ is greater (resp. lower) than $\frac{1}{4}$, then this quantity is positive (resp. negative).
This inequality remains valid in a neighborhood of~$\hat{P}$, implying that candidate~1 is (resp. is not) the RCW with high probability.
\end{proof}

An immediate consequence of Theorem~\ref{thm:theta_critical_rcw} is Corollary~\ref{thm:theta_u_if_protected_by_rcw}.

\ThmThetaUIfProtectedByRCW*

We can now establish the phase transition phenomenon for Black, Slater, and Copeland.

\begin{theorem}\label{thm:critical_theta_black_family}
For Black with $m \geq 3$, Slater with $m \geq 4$, and Copeland with $m \geq 5$, the critical concentration parameter is
\[
\theta_c(f, m) = \frac{1}{4}.
\]
\end{theorem}

\begin{proof}[Proof of Theorem~\ref{thm:critical_theta_black_family}]
As shown in Section~\ref{sec:analysis_sincere_profile}, with high probability candidate~1 is the Condorcet winner and is therefore elected.  
Moreover, by Corollary~\ref{thm:theta_u_if_protected_by_rcw}, if $\theta > \tfrac{1}{4}$, these rules are immune to coalitional manipulation with high probability.  

\medskip
Consider now the case $\theta < \tfrac{1}{4}$.

\paragraph{Subcritical regime for Black}

Let $\zeta > 0$ and consider a profile~$P$ such that $d_\infty(P, \hat{P}) \leq \tfrac{\zeta}{m!}$.
If $\zeta < \theta$, then candidate~1 is the Condorcet winner hence is elected in~$P$.
Construct a new profile~$Q$ from~$P$ by modifying the ballots of voters who prefer candidate~2 to candidate~1 as follows:
\begin{itemize}
\item A fraction $\tfrac{1 - 3\theta}{2(1 - \theta)}$ of them vote $(2 \succ \ldots \succ m \succ 1)$,  
\item A fraction $\tfrac{1 + \theta}{2(1 - \theta)}$ vote $(2 \succ m \succ \ldots \succ 3 \succ 1)$.  
\end{itemize}
Let $\delta > 0$ and consider a profile~$Q'$ satisfying $d_\infty(Q, Q') < \tfrac{\delta}{m!}$.  
The weighted majority matrix (WMM) of~$Q'$ is then given in Table~\ref{tab:black_manipulated_profile},  
with all entries given up to an error of absolute value $\zeta + \delta$.

\begin{table}[ht]
\centering
	\caption{Weighted majority matrix $W(Q')$ in the proof of the subcritical regime for Black, up to errors of at most $\zeta + \delta$. Indices $j$ and $k$ denote generic candidates, with $2 < j < k$, and $k$ defined if $m \geq 4$.}\label{tab:black_manipulated_profile}
	\(\begin{array}{c|c|c|c|c|}
		& 1                     & 2                              & j                              & k                              \\ \hline
		1 &                       & \frac{1}{2}(1-\theta) + \theta & \frac{1}{3}(1-\theta) + \theta & \frac{1}{3}(1-\theta) + \theta \\ \hline
		2 & \frac{1}{2}(1-\theta) &                                & \frac{2}{3}(1-\theta) + \theta & \frac{2}{3}(1-\theta) + \theta \\ \hline
		j & \frac{2}{3}(1-\theta) & \frac{1}{3}(1-\theta)          &                                & \frac{1}{2}                    \\ \hline
		k & \frac{2}{3}(1-\theta) & \frac{1}{3}(1-\theta)          & \frac{1}{2}                    &                                \\ \hline
	\end{array}\)
\end{table}

For $\zeta = \delta$ small enough, candidate~1 beats candidate~2, candidate~2 beats any candidate~$j > 2$, and each candidate~$j > 2$ beats candidate~1.  
Hence, there is no Condorcet winner.  
We now compute the Borda scores:
\[
\begin{aligned}
&s_\Bor(1, Q') = \tfrac{1}{2}(1 - \theta) + \theta 
    + (m - 2) \left[ \tfrac{1}{3}(1 - \theta) + \theta \right]
    + b\big((m - 1)(\zeta + \delta)\big), \\
&s_\Bor(2, Q') = \tfrac{1}{2}(1 - \theta)
    + (m - 2) \left[ \tfrac{2}{3}(1 - \theta) + \theta \right]
    + b\big((m - 1)(\zeta + \delta)\big),\\
&s_\Bor(j, Q') = \tfrac{2}{3}(1 - \theta) + \tfrac{1}{3}(1 - \theta)
    + (m - 3)\tfrac{1}{2}
    + b\big((m - 1)(\zeta + \delta)\big).\\
\end{aligned}
\]
We deduce:
\[
\begin{aligned}
&s_\Bor(2, Q') - s_\Bor(1, Q') 
    = \frac{(m - 2) - (m + 1)\theta}{3}
      + b\big(2(m - 1)(\zeta + \delta)\big),\\
&s_\Bor(2, Q') - s_\Bor(j, Q') 
    = \frac{m - 2 + (2m - 1)\theta}{6}
      + b\big(2(m - 1)(\zeta + \delta)\big).\\
\end{aligned}
\]
For $\zeta = \delta$ small enough, both differences are positive, 
so candidate~2 obtains the highest Borda score and is thus elected.  
We then conclude by Lemma~\ref{lem_delta_stable_cm}.

\paragraph{Subcritical regime for Slater}
Assume $m \ge 4$ and consider the unison manipulation described in Section~\ref{sec:analysis_um_2}.  
For $\theta < \tfrac{1}{4}$, the unweighted majority matrix $M(Q)$ of~$Q$ is as shown in Table~\ref{tab:umm_um_2}. Recall that this matrix indicates, for each pair of candidates, which one wins their head-to-head contest.

\begin{table}[ht]
\centering
	\caption{Unweighted majority matrix $M(Q)$ in the proof of the subcritical regime for Slater. Indices $j$ and $k$ denote generic candidates, with $2 < j < k$. Null coefficients are omitted for legibility.}\label{tab:umm_um_2}
	\(\begin{array}{c|c|c|c|c|}
		  & 1 & 2 & j & k \\ \hline
		1 &   & 1 &   &   \\ \hline
		2 &   &   & 1 & 1 \\ \hline
		j & 1 &   &   & 1 \\ \hline
		k & 1 &   &   &   \\ \hline
	\end{array}\)
\end{table}

Consider the ranking $r = (2 \succ \ldots \succ m \succ 1)$.  
This order is contradicted only by the victory of candidate~1 over candidate~2, hence $p_\Sla(r) = 1$.  
Assume that another ranking~$r'$ has a penalty of at most~1.  
Then its last candidate must have at most one victory, i.e., it must be either~1 or~$m$.  
If the last candidate is~1, then since all other candidates follow a Condorcet order, we must have $r' = r$.  
If the last candidate is~$m$, the remaining candidates are not in a Condorcet order, and the penalty is therefore at least~2.  
Hence the winning order is~$r$, and the winning candidate is~2.  
We then conclude by Lemma~\ref{lem_um_in_neighborhood}.

\paragraph{Subcritical regime for Copeland}
Let $\zeta > 0$ and consider a profile~$P$ such that $d_\infty(P, \hat{P}) \le \tfrac{\zeta}{m!}$.  
Construct a new profile~$Q$ from~$P$ by modifying the ballots of voters who prefer candidate~2 to candidate~1 as follows:
\begin{itemize}
    \item A fraction $\tfrac{2\theta}{1 - \theta}$ of them vote $(2 \succ m \succ \ldots \succ 3 \succ 1)$;
    \item Let $m'$ be the largest odd integer such that $m' \le m$.  
    A fraction $\tfrac{1 - 3\theta}{1 - \theta}$ of them is then evenly distributed among the $m' - 2$ rankings obtained by circularly permuting the candidates $(3, \ldots, m')$ within the order $(2 \succ \ldots \succ m \succ 1)$.
\end{itemize}
The second fraction is positive since $\theta < \tfrac{1}{4}$, and the two fractions clearly sum to~1.

Let $\delta > 0$ and consider a profile~$Q'$ such that 
$d_\infty(Q, Q') < \tfrac{\delta}{m!}$.  
We have:
\[
\begin{aligned}
&W(1, 2, Q') - \tfrac{1}{2} 
    = \frac{\theta}{2} + b(\zeta + \delta), \\
&W(2, j, Q') - \tfrac{1}{2} 
    = \frac{1 + 2\theta}{6} + b(\zeta + \delta),\\
&W(j, 1, Q') - \tfrac{1}{2} 
    = \frac{1 - 4\theta}{6} + b(\zeta + \delta),\\
&W(j, m, Q') - \tfrac{1}{2} 
    = \frac{1 - 3 \theta}{4} + b(\zeta + \delta) 
    \quad \text{(if $m$ is even),}\\
\end{aligned}
\]
For $\zeta = \delta$ small enough, all the quantities above are positive. Moreover:
\[
W(3, k, Q') - \tfrac{1}{2} 
    = \tfrac{1}{4}(1 - \theta) + \theta 
        + \frac{m' - k + 1}{m' - 2} \frac{1 - 3\theta}{2}
        + b(\zeta + \delta).
\]
For $\zeta = \delta$ small enough, this quantity is positive for $k \le \tfrac{3 + m'}{2}$, and negative otherwise.
The (unweighted) majority matrix is therefore as illustrated by Table~\ref{tab:umm_manip_copeland} for the example $m=8$.

\begin{table}[ht]
\centering
	\caption{Unweighted majority matrix $M(Q)$ in the proof of the subcritical regime for Copeland, for $m=8$. Null coefficients are omitted for legibility.}\label{tab:umm_manip_copeland}
	\(\begin{array}{c|c|c|c|c|c|c|c|c|}
		  & 1 & 2 & 3 & 4 & 5 & 6 & 7 & 8 \\ \hline
		1 &   & 1 &   &   &   &   &   &   \\ \hline
		2 &   &   & 1 & 1 & 1 & 1 & 1 & 1 \\ \hline
		3 & 1 &   &   & 1 & 1 &   &   & 1 \\ \hline
		4 & 1 &   &   &   & 1 & 1 &   & 1 \\ \hline
		5 & 1 &   &   &   &   & 1 & 1 & 1 \\ \hline
		6 & 1 &   & 1 &   &   &   & 1 & 1 \\ \hline
		7 & 1 &   & 1 & 1 &   &   &   & 1 \\ \hline
		8 & 1 &   &   &   &   &   &   &   \\ \hline
	\end{array}\)
\end{table}

Note that the submatrix corresponding to candidates $\{3, \ldots, m'\}$ is circulant,  
so that each of these candidates wins against exactly half of the others.  
As a consequence:
\[
\begin{aligned}
&s_\Cop(1, Q') = 1, \\
&s_\Cop(2, Q') = m - 2, \\
&s_\Cop(j, Q') = \left\lfloor \tfrac{m}{2} \right\rfloor
    \;\; \text{for } j \in \{3, \ldots, m'\},\\
&s_\Cop(m, Q') = 1 \quad \text{if $m$ is even.} \\
\end{aligned}
\]
In particular,
\[
s_\Cop(2, Q') - s_\Cop(j, Q') 
    = \left\lceil \tfrac{m}{2} \right\rceil - 2
    > \left\lceil \tfrac{5}{2} \right\rceil - 2 > 0,
\]
hence candidate~2 obtains the highest Copeland score and is therefore elected.  
We conclude by Lemma~\ref{lem_delta_stable_cm}.
\end{proof}

In Section~\ref{sec:rcw} of the main paper, we already observed that for Slater and Copeland with $m=3$, the critical concentration parameter $\theta_c(f,3)$ may range from $0$ to $\tfrac{1}{4}$, depending on the tie-breaking rule. We now provide a more complete analysis of the cases excluded from Theorem~\ref{thm:critical_theta_black_family}: Slater with $m=3$, and Copeland with $m \in \{3, 4\}$. When $m=3$, the Slater and Copeland rules are equivalent; it therefore suffices to study Copeland with $m \in \{3, 4\}$.

Throughout the paper, we often refer simply to the \emph{Copeland rule}, since the parameter~$\alpha$ does not affect our main results.
Recall, however, that in the general definition given in Section~\ref{sec:zoology}, Copeland is parameterized by $\alpha\in[0,1]$, which specifies the additional score awarded for each tied pairwise comparison.
For the small values of~$m$ considered here, this parameter does play a role in the critical concentration parameters, together with the choice of tie-breaking rule.

\begin{proposition}\label{thm:copeland_few_candidates}
Let $f$ be $\alpha$-Copeland with $\alpha \in [0, 1]$ and let $m \in \{3, 4\}$. Depending on the tie-breaking rule, the critical concentration parameters may vary, but they always satisfy the following bounds:
\begin{enumerate}
\item The lower critical concentration parameter satisfies $\theta_\ell(f, m) \in [0, \frac{1}{4}]$.\label{enum:theta_ell}
\item The upper critical concentration parameter satisfies:\label{enum:theta_u}
\begin{enumerate}
\item If $m=3$, or if $m=4$ and $\alpha = 1$, then $\theta_u(f, m) \in[0, \frac{1}{4}]$.\label{enum:theta_u_1}
\item If $m=4$ and $\alpha < 1$, then $\theta_u(f, m) = \frac{1}{4}$.\label{enum:theta_u_2}
\end{enumerate}
\item The values of $\theta_\ell(f, m)$ and $\theta_u(f, m)$ may coincide or be distinct.\label{enum:theta_coincide_or_differ}
\end{enumerate}

Moreover, all the bounds stated above can be attained under suitable tie-breaking rules.

When $m=3$, the same conclusions hold for the Slater rule, since it coincides with Copeland.
\end{proposition}

\begin{proof}[Proof of Proposition~\ref{thm:copeland_few_candidates}]
Since $f$ is Condorcet-consistent, Corollary~\ref{thm:theta_u_if_protected_by_rcw}
yields $\theta_\ell(f,m) \le \theta_u(f,m) \le \tfrac{1}{4}$. This establishes all the upper bounds in items (\ref{enum:theta_ell}), (\ref{enum:theta_u_1}), and~(\ref{enum:theta_u_2}).

Let $T$ be a tie-breaking rule defined as follows: among tied candidates, elect candidate~2 if possible; otherwise, elect an arbitrary candidate.
Consider the unison manipulation leading to the profile~$Q$ described in Section~\ref{sec:analysis_um_2}. 
The Copeland scores satisfy 
$s_\Cop(1, Q) = 1$ (candidate~1 defeats candidate~2), 
$s_\Cop(2, Q) = m-2$ (candidate~2 defeats all candidates except~1 and itself), 
$s_\Cop(3, Q) = m-2$ (candidate~3 defeats all candidates except~2 and itself), 
and, if $m=4$, $s_\Cop(4, Q) = 1$ (candidate~4 defeats candidate~1).
The tie-breaking rule~$T$ therefore selects candidate~2 as the winner.
It follows that $\theta_\ell(f,m) = \theta_u(f,m) = \tfrac{1}{4}$ under~$T$, showing that all the stated upper bounds are attainable.

Let $T'$ be a tie-breaking rule defined as follows: among tied candidates, elect candidate~1 if possible; otherwise, elect a candidate who does not suffer a pairwise defeat against~1 if possible; otherwise, elect an arbitrary candidate.
Although somewhat unconventional, this rule has a natural interpretation in a setting where candidate~1 represents the status quo and the other candidates represent proposals for change: in the event of a tie, the mechanism designer may wish to favor the status quo, or, failing that, a change that is preferred to it by a majority.
We now consider a manipulation attempt to a target profile~$Q$ in favor of some candidate $c \neq 1$.

\begin{itemize}
\item If $m=3$, then after manipulation, candidate~1 still defeats candidate~$c$, so that $s_\Cop(1, Q) \geq 1$ and $s_\Cop(c, Q) \leq 1$. The tie-breaking rule~$T'$ therefore ensures that candidate~$c$ cannot be elected (only the first clause of the tie-breaking rule is relevant in this case). Hence, $\theta_\ell(f,m) = \theta_u(f,m) = 0$.
\item If $m=4$ and $\alpha = 1$, then after manipulation, candidate~$c$ still suffers a defeat against candidate~$1$, and therefore $s_\Cop(c, Q) \leq 2$.
There are at least four integer Copeland points to be distributed among the three remaining candidates, denoted $(1,d,e)$, so at least one of them must receive at least two points.
If this candidate is~$1$, then the tie-breaking rule~$T'$ ensures that candidate~$c$ cannot be elected.
Otherwise, candidate~$1$ is defeated by both $d$ and~$e$, so one of these candidates has at least two points and is favored by the tie-breaking rule~$T'$.
In both cases, the manipulation attempt fails, which implies that $\theta_\ell(f,m) = \theta_u(f,m) = 0$.
\item If $m=4$ and $\alpha < 1$, the parity of $n$ must be taken into account.
If $n$ tends to infinity through odd values, the reasoning above still applies, and any manipulation attempt fails. Therefore, for any $\theta > 0$, we have $\liminf_{n\to\infty} \rho(f,m,n,\theta) = 0$. This implies that $\theta_\ell(f,m) = 0$.
\end{itemize}
In summary, the reasoning above shows that the (trivial) lower bound~$0$ is attainable for $\theta_\ell(f,m)$ in all cases of item~(\ref{enum:theta_ell}), and for $\theta_u(f,m)$ in the two cases of item~(\ref{enum:theta_u_1}).

We now examine the case where $m=4$, $\alpha < 1$, and $n$ tends to infinity through even values, regardless of the tie-breaking rule.
Assume that $\theta < \tfrac{1}{4}$.
We show that there exists a manipulation to a target profile~$Q$ in favor of candidate~$2$.
To begin with, manipulators rank candidate~$2$ first and candidate~$1$ last.
As in the unison manipulation described in Section~\ref{sec:analysis_um_2}, this yields
$s_\Cop(1,Q)=1$ and $s_\Cop(2,Q)=2$.
Next, in the expected profile~$\hat{P}$, Table~\ref{tab:wmm_sincere_voters} shows that in the pairwise comparison between candidates~$3$ and~$4$, each receives a score of at most $\frac{1}{4}(1-\theta)+\theta < \frac{1}{2}$.
Consequently, with high probability, their pairwise scores in the random profile~$P$ are both smaller than $\frac{n}{2}$.
Manipulators can therefore coordinate to enforce a tie between candidates~$3$ and~$4$, resulting in $s_\Cop(3,Q)=s_\Cop(4,Q)=1+\alpha$. Since $\alpha<1$, the manipulation succeeds. Thus, for $\theta < \tfrac{1}{4}$, $\limsup_{n\to\infty}\rho(f,m,n,\theta)=1$. This establishes $\theta_u(f,m)\ge \tfrac{1}{4}$, the lower bound in case~(\ref{enum:theta_u_2}).

We now show that the critical concentration parameters $\theta_\ell(f,m)$ and $\theta_u(f,m)$ may coincide or be distinct. The tie-breaking rule~$T$ yields $\theta_\ell(f,m)=\theta_u(f,m)=\tfrac{1}{4}$, whereas the tie-breaking rule~$T'$ yields $\theta_\ell(f,m)=\theta_u(f,m)=0$ in all cases where this is possible. Moreover, when $m=4$ and $\alpha<1$, we have seen that under~$T'$ the two critical concentration parameters differ, namely $\theta_\ell(f,m)=0$ and $\theta_u(f,m)=\tfrac{1}{4}$. In the general case, it suffices to use~$T'$ when $n$ is odd and~$T$ when $n$ is even to obtain $\theta_\ell(f, m) = 0$ and $\theta_u(f, m) = \frac{1}{4}$. This establishes item~(\ref{enum:theta_coincide_or_differ}).
\end{proof}

For $m=4$ and $\alpha<1$, note that the tie-breaking rule~$T'$ constructed in the proof preserves the homogeneity of~$f$.
It therefore provides an example of a voting rule that does not exhibit a phase transition, as in the counterexample introduced by \citet{durand2025irv},
but with the additional property of being homogeneous.
Moreover, the Copeland rule is standard; only the tie-breaking rule can be regarded as somewhat exotic.

\subsection{Coombs}

\begin{theorem}\label{thm:critical_theta_coombs}
For Coombs, the critical concentration parameter is
\[
\theta_c(\Coo, m) = \frac{m-1}{3m-1}.
\]
\end{theorem}

\begin{proof}[Proof of Theorem~\ref{thm:critical_theta_coombs}]
In the expected profile~$\hat{P}$, candidate~$m$ receives a fraction $\theta + \tfrac{1 - \theta}{m}$ of the vetoes,  
while every other candidate receives only $\tfrac{1 - \theta}{m}$.  
Thus, candidate~$m$ is eliminated first.  
At the next counting round, the same reasoning applies: candidate~$m-1$ is eliminated,  
and the process continues until only candidate~1 remains.  
Hence, candidate~1 is the winner in~$\hat{P}$ and in a neighborhood of it.

Assume that Coombs is coalitionally manipulable from the expected profile~$\hat{P}$ to another profile~$Q$ in favor of some candidate~$c$.  
Let $\{1, c, j_1, \ldots, j_k\}$ be the candidates still in contention at the round where candidate~1 is eliminated.  
Candidate~1 receives no vetoes from the sincere voters by definition, and at most $\tfrac{1}{2}(1 - \theta)$ vetoes from the manipulators.  
If $\max(c, j_1, \ldots, j_k) = j_\ell$ for some $\ell$, that is, if the opponent of candidate~1 of highest index is one $j_\ell$, then this candidate~$j_\ell$ receives $\tfrac{1}{2(k + 2)}(1 - \theta) + \theta$ vetoes from the sincere voters.  
Therefore, we must have
\[
\tfrac{1}{2}(1 - \theta) \ge \frac{1}{2(k + 2)}(1 - \theta) + \theta 
    \ge \frac{1}{2m}(1 - \theta) + \theta,
\]
which simplifies to $\theta \le \tfrac{m - 1}{3m - 1}$.  
On the other hand, if we have $\max(c, j_1, \ldots, j_k) = c$, that is, if the opponent of highest index is $c$, then candidate~$c$ receives $\tfrac{1}{k + 2}(1 - \theta) + \theta$ vetoes from the sincere voters,  
which leads to an even more restrictive condition on $\theta$.  
In either case, if $\theta > \tfrac{m - 1}{3m - 1}$, then Coombs is non-CM in~$\hat{P}$,  
and this property holds in a neighborhood of it.

Now consider the attempt of unison manipulation in favor of candidate~2 with ballot $(2 \succ \ldots \succ m \succ 1)$, as described in Section~\ref{sec:analysis_um_2}.  
At the first round, candidate~1 receives $\tfrac{1}{2}(1 - \theta)$ vetoes,  
candidate~2 receives $\tfrac{1}{m}(1 - \theta)$ (strictly fewer than candidate~1),  
candidate~$m$ receives $\tfrac{1}{2m}(1 - \theta) + \theta$,  
and any candidate $j \in \{3, \ldots, m - 1\}$ receives $\tfrac{1}{2m}(1 - \theta)$ (strictly fewer than candidate~$m$).  
Thus the worst veto score at round 1 is attained by either candidate~1 or~$m$, and only these two candidates need to be compared.\footnote{Note the importance of the assumption $m \ge 3$ here: if $m = 2$, candidates~2 and~$m$ coincide, and candidate~2 then receives $\tfrac{1}{2}(1 - \theta) + \theta$ vetoes.} If $\theta < \tfrac{m - 1}{3m - 1}$, candidate~1 obtains more vetoes than candidate~$m$ and is therefore eliminated.  
At the second round, candidate~$m$ receives more than $\tfrac{1}{2}(1 - \theta) + \theta > \tfrac{1}{2}$ vetoes and is thus eliminated.  
Similarly, in the subsequent counting rounds, the remaining candidates $3, \ldots, m - 1$ (if any, i.e., for $m \ge 4$)  
are eliminated in decreasing order of index, until candidate~2 is declared the winner.  
Therefore, Coombs is unison-manipulable in~$\hat{P}$ and in a neighborhood of it.

We conclude by applying Lemmas~\ref{lem_not_CM_in_neighborhood} and~\ref{lem_um_in_neighborhood}.
\end{proof}

\subsection{Bucklin}\label{sec:proof_bucklin}

While in the main body of the paper, we defined the Bucklin rule based on the median rank for concision, we use here the equivalent and more usual round-based definition of the Bucklin rule. At each round~$t$, the Bucklin score of a candidate~$c$ in profile~$P$ is defined as
\[
s_\Buc^t(c, P) = w(P^{r(c) \leq t}).
\]
If, at some round~$t$, a candidate reaches a score greater than $w(P)/2$,  
the procedure stops and the candidate with the highest current score is declared the winner. This process necessarily terminates, since for every candidate~$c$, we have $s_\Buc^{m(P)}(c, P) = w(P)$.

\begin{theorem}\label{thm:critical_theta_bucklin}
For Bucklin, the critical concentration parameter is
\[
\theta_c(\Buc, m) = \frac{m-2}{2m-2}.
\]
\end{theorem}

\begin{proof}[Proof of Theorem~\ref{thm:critical_theta_bucklin}]
Assume $\theta > \tfrac{m - 2}{2m - 2}$.  
In the expected profile~$\hat{P}$, the first-round score of candidate~1 is
$\theta + \tfrac{1 - \theta}{m} > \tfrac{1}{2}$,  
so candidate~1 is immediately elected.  
Consequently, Bucklin is non-CM,  
since this inequality would remain valid in the manipulated profile.

Now assume $\theta < \tfrac{m - 2}{2m - 2}$.  
In $\hat{P}$, and in a neighborhood of it, candidate~1 does not reach a majority in the first round.  
For now, we leave aside the question of who the eventual winner is.

Let $j$ be a fixed candidate and consider the profile~$Q$ obtained from~$\hat{P}$ by having all voters who prefer candidate~1 to candidate~$j$ cast the ballot $(1 \succ \ldots \succ m)$.  
Then
\[
s_\Buc^1(1, Q) = \theta + \tfrac{1}{2}(1 - \theta) > \tfrac{1}{2},
\]
so candidate~1 wins in~$Q$.  
The same conclusion holds when applying the same transformation to any profile~$P$ in a neighborhood of~$\hat{P}$.

Consider now the profile~$R$ obtained from~$\hat{P}$ by having all voters who prefer candidate~2 over candidate~1 cast the ballot $(2 \succ m \succ \ldots \succ 1)$.  
At round~2, we have:
\[
s_\Buc^2(2, R) = \theta + \frac{1}{m(m-1)}(1 - \theta) + \tfrac{1}{2}(1 - \theta) > \tfrac{1}{2}.
\]
Indeed, there are $(m - 2)!$ permutations in which candidate~2 is in second position and candidate~1 is above her (hence in first position),  
out of $m!$ possible permutations.  
Therefore, candidate~2 attains a strict majority. As for candidate~1, we have:
\[
s_\Buc^2(1, R) = \theta + \frac{1}{m}(1 - \theta) + \frac{m - 2}{m(m - 1)}(1 - \theta).
\]
This is because there are $(m - 2)(m - 2)!$ permutations where candidate~1 is in second position,  
candidate~2 is below her (yielding $m - 2$ possibilities),  
and the remaining candidates appear in any order, out of $m!$ total permutations.  
Hence,
\[
s_\Buc^2(2, R) - s_\Buc^2(1, R) 
    = \frac{m^2 - 5m + 8}{2m(m - 1)}(1 - \theta) > 0.
\]
Finally, for any $j > 2$, 
\[
s_\Buc^2(j, R) = \frac{1}{2m}(1 - \theta) + \frac{1}{2m}(1 - \theta) < s_\Buc^2(2, R),
\]
so candidate~2 is the winner.  
This conclusion also holds when applying the same transformation to any profile~$P$ in a neighborhood of~$\hat{P}$.

Now consider any profile~$P$ in a sufficiently small neighborhood of~$\hat{P}$ so that all the above inequalities hold.  
If $\Buc(P) = 1$, then $\Buc$ is unison-manipulable in favor of candidate~2;  
otherwise, it is unison-manipulable in favor of candidate~1.

We conclude by applying Lemmas~\ref{lem_not_CM_in_neighborhood} and~\ref{lem_um_in_neighborhood}.
\end{proof}

\subsection{Borda}

\begin{theorem}\label{thm:critical_theta_borda}
For Borda, the critical concentration parameter is
\[
\theta_c(\Bor, m) = \frac{m-2}{m+1}.
\]
\end{theorem}

\begin{proof}[Proof of Theorem~\ref{thm:critical_theta_borda}]
Assume that the expected profile~$\hat{P}$ is coalitionally manipulable towards some profile~$Q$ in favor of a candidate~$c$.  
In particular, candidate~$c$ must obtain a higher Borda score than candidate~1 in~$Q$.  
The most favorable case clearly occurs for $c = 2$, when the manipulators place~2 at the top and~1 at the bottom of their rankings,  
as in the unison manipulation described in Section~\ref{sec:analysis_um_2}.  
We then have:
\[
\begin{aligned}
s_\Bor(1, Q) &= \tfrac{1}{2}(1 - \theta) + \theta 
    + (m - 2)\left[\tfrac{1}{3}(1 - \theta) + \theta\right], \\
s_\Bor(2, Q) &= \tfrac{1}{2}(1 - \theta) 
    + (m - 2)\left[\tfrac{2}{3}(1 - \theta) + \theta\right].
\end{aligned}
\]
For the manipulation to succeed, we require $s_\Bor(2, Q) \ge s_\Bor(1, Q)$,  
which simplifies to
\[
\theta \le \frac{m - 2}{m + 1}.
\]
By contraposition, if $\theta > \tfrac{m - 2}{m + 1}$, then Borda is non-CM,  
and this property also holds in a neighborhood of~$\hat{P}$.

\medskip
Assume now that $\theta < \tfrac{m - 2}{m + 1}$.  
We construct a manipulation in which all voters who prefer candidate~2 over candidate~1 cast ballots with candidate~2 first and candidate~1 last.  
By the above calculus, we already know that candidate~2 then obtains a strictly higher score than candidate~1;  
it therefore remains to ensure that candidate~2 also outperforms every other candidate~$j$.

\paragraph{Case $\theta > \tfrac{1}{3}$}

Consider a manipulation attempt~$Q$ in favor of candidate~2,  
where all manipulators cast the common ballot $(2 \succ m \succ \ldots \succ 3 \succ 1)$.  
For any candidate $j \notin \{1, 2\}$, we have:
\[
s_\Bor(j, Q)
    = \frac{1 - \theta}{6} + \frac{1 - \theta}{3} + (m - 3)\frac{1 - \theta}{4}
      + \theta (m - j) + \frac{1 - \theta}{2}(j - 2).
\]
From this, we deduce:
\[
s_\Bor(2, Q) - s_\Bor(j, Q)
    = \tfrac{1}{12} \left[ 5m - 6j + 5 + (-5m + 18j - 29)\theta \right].
\]
The dependence on~$j$ is given by the coefficient $(-6 + 18\theta)$.  
Since $\theta > \tfrac{1}{3}$, this coefficient is positive,  
and the most threatening contender is therefore candidate~3.  
We then have:
\[
s_\Bor(2, Q) - s_\Bor(3, Q)
    = \tfrac{1}{12}\left[ 5m - 13 + (-5m + 25)\theta \right].
\]
This expression is affine in~$\theta$,  
so it suffices to check that it is positive at the extreme values of~$\theta$.  
For $\theta = 0$:
\[
s_\Bor(2, Q) - s_\Bor(3, Q) = \tfrac{1}{12}(5m - 13) > 0.
\]
For $\theta = 1$:
\[
s_\Bor(2, Q) - s_\Bor(3, Q) = 1 > 0.
\]
Hence, the manipulation succeeds.

\paragraph{Case $\theta \leq \tfrac{1}{3}$}
Let $Q$ be the profile obtained from~$\hat{P}$ by having all voters who prefer candidate~2 over candidate~1 modify their ballots as follows:
\begin{itemize}
    \item A fraction~$\alpha$ of them vote $(2 \succ m \succ \ldots \succ 3 \succ 1)$,
    \item A fraction~$(1 - \alpha)$ of them vote $(2 \succ 3 \succ \ldots \succ m \succ 1)$,
\end{itemize}
where $\alpha \in [0, 1]$ will be chosen shortly.  
For any $j \notin \{1, 2\}$, we have:
\[
\begin{aligned}
s_\Bor(j, Q) 
    &= \frac{1 - \theta}{6} + \frac{1 - \theta}{3} 
       + (m - 3)\frac{1 - \theta}{4} + \theta (m - j) \\
    &\quad + \alpha \frac{1 - \theta}{2}(j - 2)
       + (1 - \alpha)\frac{1 - \theta}{2}(m + 1 - j).
\end{aligned}
\]
Choosing $\alpha = \tfrac{1 + \theta}{2(1 - \theta)}$ makes this expression independent of~$j$, with resulting value:
\[
s_\Bor(j, Q) = \tfrac{1}{12}\big[ 6m - 6 - 12\theta \big].
\]
We then obtain:
\[
s_\Bor(2, Q) - s_\Bor(j, Q) 
    = \tfrac{1}{6}\big[ m - 2 + (2m - 1)\theta \big] > 0.
\]
Hence, the manipulation succeeds. It just remains to check that we have $\alpha \leq 1$, which is the case because $\theta \leq \frac{1}{3}$.

We conclude by using Lemmas~\ref{lem_not_CM_in_neighborhood}, \ref{lem_um_in_neighborhood} and \ref{lem_delta_stable_cm}.
\end{proof}

\subsection{Kim--Roush}

\begin{theorem}\label{thm:critical_theta_kim_roush}
For Kim--Roush, the critical concentration parameter is
\[
\theta_c(\KR, m) = \frac{m-2}{m}.
\]
\end{theorem}

\begin{proof}[Proof of Theorem~\ref{thm:critical_theta_kim_roush}]
In the expected profile~$\hat{P}$, candidate~$m$ receives $\theta + \tfrac{1 - \theta}{m}$ vetoes,  
while every other candidate receives only $\tfrac{1 - \theta}{m}$.  
Hence, only candidate~$m$ is eliminated in the first round.  
The same reasoning applies at each subsequent round:  
the candidate with the highest remaining index is eliminated,  
until candidate~1 is declared the winner.

Assume that Kim--Roush is coalitionally manipulable from the expected profile~$\hat{P}$ to some profile~$Q$ in favor of a candidate~$c$.  
Consider the round in which candidate~1 is eliminated.  
Since~$c$ must still be present at that stage, candidate~1 receives no vetoes from the sincere voters  
and at most $\tfrac{1}{2}(1 - \theta)$ vetoes from the manipulators.  
In order to be eliminated, candidate~1 must receive a fraction of vetoes
at least equal to the average. Hence, it must hold that
\[
\tfrac{1}{2}(1 - \theta) \ge \frac{1}{k},
\]
where $k$ denotes the number of remaining candidates. 
In particular, we must have $\tfrac{1}{2}(1 - \theta) \ge \tfrac{1}{m}$,  
which simplifies to $\theta \le \tfrac{m - 2}{m}$.  
Hence, if $\theta > \tfrac{m - 2}{m}$, the Kim--Roush rule is non-CM  
in~$\hat{P}$ and in a neighborhood of it.

Now consider the case where all voters who prefer candidate~2 to candidate~1 cast the ballot $(2 \succ \ldots \succ m \succ 1)$, as described in Section~\ref{sec:analysis_um_2}.  
At the first round, candidate~1 receives $\tfrac{1}{2}(1 - \theta)$ vetoes.  
If $\theta < \tfrac{m - 2}{m}$, this quantity exceeds $\tfrac{1}{m}$,  
so candidate~1 is eliminated.  
Conversely, candidate~2 receives only $\tfrac{1}{m}(1 - \theta) < \tfrac{1}{m}$ vetoes and is therefore not eliminated.  
In the subsequent rounds, since there is always at least one candidate with index greater than~2,  
candidate~2 continues to receive $\tfrac{1}{k}(1 - \theta) < \tfrac{1}{k}$ vetoes  
(where $k$ denotes the number of remaining candidates).  
Hence, candidate~2 is never eliminated and eventually wins.  
The same conclusion holds for any profile in a neighborhood of~$\hat{P}$.

We conclude by using Lemmas \ref{lem_not_CM_in_neighborhood} and \ref{lem_um_in_neighborhood}.
\end{proof}

\subsection{Veto}

\begin{theorem}\label{thm:critical_theta_veto}
For Veto, the critical concentration parameter is
\[
\theta_c(\Vet, m) = 1.
\]
\end{theorem}

\begin{proof}[Proof of Theorem~\ref{thm:critical_theta_veto}]
In a neighborhood of~$\hat{P}$, candidate~$m$ cannot win, while any other candidate may be elected. For now, we leave aside the question of who the actual winner is.

Consider the profile~$Q$ obtained from~$\hat{P}$ by the following transformation.  
Select a fraction~$\theta$ of voters who have the sincere ranking $(1 \succ \ldots \succ m)$ (intuitively, the ``Dirac'' component of the profile), and modify their ballots so that their vetoes are evenly distributed among the $m-1$ other candidates.
By symmetry of the remaining voters, it is then clear that candidate~1 becomes the winner. Note that this argument requires $\theta > 0$.

Consider now a manipulation attempt in favor of candidate~2, performed by voters who prefer candidate~2 to candidate~1.
The contributions of the sincere voters to the Veto scores are as follows:
\[
\begin{aligned}
&s_\Vet(1, \hat{P}^{1 \succ 2}) = 0,\\
&s_\Vet(2, \hat{P}^{1 \succ 2}) = -\frac{1}{m}(1 - \theta),\\
&s_\Vet(j \notin \{1, 2, m\}, \hat{P}^{1 \succ 2}) = -\frac{1}{2m}(1 - \theta),\\
&s_\Vet(m, \hat{P}^{1 \succ 2}) = -\theta - \frac{1}{2m}(1 - \theta).
\end{aligned}
\]
We need a sufficient number of manipulators to compensate for the score differences  
between candidate~2 and the others ($1, 3, \ldots, m$):
\[
\begin{aligned}
\frac{1 - \theta}{2} >\;& 
(s_\Vet(1, \hat{P}^{1 \succ 2}) - s_\Vet(2, \hat{P}^{1 \succ 2})) \\
&+ (m - 3)\big(s_\Vet(3, \hat{P}^{1 \succ 2}) - s_\Vet(2, \hat{P}^{1 \succ 2})\big) \\
&+ \max\big(0,\, s_\Vet(m, \hat{P}^{1 \succ 2}) - s_\Vet(2, \hat{P}^{1 \succ 2})\big),
\end{aligned}
\]
which simplifies to
\[
\theta > 0.
\]
This inequality always holds,  
so it is possible to distribute the manipulators’ vetoes in such a way  
that candidate~2 obtains the highest score.

Now consider any profile~$P$ in a sufficiently small neighborhood of~$\hat{P}$.  
If $\Vet(P) = 1$, then Veto is coalitionally manipulable in favor of candidate~2;  
otherwise, it is manipulable in favor of candidate~1.  
In both cases, one can establish $\delta$-stable manipulability using the same arguments as in the previous proofs.  

We conclude by applying Lemma~\ref{lem_delta_stable_cm}.
\end{proof}

Let us temporarily allow the case $\theta = 0$, corresponding to Impartial Culture.
In this model, \citet{kim1996manipulability} showed that the common limiting CM rate of Kim--Roush and Veto is strictly less than~1.
However, in both cases, we have established that the limiting CM rate is~1 for $\theta \in (0, \theta_c(f, m))$, which is a non-empty interval since $\theta_c(f, m) > 0$.
Each of these two rules thus has the surprising property that its limiting CM rate is not a decreasing function of~$\theta$, because of their singular behavior when $\theta=0$.

\medskip
Taken together, the results of Appendix~\ref{sec:appendix_proof_of_theoretical_results} establish the critical concentration parameters for all voting rules considered in this work, as summarized in Theorem~\ref{thm:critical_thetas}, and additionally show that $\theta_c(f, m) = \frac{1}{7}$ for Split Cycle and Viennot.

%\ThmCriticalThetas*

\section{Additional Figures}\label{sec:appendix_additional_figures}

\renewcommand{\axisWidth}{14cm}
\renewcommand{\axisHeight}{8.1cm}
\begin{figure}[t]
\begin{subfigure}[t]{\linewidth}
  \centering
  \input{netflix_more_rules_cm_rate_bar_plot.tex}
  \caption{Netflix dataset.}
  \label{fig:netflix_cm_rate_bar_plot_more}
\end{subfigure}

\bigskip\bigskip
\begin{subfigure}[t]{\linewidth}
  \centering
  \input{fairvote_more_rules_cm_rate_bar_plot.tex}
  \caption{FairVote dataset.}
  \label{fig:fairvote_cm_rate_bar_plot_more_rules}
\end{subfigure}

\caption{CM rates for all values of $m$. Solid bars show the fraction of coalitionally manipulable profiles, with thin vertical lines indicating algorithmic uncertainty. Colors group rules by critical concentration parameter $\theta_c(f,m)$ (Table~\ref{tab:theo_results}). Dashed horizontal lines indicate, from top to bottom, the fraction of profiles without a Resistant, Set-Safe, Pair-Safe, or Super Condorcet winner.}
\label{fig:cm_rate_bar_plot}
\end{figure}
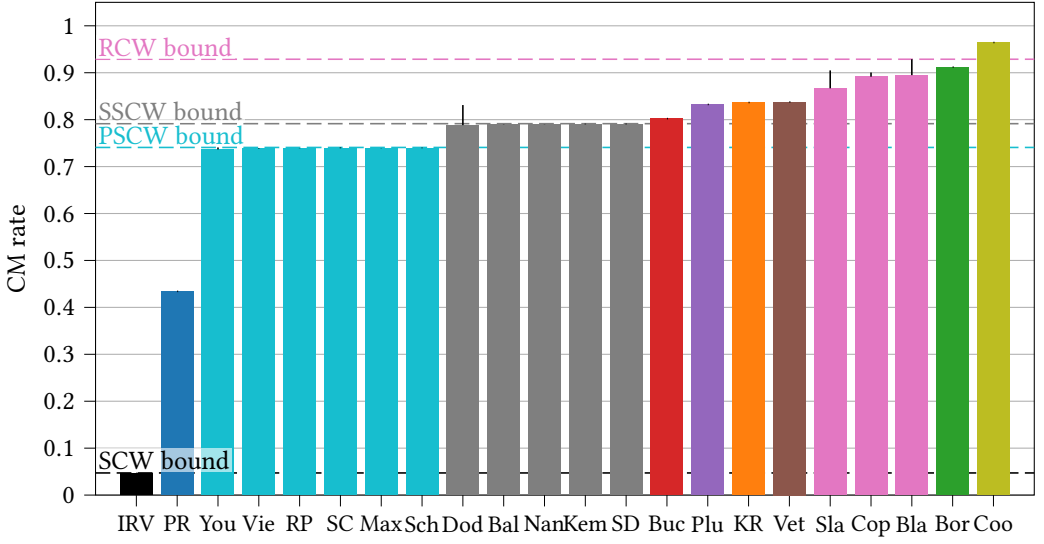
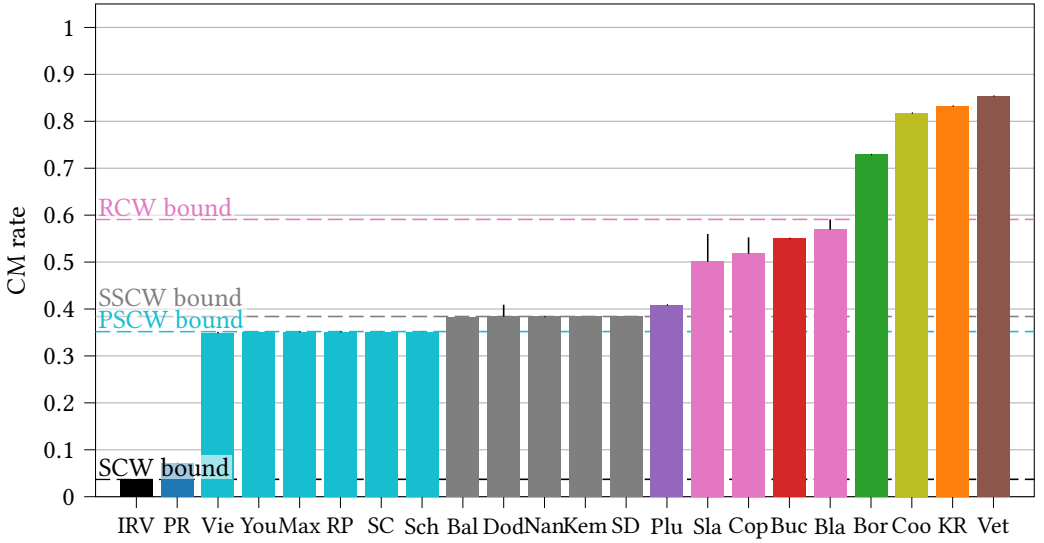
\renewcommand{\axisWidth}{\axisWidthMemory}
\renewcommand{\axisHeight}{\axisHeightMemory}

Figures~\ref{fig:netflix_cm_rate_bar_plot_more} and~\ref{fig:fairvote_cm_rate_bar_plot_more_rules} revisit Figures~\ref{fig:netflix_cm_rate_bar_plot} and~\ref{fig:fairvote_cm_rate_bar_plot}.
They again report overall CM rates, but now include all values of~$m$, rather than being restricted to $m \geq 5$.
In addition, they incorporate the two voting rules analyzed only in the appendix,
namely Split Cycle and Viennot.

First, we observe that for Slater and Copeland, the CM rates are less closely aligned with the RCW bound and with the CM rate of Black.
This behavior was expected, since our theoretical results do not apply to Slater when $m=3$, nor to Copeland when $m \in \{3,4\}$.

Split Cycle empirically confirms its classification within the Maximin family.

The same observation holds for Viennot, and this case is particularly interesting.
Indeed, Viennot shares the same critical concentration parameter as the other rules in the family, but not the same structural connection with the PSCW notion.
This further supports the predictive power of the critical concentration parameter.

\end{document}

%% file: commands.tex
\usepackage{subcaption}
\usepackage{etoc}
\usepackage{ifthen}
\usepackage{thmtools}

\usepackage{pgfplots}
\pgfplotsset{compat=1.18}
\newcommand{\axisWidth}{7.5cm}
\newcommand{\axisWidthMemory}{7.5cm}
\newcommand{\axisSmallerWidth}{7.5cm}
\newcommand{\axisHeight}{5.8cm}
\newcommand{\axisHeightMemory}{5.8cm}
\newcommand{\legendFont}{\small}

\usepackage{microtype}

\newcommand{\Bal}{\textnormal{Bal}}

\newcommand{\Bor}{\textnormal{Bor}}
\newcommand{\Buc}{\textnormal{Buc}}
\newcommand{\Coo}{\textnormal{Coo}}
\newcommand{\Cop}{\textnormal{Cop}}
\newcommand{\Dod}{\textnormal{Dod}}
\newcommand{\IRV}{\textnormal{IRV}}

\newcommand{\Kem}{\textnormal{Kem}}
\newcommand{\KR}{\textnormal{KR}}
\newcommand{\Max}{\textnormal{Max}}

\newcommand{\Nan}{\textnormal{Nan}}
\newcommand{\Plu}{\textnormal{Plu}}
\newcommand{\PR}{\textnormal{PR}}
\newcommand{\RP}{\textnormal{RP}}

\newcommand{\Sla}{\textnormal{Sla}}
\newcommand{\Sch}{\textnormal{Sch}}
\newcommand{\SC}{\textnormal{SC}}
\newcommand{\SD}{\textnormal{SD}}

\newcommand{\Vet}{\textnormal{Vet}}
\newcommand{\Vie}{\textnormal{Vie}}
\newcommand{\You}{\textnormal{You}}

\newcommand{\Strength}{\textnormal{Strength}}

\newcommand{\PSCW}{\textnormal{PSCW}}
\newcommand{\SSCW}{\textnormal{SSCW}}
\newcommand{\RCW}{\textnormal{RCW}}

%% file: cm_rate_of_theta_and_n_v_plu.tex
% This file was created with tikzplotlib v0.10.1.
\begin{tikzpicture}

\definecolor{crimson2143940}{RGB}{214,39,40}
\definecolor{darkgrey176}{RGB}{176,176,176}
\definecolor{darkorange25512714}{RGB}{255,127,14}
\definecolor{darkturquoise23190207}{RGB}{23,190,207}
\definecolor{forestgreen4416044}{RGB}{44,160,44}
\definecolor{goldenrod18818934}{RGB}{188,189,34}
\definecolor{grey127}{RGB}{127,127,127}
\definecolor{lightgrey204}{RGB}{204,204,204}
\definecolor{mediumpurple148103189}{RGB}{148,103,189}
\definecolor{orchid227119194}{RGB}{227,119,194}
\definecolor{sienna1408675}{RGB}{140,86,75}
\definecolor{steelblue31119180}{RGB}{31,119,180}

\begin{axis}[reverse legend,
height=\axisHeight,
legend cell align={left},
legend style={font=\legendFont, fill opacity=1, draw opacity=1, text opacity=1, draw=lightgrey204},
tick align=outside,
tick pos=left,
width=\axisWidth,
x grid style={darkgrey176},
xlabel={$\theta$},
xmin=0, xmax=1,
xtick style={color=black},
y grid style={darkgrey176},
ylabel={$\rho(\text{Plu}, 4, n, \theta)$},
ymin=-0.05, ymax=1.05,
ytick={-0.2, 0.0, 0.2, 0.4000000000000001, 0.6000000000000001, 0.8, 1.0000000000000002, 1.2000000000000002},
ytick style={color=black}
]
\addplot [semithick, steelblue31119180]
table {%
0 0
0.01 0
0.02 0
0.03 0
0.04 0
0.05 0
0.06 0
0.07 0
0.08 0
0.09 0
0.1 0
0.11 0
0.12 0
0.13 0
0.14 0
0.15 0
0.16 0
0.17 0
0.18 0
0.19 0
0.2 0
0.21 0
0.22 0
0.23 0
0.24 0
0.25 0
0.26 0
0.27 0
0.28 0
0.29 0
0.3 0
0.31 0
0.32 0
0.33 0
0.34 0
0.35 0
0.36 0
0.37 0
0.38 0
0.39 0
0.4 0
0.41 0
0.42 0
0.43 0
0.44 0
0.45 0
0.46 0
0.47 0
0.48 0
0.49 0
0.5 0
0.51 0
0.52 0
0.53 0
0.54 0
0.55 0
0.56 0
0.57 0
0.58 0
0.59 0
0.6 0
0.61 0
0.62 0
0.63 0
0.64 0
0.65 0
0.66 0
0.67 0
0.68 0
0.69 0
0.7 0
0.71 0
0.72 0
0.73 0
0.74 0
0.75 0
0.76 0
0.77 0
0.78 0
0.79 0
0.8 0
0.81 0
0.82 0
0.83 0
0.84 0
0.85 0
0.86 0
0.87 0
0.88 0
0.89 0
0.9 0
0.91 0
0.92 0
0.93 0
0.94 0
0.95 0
0.96 0
0.97 0
0.98 0
0.99 0
1 0
};
\addlegendentry{$n = 1$}
\addplot [semithick, darkorange25512714]
table {%
0 0.207732
0.01 0.204175
0.02 0.199946
0.03 0.195708
0.04 0.192215
0.05 0.187506
0.06 0.184117
0.07 0.180447
0.08 0.177001
0.09 0.172327
0.1 0.16927
0.11 0.164601
0.12 0.160935
0.13 0.15768
0.14 0.15455
0.15 0.15061
0.16 0.147011
0.17 0.143495
0.18 0.13991
0.19 0.136604
0.2 0.133536
0.21 0.130266
0.22 0.126442
0.23 0.123513
0.24 0.120831
0.25 0.117202
0.26 0.114439
0.27 0.111191
0.28 0.108487
0.29 0.105033
0.3 0.101291
0.31 0.099283
0.32 0.09597
0.33 0.093233
0.34 0.090576
0.35 0.087931
0.36 0.085217
0.37 0.082931
0.38 0.079866
0.39 0.077803
0.4 0.075076
0.41 0.071843
0.42 0.070146
0.43 0.067456
0.44 0.065254
0.45 0.063015
0.46 0.061013
0.47 0.058269
0.48 0.056407
0.49 0.054366
0.5 0.05168
0.51 0.050222
0.52 0.048341
0.53 0.046402
0.54 0.043922
0.55 0.041898
0.56 0.039777
0.57 0.03826
0.58 0.036545
0.59 0.03498
0.6 0.033385
0.61 0.031669
0.62 0.030263
0.63 0.028589
0.64 0.027034
0.65 0.02552
0.66 0.024157
0.67 0.022534
0.68 0.02128
0.69 0.019835
0.7 0.018815
0.71 0.017478
0.72 0.01637
0.73 0.015174
0.74 0.013953
0.75 0.01324
0.76 0.011859
0.77 0.011133
0.78 0.009983
0.79 0.009217
0.8 0.008301
0.81 0.007558
0.82 0.006872
0.83 0.006035
0.84 0.005299
0.85 0.004554
0.86 0.004148
0.87 0.003461
0.88 0.002978
0.89 0.002476
0.9 0.00206
0.91 0.001644
0.92 0.001368
0.93 0.000965
0.94 0.00073
0.95 0.00053
0.96 0.000332
0.97 0.000174
0.98 9.2e-05
0.99 2e-05
1 0
};
\addlegendentry{$n = 2$}
\addplot [semithick, forestgreen4416044]
table {%
0 0.368973
0.01 0.366565
0.02 0.363367
0.03 0.360834
0.04 0.35691
0.05 0.353492
0.06 0.350395
0.07 0.345719
0.08 0.341126
0.09 0.336075
0.1 0.332154
0.11 0.327438
0.12 0.322577
0.13 0.316439
0.14 0.311059
0.15 0.307323
0.16 0.30173
0.17 0.295447
0.18 0.289788
0.19 0.283882
0.2 0.277534
0.21 0.271726
0.22 0.265731
0.23 0.258926
0.24 0.253573
0.25 0.247583
0.26 0.240423
0.27 0.235092
0.28 0.227814
0.29 0.222189
0.3 0.21539
0.31 0.209343
0.32 0.203722
0.33 0.197454
0.34 0.191455
0.35 0.184644
0.36 0.178151
0.37 0.17231
0.38 0.166726
0.39 0.160415
0.4 0.154802
0.41 0.148571
0.42 0.143288
0.43 0.137241
0.44 0.132106
0.45 0.12655
0.46 0.121315
0.47 0.115329
0.48 0.110502
0.49 0.105005
0.5 0.099915
0.51 0.095001
0.52 0.090863
0.53 0.085928
0.54 0.081332
0.55 0.076918
0.56 0.072642
0.57 0.068851
0.58 0.064197
0.59 0.060664
0.6 0.056856
0.61 0.053596
0.62 0.04981
0.63 0.045791
0.64 0.043308
0.65 0.039775
0.66 0.037031
0.67 0.033962
0.68 0.031473
0.69 0.029055
0.7 0.026353
0.71 0.023782
0.72 0.021663
0.73 0.01964
0.74 0.017938
0.75 0.015683
0.76 0.014361
0.77 0.012556
0.78 0.011076
0.79 0.009863
0.8 0.008604
0.81 0.007387
0.82 0.006346
0.83 0.005411
0.84 0.004428
0.85 0.003773
0.86 0.0032
0.87 0.002519
0.88 0.001902
0.89 0.001473
0.9 0.001177
0.91 0.000819
0.92 0.000605
0.93 0.000415
0.94 0.000253
0.95 0.000122
0.96 7.3e-05
0.97 2.9e-05
0.98 1e-05
0.99 1e-06
1 0
};
\addlegendentry{$n = 4$}
\addplot [semithick, crimson2143940]
table {%
0 0.589223
0.01 0.586979
0.02 0.585222
0.03 0.581738
0.04 0.57882
0.05 0.573567
0.06 0.5682
0.07 0.562479
0.08 0.556307
0.09 0.549662
0.1 0.542904
0.11 0.534304
0.12 0.524903
0.13 0.517299
0.14 0.508149
0.15 0.498124
0.16 0.487872
0.17 0.477992
0.18 0.467986
0.19 0.455803
0.2 0.444581
0.21 0.432529
0.22 0.421017
0.23 0.410596
0.24 0.397331
0.25 0.384923
0.26 0.372543
0.27 0.36113
0.28 0.34804
0.29 0.335284
0.3 0.32384
0.31 0.311101
0.32 0.29856
0.33 0.285986
0.34 0.273976
0.35 0.262312
0.36 0.250528
0.37 0.238124
0.38 0.226434
0.39 0.215721
0.4 0.204777
0.41 0.194028
0.42 0.18333
0.43 0.173126
0.44 0.163162
0.45 0.153248
0.46 0.143695
0.47 0.134924
0.48 0.126386
0.49 0.117006
0.5 0.109051
0.51 0.101492
0.52 0.094558
0.53 0.086842
0.54 0.08056
0.55 0.073986
0.56 0.067528
0.57 0.061878
0.58 0.056555
0.59 0.051487
0.6 0.046658
0.61 0.042036
0.62 0.037858
0.63 0.034161
0.64 0.030226
0.65 0.027224
0.66 0.023823
0.67 0.021017
0.68 0.018637
0.69 0.015975
0.7 0.014217
0.71 0.011976
0.72 0.010408
0.73 0.00878
0.74 0.007499
0.75 0.006311
0.76 0.005115
0.77 0.004399
0.78 0.003653
0.79 0.002865
0.8 0.002353
0.81 0.001811
0.82 0.001496
0.83 0.001057
0.84 0.00078
0.85 0.000631
0.86 0.000484
0.87 0.000276
0.88 0.000194
0.89 0.000127
0.9 9.2e-05
0.91 6.4e-05
0.92 3.4e-05
0.93 2.2e-05
0.94 4e-06
0.95 6e-06
0.96 1e-06
0.97 0
0.98 0
0.99 0
1 0
};
\addlegendentry{$n = 8$}
\addplot [semithick, mediumpurple148103189]
table {%
0 0.794358
0.01 0.7923
0.02 0.788882
0.03 0.784128
0.04 0.778883
0.05 0.772479
0.06 0.764058
0.07 0.755292
0.08 0.745125
0.09 0.732845
0.1 0.722153
0.11 0.708055
0.12 0.693631
0.13 0.677565
0.14 0.662609
0.15 0.644964
0.16 0.626332
0.17 0.607244
0.18 0.589451
0.19 0.568264
0.2 0.549275
0.21 0.527985
0.22 0.507383
0.23 0.487049
0.24 0.4653
0.25 0.444741
0.26 0.422981
0.27 0.401507
0.28 0.381003
0.29 0.359236
0.3 0.339184
0.31 0.320074
0.32 0.299286
0.33 0.280044
0.34 0.261541
0.35 0.243414
0.36 0.225829
0.37 0.210099
0.38 0.192963
0.39 0.17813
0.4 0.163016
0.41 0.149005
0.42 0.136292
0.43 0.124024
0.44 0.112302
0.45 0.101748
0.46 0.091005
0.47 0.082101
0.48 0.073172
0.49 0.064549
0.5 0.057848
0.51 0.051089
0.52 0.044959
0.53 0.03884
0.54 0.034276
0.55 0.029673
0.56 0.025434
0.57 0.021885
0.58 0.018579
0.59 0.01597
0.6 0.013487
0.61 0.011089
0.62 0.009451
0.63 0.007801
0.64 0.00654
0.65 0.00544
0.66 0.004292
0.67 0.003371
0.68 0.002715
0.69 0.002107
0.7 0.001643
0.71 0.001352
0.72 0.00103
0.73 0.000773
0.74 0.000555
0.75 0.000408
0.76 0.0003
0.77 0.000234
0.78 0.00014
0.79 9.8e-05
0.8 7.9e-05
0.81 5.2e-05
0.82 3.3e-05
0.83 1.5e-05
0.84 1e-05
0.85 7e-06
0.86 3e-06
0.87 1e-06
0.88 0
0.89 0
0.9 0
0.91 1e-06
0.92 0
0.93 0
0.94 0
0.95 0
0.96 0
0.97 0
0.98 0
0.99 0
1 0
};
\addlegendentry{$n = 16$}
\addplot [semithick, sienna1408675]
table {%
0 0.936218
0.01 0.935184
0.02 0.932124
0.03 0.927985
0.04 0.921957
0.05 0.915007
0.06 0.90574
0.07 0.895191
0.08 0.882736
0.09 0.868221
0.1 0.851553
0.11 0.83342
0.12 0.81295
0.13 0.792013
0.14 0.767657
0.15 0.743074
0.16 0.715169
0.17 0.686861
0.18 0.657805
0.19 0.626563
0.2 0.595789
0.21 0.563472
0.22 0.531347
0.23 0.498495
0.24 0.465848
0.25 0.433702
0.26 0.401707
0.27 0.370892
0.28 0.339185
0.29 0.310865
0.3 0.282981
0.31 0.254935
0.32 0.230045
0.33 0.206712
0.34 0.183655
0.35 0.16246
0.36 0.143566
0.37 0.125538
0.38 0.109728
0.39 0.095611
0.4 0.082248
0.41 0.070059
0.42 0.060079
0.43 0.050567
0.44 0.042574
0.45 0.035646
0.46 0.029089
0.47 0.023965
0.48 0.019406
0.49 0.015988
0.5 0.0127
0.51 0.010187
0.52 0.007996
0.53 0.006223
0.54 0.004873
0.55 0.003655
0.56 0.002883
0.57 0.002209
0.58 0.001534
0.59 0.001219
0.6 0.000821
0.61 0.000613
0.62 0.000433
0.63 0.000277
0.64 0.000223
0.65 0.000148
0.66 0.000112
0.67 6.4e-05
0.68 4.6e-05
0.69 3.1e-05
0.7 1.9e-05
0.71 1.1e-05
0.72 9e-06
0.73 1e-06
0.74 2e-06
0.75 3e-06
0.76 1e-06
0.77 1e-06
0.78 0
0.79 0
0.8 0
0.81 0
0.82 0
0.83 0
0.84 0
0.85 0
0.86 0
0.87 0
0.88 0
0.89 0
0.9 0
0.91 0
0.92 0
0.93 0
0.94 0
0.95 0
0.96 0
0.97 0
0.98 0
0.99 0
1 0
};
\addlegendentry{$n = 32$}
\addplot [semithick, orchid227119194]
table {%
0 0.992108
0.01 0.991459
0.02 0.990203
0.03 0.988416
0.04 0.985248
0.05 0.981105
0.06 0.975577
0.07 0.968034
0.08 0.957749
0.09 0.946067
0.1 0.930737
0.11 0.913792
0.12 0.893012
0.13 0.867962
0.14 0.840927
0.15 0.809798
0.16 0.77555
0.17 0.738084
0.18 0.697445
0.19 0.655613
0.2 0.611995
0.21 0.5652
0.22 0.519436
0.23 0.47254
0.24 0.426761
0.25 0.381237
0.26 0.33815
0.27 0.297511
0.28 0.259133
0.29 0.223252
0.3 0.190384
0.31 0.159965
0.32 0.134172
0.33 0.111047
0.34 0.090586
0.35 0.073832
0.36 0.05825
0.37 0.046151
0.38 0.03607
0.39 0.027701
0.4 0.021318
0.41 0.01587
0.42 0.011742
0.43 0.008539
0.44 0.006313
0.45 0.004532
0.46 0.003214
0.47 0.002125
0.48 0.001593
0.49 0.00107
0.5 0.000677
0.51 0.00041
0.52 0.000267
0.53 0.000163
0.54 0.000114
0.55 6.6e-05
0.56 3.4e-05
0.57 2.6e-05
0.58 1.9e-05
0.59 8e-06
0.6 7e-06
0.61 2e-06
0.62 2e-06
0.63 1e-06
0.64 0
0.65 0
0.66 0
0.67 0
0.68 0
0.69 0
0.7 0
0.71 0
0.72 0
0.73 0
0.74 0
0.75 0
0.76 0
0.77 0
0.78 0
0.79 0
0.8 0
0.81 0
0.82 0
0.83 0
0.84 0
0.85 0
0.86 0
0.87 0
0.88 0
0.89 0
0.9 0
0.91 0
0.92 0
0.93 0
0.94 0
0.95 0
0.96 0
0.97 0
0.98 0
0.99 0
1 0
};
\addlegendentry{$n = 64$}
\addplot [semithick, grey127]
table {%
0 0.999836
0.01 0.999776
0.02 0.999662
0.03 0.999369
0.04 0.998928
0.05 0.998151
0.06 0.996847
0.07 0.994362
0.08 0.990959
0.09 0.985436
0.1 0.977492
0.11 0.966516
0.12 0.951367
0.13 0.931379
0.14 0.905278
0.15 0.873772
0.16 0.834883
0.17 0.789698
0.18 0.738013
0.19 0.680977
0.2 0.619712
0.21 0.555958
0.22 0.489951
0.23 0.424132
0.24 0.361902
0.25 0.30226
0.26 0.248279
0.27 0.200416
0.28 0.157577
0.29 0.122175
0.3 0.092011
0.31 0.068409
0.32 0.049453
0.33 0.034961
0.34 0.024136
0.35 0.0163
0.36 0.010678
0.37 0.00685
0.38 0.004314
0.39 0.002714
0.4 0.001585
0.41 0.000918
0.42 0.000493
0.43 0.000286
0.44 0.000155
0.45 8.8e-05
0.46 4.4e-05
0.47 2.7e-05
0.48 6e-06
0.49 6e-06
0.5 4e-06
0.51 0
0.52 0
0.53 0
0.54 0
0.55 0
0.56 0
0.57 0
0.58 0
0.59 0
0.6 0
0.61 0
0.62 0
0.63 0
0.64 0
0.65 0
0.66 0
0.67 0
0.68 0
0.69 0
0.7 0
0.71 0
0.72 0
0.73 0
0.74 0
0.75 0
0.76 0
0.77 0
0.78 0
0.79 0
0.8 0
0.81 0
0.82 0
0.83 0
0.84 0
0.85 0
0.86 0
0.87 0
0.88 0
0.89 0
0.9 0
0.91 0
0.92 0
0.93 0
0.94 0
0.95 0
0.96 0
0.97 0
0.98 0
0.99 0
1 0
};
\addlegendentry{$n = 128$}
\addplot [semithick, goldenrod18818934]
table {%
0 1
0.01 1
0.02 0.999999
0.03 0.999998
0.04 0.999991
0.05 0.999971
0.06 0.999914
0.07 0.999769
0.08 0.999366
0.09 0.998567
0.1 0.996778
0.11 0.993363
0.12 0.987118
0.13 0.976485
0.14 0.959255
0.15 0.934035
0.16 0.897152
0.17 0.847968
0.18 0.784636
0.19 0.710223
0.2 0.626187
0.21 0.535256
0.22 0.441716
0.23 0.353059
0.24 0.27161
0.25 0.201722
0.26 0.142925
0.27 0.097545
0.28 0.063615
0.29 0.039867
0.3 0.023883
0.31 0.01354
0.32 0.007532
0.33 0.003935
0.34 0.001896
0.35 0.000911
0.36 0.000442
0.37 0.000176
0.38 6.2e-05
0.39 2.6e-05
0.4 7e-06
0.41 2e-06
0.42 1e-06
0.43 0
0.44 0
0.45 1e-06
0.46 0
0.47 0
0.48 0
0.49 0
0.5 0
0.51 0
0.52 0
0.53 0
0.54 0
0.55 0
0.56 0
0.57 0
0.58 0
0.59 0
0.6 0
0.61 0
0.62 0
0.63 0
0.64 0
0.65 0
0.66 0
0.67 0
0.68 0
0.69 0
0.7 0
0.71 0
0.72 0
0.73 0
0.74 0
0.75 0
0.76 0
0.77 0
0.78 0
0.79 0
0.8 0
0.81 0
0.82 0
0.83 0
0.84 0
0.85 0
0.86 0
0.87 0
0.88 0
0.89 0
0.9 0
0.91 0
0.92 0
0.93 0
0.94 0
0.95 0
0.96 0
0.97 0
0.98 0
0.99 0
1 0
};
\addlegendentry{$n = 256$}
\addplot [semithick, darkturquoise23190207]
table {%
0 1
0.2 1
0.2 0
1 0
};
\addlegendentry{$n \rightarrow \infty$}
\end{axis}

\end{tikzpicture}

%% file: netflix_selected_rules_cm_rate_bar_plot_n_c_ge=5.tex
% This file was created with tikzplotlib v0.10.1.
\begin{tikzpicture}

\definecolor{crimson2143940}{RGB}{214,39,40}
\definecolor{darkgrey176}{RGB}{176,176,176}
\definecolor{darkorange25512714}{RGB}{255,127,14}
\definecolor{darkturquoise23190207}{RGB}{23,190,207}
\definecolor{forestgreen4416044}{RGB}{44,160,44}
\definecolor{goldenrod18818934}{RGB}{188,189,34}
\definecolor{grey127}{RGB}{127,127,127}
\definecolor{mediumpurple148103189}{RGB}{148,103,189}
\definecolor{orchid227119194}{RGB}{227,119,194}
\definecolor{sienna1408675}{RGB}{140,86,75}
\definecolor{steelblue31119180}{RGB}{31,119,180}

\begin{axis}[
height=\axisHeight,
tick align=outside,
tick pos=left,
width=\axisWidth,
x grid style={darkgrey176},
xmin=-1, xmax=20,
xtick style={color=black},
xtick={0,1,2,3,4,5,6,7,8,9,10,11,12,13,14,15,16,17,18,19},
xticklabels={IRV,PR,You,RP,Max,Sch,Dod,Bal,Nan,Kem,SD,Buc,Cop,Sla,Plu,Bla,Bor,Vet,KR,Coo},
xtick={0, 1, 2, 3, 4, 5, 6, 7, 8, 9, 10, 11, 12, 13, 14, 15, 16, 17, 18, 19},
xticklabels = {IRV, PR, You, RP, Max, Sch, Dod, Bal, Nan, Kem, SD, Buc, Cop, Sla, Plu, Bla, Bor, Vet, KR, Coo},
xticklabel style={font=\small},
y grid style={darkgrey176},
ylabel={CM rate},
ymajorgrids,
ymin=0, ymax=1.05,
ytick={0.0, 0.1, 0.2, 0.30000000000000004, 0.4, 0.5, 0.6000000000000001, 0.7000000000000001, 0.8, 0.9, 1.0},
ytick style={color=black}
]
\path [draw=orchid227119194, semithick, dash pattern=on 5.55pt off 2.4pt]
(axis cs:-1,0.970852428964253)
--(axis cs:20,0.970852428964253);

\path [draw=darkturquoise23190207, semithick, dash pattern=on 5.55pt off 2.4pt]
(axis cs:-1,0.842713107241063)
--(axis cs:20,0.842713107241063);

\path [draw=grey127, semithick, dash pattern=on 5.55pt off 2.4pt]
(axis cs:-1,0.912740604949588)
--(axis cs:20,0.912740604949588);

\path [draw=black, semithick, dash pattern=on 5.55pt off 2.4pt]
(axis cs:-1,0.0691109074243813)
--(axis cs:20,0.0691109074243813);

\draw[draw=none,fill=black] (axis cs:-0.4,0) rectangle (axis cs:0.4,0.0687442713107241);
\draw[draw=none,fill=steelblue31119180] (axis cs:0.6,0) rectangle (axis cs:1.4,0.767552703941338);
\draw[draw=none,fill=darkturquoise23190207] (axis cs:1.6,0) rectangle (axis cs:2.4,0.838496791934006);
\draw[draw=none,fill=darkturquoise23190207] (axis cs:2.6,0) rectangle (axis cs:3.4,0.842163153070577);
\draw[draw=none,fill=darkturquoise23190207] (axis cs:3.6,0) rectangle (axis cs:4.4,0.842346471127406);
\draw[draw=none,fill=darkturquoise23190207] (axis cs:4.6,0) rectangle (axis cs:5.4,0.842346471127406);
\draw[draw=none,fill=grey127] (axis cs:5.6,0) rectangle (axis cs:6.4,0.909440879926673);
\draw[draw=none,fill=grey127] (axis cs:6.6,0) rectangle (axis cs:7.4,0.91237396883593);
\draw[draw=none,fill=grey127] (axis cs:7.6,0) rectangle (axis cs:8.4,0.912557286892759);
\draw[draw=none,fill=grey127] (axis cs:8.6,0) rectangle (axis cs:9.4,0.912557286892759);
\draw[draw=none,fill=grey127] (axis cs:9.6,0) rectangle (axis cs:10.4,0.912740604949588);
\draw[draw=none,fill=crimson2143940] (axis cs:10.6,0) rectangle (axis cs:11.4,0.944087992667278);
\draw[draw=none,fill=orchid227119194] (axis cs:11.6,0) rectangle (axis cs:12.4,0.960036663611366);
\draw[draw=none,fill=orchid227119194] (axis cs:12.6,0) rectangle (axis cs:13.4,0.965352887259395);
\draw[draw=none,fill=mediumpurple148103189] (axis cs:13.6,0) rectangle (axis cs:14.4,0.96865261228231);
\draw[draw=none,fill=orchid227119194] (axis cs:14.6,0) rectangle (axis cs:15.4,0.970669110907424);
\draw[draw=none,fill=forestgreen4416044] (axis cs:15.6,0) rectangle (axis cs:16.4,0.990100824931256);
\draw[draw=none,fill=sienna1408675] (axis cs:16.6,0) rectangle (axis cs:17.4,0.997066911090742);
\draw[draw=none,fill=darkorange25512714] (axis cs:17.6,0) rectangle (axis cs:18.4,0.9974335472044);
\draw[draw=none,fill=goldenrod18818934] (axis cs:18.6,0) rectangle (axis cs:19.4,0.998900091659028);
\path [draw=black, semithick]
(axis cs:0,0.0687442713107241)
--(axis cs:0,0.0687442713107241);

\path [draw=black, semithick]
(axis cs:1,0.767552703941338)
--(axis cs:1,0.767552703941338);

\path [draw=black, semithick]
(axis cs:2,0.838496791934006)
--(axis cs:2,0.842713107241063);

\path [draw=black, semithick]
(axis cs:3,0.842163153070577)
--(axis cs:3,0.842346471127406);

\path [draw=black, semithick]
(axis cs:4,0.842346471127406)
--(axis cs:4,0.842346471127406);

\path [draw=black, semithick]
(axis cs:5,0.842346471127406)
--(axis cs:5,0.842713107241063);

\path [draw=black, semithick]
(axis cs:6,0.909440879926673)
--(axis cs:6,0.963703024747938);

\path [draw=black, semithick]
(axis cs:7,0.91237396883593)
--(axis cs:7,0.912557286892759);

\path [draw=black, semithick]
(axis cs:8,0.912557286892759)
--(axis cs:8,0.912557286892759);

\path [draw=black, semithick]
(axis cs:9,0.912557286892759)
--(axis cs:9,0.912740604949588);

\path [draw=black, semithick]
(axis cs:10,0.912740604949588)
--(axis cs:10,0.913107241063245);

\path [draw=black, semithick]
(axis cs:11,0.944087992667278)
--(axis cs:11,0.944087992667278);

\path [draw=black, semithick]
(axis cs:12,0.960036663611366)
--(axis cs:12,0.968835930339138);

\path [draw=black, semithick]
(axis cs:13,0.965352887259395)
--(axis cs:13,0.969019248395967);

\path [draw=black, semithick]
(axis cs:14,0.96865261228231)
--(axis cs:14,0.96865261228231);

\path [draw=black, semithick]
(axis cs:15,0.970669110907424)
--(axis cs:15,0.970852428964253);

\path [draw=black, semithick]
(axis cs:16,0.990100824931256)
--(axis cs:16,0.990100824931256);

\path [draw=black, semithick]
(axis cs:17,0.997066911090742)
--(axis cs:17,0.997066911090742);

\path [draw=black, semithick]
(axis cs:18,0.9974335472044)
--(axis cs:18,0.9974335472044);

\path [draw=black, semithick]
(axis cs:19,0.998900091659028)
--(axis cs:19,0.998900091659028);

\draw (axis cs:-0.5-.50,0.960852428964253+.010) node[
  inner sep=1pt,
  fill=white,
  opacity=.6,
  text opacity=1,
  scale=1.0,
  anchor=south west,
  text=orchid227119194,
  rotate=0.0
]{RCW bound};
\draw (axis cs:-0.5-.50,0.832713107241063+.010) node[
  inner sep=1pt,
  text opacity=1,
  scale=1.0,
  anchor=south west,
  text=darkturquoise23190207,
  rotate=0.0
]{PSCW bound};
\draw (axis cs:-0.5-.50,0.902740604949588+.010) node[
  inner sep=1pt,
  text opacity=1,
  scale=1.0,
  anchor=south west,
  text=grey127,
  rotate=0.0
]{SSCW bound};
\draw (axis cs:-0.5-.50,0.0591109074243813+.010) node[
  inner sep=1pt,
  fill=white,
  opacity=.6,
  text opacity=1,
  scale=1.0,
  anchor=south west,
  text=black,
  rotate=0.0
]{SCW bound};
\end{axis}

\end{tikzpicture}

%% file: fairvote_selected_rules_cm_rate_bar_plot_n_c_ge=5.tex
% This file was created with tikzplotlib v0.10.1.
\begin{tikzpicture}

\definecolor{crimson2143940}{RGB}{214,39,40}
\definecolor{darkgrey176}{RGB}{176,176,176}
\definecolor{darkorange25512714}{RGB}{255,127,14}
\definecolor{darkturquoise23190207}{RGB}{23,190,207}
\definecolor{forestgreen4416044}{RGB}{44,160,44}
\definecolor{goldenrod18818934}{RGB}{188,189,34}
\definecolor{grey127}{RGB}{127,127,127}
\definecolor{mediumpurple148103189}{RGB}{148,103,189}
\definecolor{orchid227119194}{RGB}{227,119,194}
\definecolor{sienna1408675}{RGB}{140,86,75}
\definecolor{steelblue31119180}{RGB}{31,119,180}

\begin{axis}[
height=\axisHeight,
tick align=outside,
tick pos=left,
width=\axisWidth,
x grid style={darkgrey176},
xmin=-1, xmax=20,
xtick style={color=black},
xtick={0,1,2,3,4,5,6,7,8,9,10,11,12,13,14,15,16,17,18,19},
xticklabels={IRV,PR,Max,RP,Sch,You,Bal,Nan,Kem,SD,Dod,Plu,Cop,Sla,Bla,Buc,Bor,Coo,KR,Vet},
xtick={0, 1, 2, 3, 4, 5, 6, 7, 8, 9, 10, 11, 12, 13, 14, 15, 16, 17, 18, 19},
xticklabels = {IRV, PR, Max, RP, Sch, You, Bal, Nan, Kem, SD, Dod, Plu, Cop, Sla, Bla, Buc, Bor, Coo, KR, Vet},
xticklabel style={font=\small},
y grid style={darkgrey176},
ylabel={CM rate},
ymajorgrids,
ymin=0, ymax=1.05,
ytick={0.0, 0.1, 0.2, 0.30000000000000004, 0.4, 0.5, 0.6000000000000001, 0.7000000000000001, 0.8, 0.9, 1.0},
ytick style={color=black}
]
\path [draw=orchid227119194, semithick, dash pattern=on 5.55pt off 2.4pt]
(axis cs:-1,0.717878622197922)
--(axis cs:20,0.717878622197922);

\path [draw=darkturquoise23190207, semithick, dash pattern=on 5.55pt off 2.4pt]
(axis cs:-1,0.52159650082012)
--(axis cs:20,0.52159650082012);

\path [draw=grey127, semithick, dash pattern=on 5.55pt off 2.4pt]
(axis cs:-1,0.595680699835976)
--(axis cs:20,0.595680699835976);

\path [draw=black, semithick, dash pattern=on 5.55pt off 2.4pt]
(axis cs:-1,0.0508474576271186)
--(axis cs:20,0.0508474576271186);

\draw[draw=none,fill=black] (axis cs:-0.4,0) rectangle (axis cs:0.4,0.0508474576271186);
\draw[draw=none,fill=steelblue31119180] (axis cs:0.6,0) rectangle (axis cs:1.4,0.129032258064516);
\draw[draw=none,fill=darkturquoise23190207] (axis cs:1.6,0) rectangle (axis cs:2.4,0.521323127392017);
\draw[draw=none,fill=darkturquoise23190207] (axis cs:2.6,0) rectangle (axis cs:3.4,0.521323127392017);
\draw[draw=none,fill=darkturquoise23190207] (axis cs:3.6,0) rectangle (axis cs:4.4,0.521323127392017);
\draw[draw=none,fill=darkturquoise23190207] (axis cs:4.6,0) rectangle (axis cs:5.4,0.521323127392017);
\draw[draw=none,fill=grey127] (axis cs:5.6,0) rectangle (axis cs:6.4,0.59513395297977);
\draw[draw=none,fill=grey127] (axis cs:6.6,0) rectangle (axis cs:7.4,0.595680699835976);
\draw[draw=none,fill=grey127] (axis cs:7.6,0) rectangle (axis cs:8.4,0.595680699835976);
\draw[draw=none,fill=grey127] (axis cs:8.6,0) rectangle (axis cs:9.4,0.595680699835976);
\draw[draw=none,fill=grey127] (axis cs:9.6,0) rectangle (axis cs:10.4,0.595680699835976);
\draw[draw=none,fill=mediumpurple148103189] (axis cs:10.6,0) rectangle (axis cs:11.4,0.644067796610169);
\draw[draw=none,fill=orchid227119194] (axis cs:11.6,0) rectangle (axis cs:12.4,0.674138873701476);
\draw[draw=none,fill=orchid227119194] (axis cs:12.6,0) rectangle (axis cs:13.4,0.694641880809185);
\draw[draw=none,fill=orchid227119194] (axis cs:13.6,0) rectangle (axis cs:14.4,0.717878622197922);
\draw[draw=none,fill=crimson2143940] (axis cs:14.6,0) rectangle (axis cs:15.4,0.728813559322034);
\draw[draw=none,fill=forestgreen4416044] (axis cs:15.6,0) rectangle (axis cs:16.4,0.864406779661017);
\draw[draw=none,fill=goldenrod18818934] (axis cs:16.6,0) rectangle (axis cs:17.4,0.983050847457627);
\draw[draw=none,fill=darkorange25512714] (axis cs:17.6,0) rectangle (axis cs:18.4,0.998359759431383);
\draw[draw=none,fill=sienna1408675] (axis cs:18.6,0) rectangle (axis cs:19.4,0.999453253143794);
\path [draw=black, semithick]
(axis cs:0,0.0508474576271186)
--(axis cs:0,0.0508474576271186);

\path [draw=black, semithick]
(axis cs:1,0.129032258064516)
--(axis cs:1,0.129032258064516);

\path [draw=black, semithick]
(axis cs:2,0.521323127392017)
--(axis cs:2,0.521323127392017);

\path [draw=black, semithick]
(axis cs:3,0.521323127392017)
--(axis cs:3,0.521323127392017);

\path [draw=black, semithick]
(axis cs:4,0.521323127392017)
--(axis cs:4,0.52159650082012);

\path [draw=black, semithick]
(axis cs:5,0.521323127392017)
--(axis cs:5,0.52159650082012);

\path [draw=black, semithick]
(axis cs:6,0.59513395297977)
--(axis cs:6,0.595680699835976);

\path [draw=black, semithick]
(axis cs:7,0.595680699835976)
--(axis cs:7,0.595680699835976);

\path [draw=black, semithick]
(axis cs:8,0.595680699835976)
--(axis cs:8,0.595680699835976);

\path [draw=black, semithick]
(axis cs:9,0.595680699835976)
--(axis cs:9,0.595680699835976);

\path [draw=black, semithick]
(axis cs:10,0.595680699835976)
--(axis cs:10,0.64406779661017);

\path [draw=black, semithick]
(axis cs:11,0.644067796610169)
--(axis cs:11,0.644067796610169);

\path [draw=black, semithick]
(axis cs:12,0.674138873701476)
--(axis cs:12,0.710497539639147);

\path [draw=black, semithick]
(axis cs:13,0.694641880809185)
--(axis cs:13,0.711864406779661);

\path [draw=black, semithick]
(axis cs:14,0.717878622197922)
--(axis cs:14,0.717878622197922);

\path [draw=black, semithick]
(axis cs:15,0.728813559322034)
--(axis cs:15,0.728813559322034);

\path [draw=black, semithick]
(axis cs:16,0.864406779661017)
--(axis cs:16,0.864406779661017);

\path [draw=black, semithick]
(axis cs:17,0.983050847457627)
--(axis cs:17,0.983050847457627);

\path [draw=black, semithick]
(axis cs:18,0.998359759431383)
--(axis cs:18,0.998359759431383);

\path [draw=black, semithick]
(axis cs:19,0.999453253143794)
--(axis cs:19,0.999453253143794);

\draw (axis cs:-0.5-.50,0.707878622197922+.010) node[
  inner sep=1pt,
  fill=white,
  opacity=.6,
  text opacity=1,
  scale=1.0,
  anchor=south west,
  text=orchid227119194,
  rotate=0.0
]{RCW bound};
\draw (axis cs:-0.5-.50,0.51159650082012+.010) node[
  inner sep=1pt,
  fill=white,
  opacity=.6,
  text opacity=1,
  scale=1.0,
  anchor=south west,
  text=darkturquoise23190207,
  rotate=0.0
]{PSCW bound};
\draw (axis cs:-0.5-.50,0.585680699835976+.010) node[
  inner sep=1pt,
  fill=white,
  opacity=.6,
  text opacity=1,
  scale=1.0,
  anchor=south west,
  text=grey127,
  rotate=0.0
]{SSCW bound};
\draw (axis cs:-0.5-.50,0.0408474576271186+.010) node[
  inner sep=1pt,
  fill=white,
  opacity=.6,
  text opacity=1,
  scale=1.0,
  anchor=south west,
  text=black,
  rotate=0.0
]{SCW bound};
\end{axis}

\end{tikzpicture}

%% file: plot_theoretical_results.tex
% This file was created with tikzplotlib v0.10.1.
\begin{tikzpicture}

\definecolor{crimson2143940}{RGB}{214,39,40}
\definecolor{darkgrey176}{RGB}{176,176,176}
\definecolor{darkorange25512714}{RGB}{255,127,14}
\definecolor{darkturquoise23190207}{RGB}{23,190,207}
\definecolor{forestgreen4416044}{RGB}{44,160,44}
\definecolor{goldenrod18818934}{RGB}{188,189,34}
\definecolor{grey127}{RGB}{127,127,127}
\definecolor{mediumpurple148103189}{RGB}{148,103,189}
\definecolor{orchid227119194}{RGB}{227,119,194}
\definecolor{sienna1408675}{RGB}{140,86,75}
\definecolor{steelblue31119180}{RGB}{31,119,180}

\begin{axis}[clip=false,
height=\axisHeight,
tick align=outside,
tick pos=left,
width=\axisWidth,
x grid style={darkgrey176},
xlabel={Number of candidates $m$},
xmin=3, xmax=11,
xtick distance={1}, xtick style={color=black},
y grid style={darkgrey176},
ylabel={$\theta_c(f, m)$},
ymin=-0.05, ymax=1.05,
ytick={-0.2, 0.0, 0.2, 0.4000000000000001, 0.6000000000000001, 0.8, 1.0000000000000002, 1.2000000000000002},
ytick style={color=black}
]
\addplot [semithick, sienna1408675]
table {%
3 1
4 1
5 1
6 1
7 1
8 1
9 1
10 1
11 1
};
\addplot [semithick, darkorange25512714]
table {%
3 0.333333333333333
4 0.5
5 0.6
6 0.666666666666667
7 0.714285714285714
8 0.75
9 0.777777777777778
10 0.8
11 0.818181818181818
};
\addplot [semithick, forestgreen4416044]
table {%
3 0.25
4 0.4
5 0.5
6 0.571428571428571
7 0.625
8 0.666666666666667
9 0.7
10 0.727272727272727
11 0.75
};
\addplot [semithick, crimson2143940]
table {%
3 0.25
4 0.333333333333333
5 0.375
6 0.4
7 0.416666666666667
8 0.428571428571429
9 0.4375
10 0.444444444444444
11 0.45
};
\addplot [semithick, goldenrod18818934]
table {%
3 0.25
4 0.272727272727273
5 0.285714285714286
6 0.294117647058824
7 0.3
8 0.304347826086957
9 0.307692307692308
10 0.310344827586207
11 0.3125
};
\addplot [semithick, mediumpurple148103189]
table {%
3 0.142857142857143
4 0.2
5 0.230769230769231
6 0.25
7 0.263157894736842
8 0.272727272727273
9 0.28
10 0.285714285714286
11 0.290322580645161
};
\addplot [semithick, orchid227119194]
table {%
3 0.25
4 0.25
5 0.25
6 0.25
7 0.25
8 0.25
9 0.25
10 0.25
11 0.25
};
\addplot [semithick, grey127]
table {%
3 0.142857142857143
4 0.181818181818182
5 0.2
6 0.210526315789474
7 0.217391304347826
8 0.222222222222222
9 0.225806451612903
10 0.228571428571429
11 0.230769230769231
};
\addplot [semithick, steelblue31119180]
table {%
3 0
4 0.0588235294117647
5 0.0909090909090909
6 0.111111111111111
7 0.125
8 0.135135135135135
9 0.142857142857143
10 0.148936170212766
11 0.153846153846154
};
\addplot [semithick, darkturquoise23190207]
table {%
3 0.142857142857143
4 0.142857142857143
5 0.142857142857143
6 0.142857142857143
7 0.142857142857143
8 0.142857142857143
9 0.142857142857143
10 0.142857142857143
11 0.142857142857143
};
\addplot [semithick, black]
table {%
3 0
4 0
5 0
6 0
7 0
8 0
9 0
10 0
11 0
};
\draw (axis cs:11.05,1) node[
  scale=1.0,
  anchor=west,
  text=sienna1408675,
  rotate=0.0
]{\footnotesize Vet};
\draw (axis cs:11.05,0.818181818181818) node[
  scale=1.0,
  anchor=west,
  text=darkorange25512714,
  rotate=0.0
]{\footnotesize KR};
\draw (axis cs:11.05,0.75) node[
  scale=1.0,
  anchor=west,
  text=forestgreen4416044,
  rotate=0.0
]{\footnotesize Bor};
\draw (axis cs:11.05,0.45) node[
  scale=1.0,
  anchor=west,
  text=crimson2143940,
  rotate=0.0
]{\footnotesize Buc};
\draw (axis cs:11.05,0.345) node[
  scale=1.0,
  anchor=base west,
  text=goldenrod18818934,
  rotate=0.0
]{\footnotesize Coo};
\draw (axis cs:11.05,0.29) node[
  scale=1.0,
  anchor=base west,
  text=mediumpurple148103189,
  rotate=0.0
]{\footnotesize Plu};
\draw (axis cs:11.05,0.235) node[
  scale=1.0,
  anchor=base west,
  text=orchid227119194,
  rotate=0.0
]{\footnotesize Bla, Sla$^\dagger$, Cop$^\dagger$};
\draw (axis cs:11.05,0.18) node[
  scale=1.0,
  anchor=base west,
  text=grey127,
  rotate=0.0
]{\footnotesize Bal, Nan, Kem, Dod, SD};
\draw (axis cs:11.05,0.125) node[
  scale=1.0,
  anchor=base west,
  text=steelblue31119180,
  rotate=0.0
]{\footnotesize PR};
\draw (axis cs:11.05,0.07) node[
  scale=1.0,
  anchor=base west,
  text=darkturquoise23190207,
  rotate=0.0
]{\footnotesize Max, RP, Sch, You};
\draw (axis cs:11.05,0) node[
  scale=1.0,
  anchor=west,
  text=black,
  rotate=0.0
]{\footnotesize IRV};
\end{axis}

\end{tikzpicture}

%% file: netflix_selected_rules_nb_candidates_rate_line_plot.tex
% This file was created with tikzplotlib v0.10.1.
\begin{tikzpicture}

\definecolor{crimson2143940}{RGB}{214,39,40}
\definecolor{darkgrey176}{RGB}{176,176,176}
\definecolor{darkorange25512714}{RGB}{255,127,14}
\definecolor{darkturquoise23190207}{RGB}{23,190,207}
\definecolor{forestgreen4416044}{RGB}{44,160,44}
\definecolor{goldenrod18818934}{RGB}{188,189,34}
\definecolor{grey127}{RGB}{127,127,127}
\definecolor{lightgrey204}{RGB}{204,204,204}
\definecolor{mediumpurple148103189}{RGB}{148,103,189}
\definecolor{orchid227119194}{RGB}{227,119,194}
\definecolor{sienna1408675}{RGB}{140,86,75}
\definecolor{steelblue31119180}{RGB}{31,119,180}

\colorlet{Vet}{sienna1408675}\colorlet{KR}{darkorange25512714}\colorlet{Bor}{forestgreen4416044}\colorlet{Buc}{crimson2143940}\colorlet{Coo}{goldenrod18818934}\colorlet{Plu}{mediumpurple148103189}\colorlet{Bla}{orchid227119194}\colorlet{Cop}{orchid227119194}\colorlet{Sla}{orchid227119194}\colorlet{Bal}{grey127}\colorlet{Nan}{grey127}\colorlet{Kem}{grey127}\colorlet{Dod}{grey127}\colorlet{SD}{grey127}\colorlet{PR}{steelblue31119180}\colorlet{Max}{darkturquoise23190207}\colorlet{RP}{darkturquoise23190207}\colorlet{Sch}{darkturquoise23190207}\colorlet{You}{darkturquoise23190207}\colorlet{IRV}{black}

\begin{axis}[clip=false,
height=\axisHeight,
legend cell align={left},
legend style={
  fill opacity=1,
  draw opacity=1,
  text opacity=1,
  at={(1.05,1)},
  anchor=north west,
  draw=lightgrey204
},
tick align=outside,
tick pos=left,
width=\axisSmallerWidth,
x grid style={darkgrey176},
xlabel={Number of candidates $m$},
%xmajorgrids,
xmin=3, xmax=11,
xtick distance={1}, xtick style={color=black},
y grid style={darkgrey176},
ylabel={CM rate},
ymajorgrids,
ymin=0.555409836065574, ymax=1.04, restrict y to domain=0.55:1.04
ytick={0.0, 0.1, 0.2, 0.30000000000000004, 0.4, 0.5, 0.6000000000000001, 0.7000000000000001, 0.8, 0.9, 1.0},
ytick distance={.1}, ytick style={color=black}
]
\addplot [semithick, forestgreen4416044]
table {%
3 0.765106382978723
4 0.953020134228188
5 0.975081967213115
6 0.994713656387665
7 0.99390243902439
8 1
9 0.990476190476191
10 1
11 1
};
%\addlegendentry{Bor}
\addplot [semithick, sienna1408675]
table {%
3 0.573333333333333
4 0.865771812080537
5 0.99016393442623
6 0.999118942731278
7 1
8 1
9 1
10 1
11 1
};
%\addlegendentry{Vet}
\addplot [semithick, darkorange25512714]
table {%
3 0.56709219858156
4 0.868903803131991
5 0.991475409836066
6 0.999118942731278
7 1
8 1
9 1
10 1
11 1
};
%\addlegendentry{KR}
\addplot [semithick, goldenrod18818934]
table {%
3 0.891631205673759
4 0.995078299776286
5 0.996065573770492
6 1
7 1
8 1
9 1
10 1
11 1
};
%\addlegendentry{Coo}
\addplot [semithick, mediumpurple148103189]
table {%
3 0.589787234042553
4 0.882774049217002
5 0.931803278688525
6 0.971806167400881
7 0.986585365853659
8 0.99236641221374
9 0.973333333333333
10 0.988505747126437
11 1
};
%\addlegendentry{Plu}
\addplot [semithick, orchid227119194]
table {%
3 0.756028368794326
4 0.928859060402685
5 0.952131147540983
6 0.970044052863436
7 0.978048780487805
8 0.989312977099237
9 0.967619047619048
10 0.977011494252874
11 0.997222222222222
};
%\addlegendentry{Bla*}
\addplot [semithick, orchid227119194]
table {%
3 0.675744680851064
4 0.92751677852349
5 0.942950819672131
6 0.964757709251101
7 0.975609756097561
8 0.984732824427481
9 0.961904761904762
10 0.977011494252874
11 0.994444444444444
};
%\addlegendentry{Sla*}
\addplot [semithick, orchid227119194]
table {%
3 0.775602836879433
4 0.90738255033557
5 0.937704918032787
6 0.959471365638767
7 0.969512195121951
8 0.984732824427481
9 0.954285714285714
10 0.970114942528736
11 0.986111111111111
};
%\addlegendentry{Cop*}
\addplot [semithick, steelblue31119180]
table {%
3 0.0178723404255319
4 0.276957494407159
5 0.55672131147541
6 0.73568281938326
7 0.852439024390244
8 0.914503816793893
9 0.851428571428572
10 0.933333333333333
11 0.977777777777778
};
%\addlegendentry{PR}
\addplot [semithick, grey127]
table {%
3 0.588368794326241
4 0.808948545861298
5 0.868852459016393
6 0.885462555066079
7 0.934146341463415
8 0.957251908396947
9 0.927619047619048
10 0.958620689655172
11 0.972222222222222
};
%\addlegendentry{Bal}
\addplot [semithick, grey127]
table {%
3 0.588368794326241
4 0.808948545861298
5 0.869508196721311
6 0.885462555066079
7 0.934146341463415
8 0.957251908396947
9 0.927619047619048
10 0.958620689655172
11 0.972222222222222
};
%\addlegendentry{Kem}
\addplot [semithick, crimson2143940]
table {%
3 0.52709219858156
4 0.88993288590604
5 0.897704918032787
6 0.948898678414097
7 0.958536585365854
8 0.981679389312977
9 0.952380952380952
10 0.977011494252874
11 0.972222222222222
};
%\addlegendentry{Buc}
\addplot [semithick, grey127]
table {%
3 0.588368794326241
4 0.808948545861298
5 0.869508196721311
6 0.885462555066079
7 0.934146341463415
8 0.957251908396947
9 0.927619047619048
10 0.958620689655172
11 0.972222222222222
};
%\addlegendentry{Nan}
\addplot [semithick, grey127]
table {%
3 0.588368794326241
4 0.808501118568233
5 0.870163934426229
6 0.885462555066079
7 0.934146341463415
8 0.96030534351145
9 0.925714285714286
10 0.958620689655172
11 0.969444444444444
};
%\addlegendentry{SD}
\addplot [semithick, grey127]
table {%
3 0.588368794326241
4 0.808053691275168
5 0.866229508196721
6 0.884581497797357
7 0.932926829268293
8 0.957251908396947
9 0.914285714285714
10 0.954022988505747
11 0.969444444444444
};
%\addlegendentry{Dod*}
\addplot [semithick, darkturquoise23190207]
table {%
3 0.588368794326241
4 0.725279642058166
5 0.809836065573771
6 0.825550660792952
7 0.870731707317073
8 0.887022900763359
9 0.8
10 0.857471264367816
11 0.930555555555556
};
%\addlegendentry{Max}
\addplot [semithick, darkturquoise23190207]
table {%
3 0.588368794326241
4 0.725279642058166
5 0.809836065573771
6 0.825550660792952
7 0.870731707317073
8 0.887022900763359
9 0.8
10 0.857471264367816
11 0.930555555555556
};
%\addlegendentry{Sch}
\addplot [semithick, darkturquoise23190207]
table {%
3 0.588368794326241
4 0.725279642058166
5 0.809180327868853
6 0.825550660792952
7 0.870731707317073
8 0.887022900763359
9 0.8
10 0.857471264367816
11 0.930555555555556
};
%\addlegendentry{RP}
\addplot [semithick, darkturquoise23190207]
table {%
3 0.586950354609929
4 0.723937360178971
5 0.805245901639344
6 0.824669603524229
7 0.869512195121951
8 0.883969465648855
9 0.786666666666667
10 0.852873563218391
11 0.927777777777778
};
%\addlegendentry{You*}
\addplot [semithick, black]
table {%
3 0.0178723404255319
4 0.0371364653243848
5 0.0655737704918033
6 0.0775330396475771
7 0.0585365853658537
8 0.067175572519084
9 0.0552380952380952
10 0.0735632183908046
11 0.0944444444444444
};
%\addlegendentry{IRV}

\draw (axis cs:11.05,1.020) node[
  anchor=base west,
]{\footnotesize\textcolor{Coo}{Coo}, \textcolor{Vet}{Vet}, \textcolor{KR}{KR}, \textcolor{Bor}{Bor}, \textcolor{Plu}{Plu}};
\draw (axis cs:11.05,0.990) node[
  anchor=base west,
]{\footnotesize\textcolor{Bla}{Bla*, Sla*, Cop*}};
\draw (axis cs:11.05,0.960) node[
  anchor=base west,
]{\footnotesize\textcolor{Buc}{Buc}};
\draw (axis cs:11.05,0.930) node[
  anchor=base west,
]{\footnotesize\textcolor{Bal}{Bal, Nan, Kem, Dod*, SD}};
\draw (axis cs:11.05,0.900) node[
  anchor=base west,
]{\footnotesize\textcolor{PR}{PR}};
\draw (axis cs:11.05,0.870) node[
  anchor=base west,
]{\footnotesize\textcolor{Max}{Max, RP, Sch, You*}};

\end{axis}

\end{tikzpicture}

%% file: netflix_selected_rules_scatter_rank_theta_c_rank_cm_rate_n_c=5_revised.tex
% This file was created with tikzplotlib v0.10.1.
\begin{tikzpicture}

\definecolor{crimson2143940}{RGB}{214,39,40}
\definecolor{darkgrey176}{RGB}{176,176,176}
\definecolor{darkorange25512714}{RGB}{255,127,14}
\definecolor{darkturquoise23190207}{RGB}{23,190,207}
\definecolor{forestgreen4416044}{RGB}{44,160,44}
\definecolor{goldenrod18818934}{RGB}{188,189,34}
\definecolor{grey127}{RGB}{127,127,127}
\definecolor{mediumpurple148103189}{RGB}{148,103,189}
\definecolor{orchid227119194}{RGB}{227,119,194}
\definecolor{sienna1408675}{RGB}{140,86,75}
\definecolor{steelblue31119180}{RGB}{31,119,180}

\begin{axis}[clip=false,
height=\axisHeight,
tick align=outside,
tick pos=left,
width=\axisWidth,
x grid style={darkgrey176},
xlabel={Rank by $\theta_c(f, 5)$},
xmajorgrids,
xmin=0, xmax=20,
xtick distance={5}, xtick style={color=black},
y grid style={darkgrey176},
ylabel={Rank by CM rate ($m = 5$)},
ymajorgrids,
ymin=0, ymax=20,
ytick distance={5},
ytick style={color=black}
]
\addplot [line width=0.32pt, black]
table {%
0 0
20 20
};
\path [draw=black, thin]
(axis cs:0,0)
--(axis cs:0,0);

\path [draw=steelblue31119180, thin]
(axis cs:1,1)
--(axis cs:1,1);

\path [draw=darkturquoise23190207, thin]
(axis cs:3.35,2)
--(axis cs:3.35,5);

\path [draw=darkturquoise23190207, thin]
(axis cs:3.45,2)
--(axis cs:3.45,5);

\path [draw=darkturquoise23190207, thin]
(axis cs:3.55,2)
--(axis cs:3.55,5);

\path [draw=darkturquoise23190207, thin]
(axis cs:3.65,2)
--(axis cs:3.65,5);

\path [draw=grey127, thin]
(axis cs:7.8,6)
--(axis cs:7.8,10);

\path [draw=grey127, thin]
(axis cs:7.9,6)
--(axis cs:7.9,10);

\path [draw=grey127, thin]
(axis cs:8,6)
--(axis cs:8,10);

\path [draw=grey127, thin]
(axis cs:8.1,6)
--(axis cs:8.1,10);

\path [draw=grey127, thin]
(axis cs:8.2,6)
--(axis cs:8.2,14);

\path [draw=crimson2143940, thin]
(axis cs:16,10)
--(axis cs:16,11);

\path [draw=mediumpurple148103189, thin]
(axis cs:11,11)
--(axis cs:11,14);

\path [draw=orchid227119194, thin]
(axis cs:12.9,11)
--(axis cs:12.9,15);

\path [draw=orchid227119194, thin]
(axis cs:13,11)
--(axis cs:13,15);

\path [draw=orchid227119194, thin]
(axis cs:13.1,13)
--(axis cs:13.1,15);

\path [draw=forestgreen4416044, thin]
(axis cs:17,16)
--(axis cs:17,18);

\path [draw=sienna1408675, thin]
(axis cs:19,16)
--(axis cs:19,19);

\path [draw=darkorange25512714, thin]
(axis cs:18,16)
--(axis cs:18,19);

\path [draw=white, thick] % To hide grid locally
(axis cs:15,17-.15)
--(axis cs:15,19+.15);
\path [draw=goldenrod18818934, thin]
(axis cs:15,17)
--(axis cs:15,19);

\addplot [semithick, black, mark=*, mark size=1, mark options={solid}, only marks]
table {%
0 0
};
\addplot [semithick, steelblue31119180, mark=*, mark size=1, mark options={solid}, only marks]
table {%
1 1
};
\addplot [semithick, darkturquoise23190207, mark=*, mark size=1, mark options={solid}, only marks]
table {%
3.35 3.5
};
\addplot [semithick, darkturquoise23190207, mark=*, mark size=1, mark options={solid}, only marks]
table {%
3.45 3.5
};
\addplot [semithick, darkturquoise23190207, mark=*, mark size=1, mark options={solid}, only marks]
table {%
3.55 3.5
};
\addplot [semithick, darkturquoise23190207, mark=*, mark size=1, mark options={solid}, only marks]
table {%
3.65 3.5
};
\addplot [semithick, grey127, mark=*, mark size=1, mark options={solid}, only marks]
table {%
7.8 8
};
\addplot [semithick, grey127, mark=*, mark size=1, mark options={solid}, only marks]
table {%
7.9 8
};
\addplot [semithick, grey127, mark=*, mark size=1, mark options={solid}, only marks]
table {%
8 8
};
\addplot [semithick, grey127, mark=*, mark size=1, mark options={solid}, only marks]
table {%
8.1 8
};
\addplot [semithick, grey127, mark=*, mark size=1, mark options={solid}, only marks]
table {%
8.2 10
};
\addplot [semithick, crimson2143940, mark=*, mark size=1, mark options={solid}, only marks]
table {%
16 10.5
};
\addplot [semithick, mediumpurple148103189, mark=*, mark size=1, mark options={solid}, only marks]
table {%
11 12.5
};
\addplot [semithick, orchid227119194, mark=*, mark size=1, mark options={solid}, only marks]
table {%
12.9 13
};
\addplot [semithick, orchid227119194, mark=*, mark size=1, mark options={solid}, only marks]
table {%
13 13
};
\addplot [semithick, orchid227119194, mark=*, mark size=1, mark options={solid}, only marks]
table {%
13.1 14
};
\addplot [semithick, forestgreen4416044, mark=*, mark size=1, mark options={solid}, only marks]
table {%
17 17
};
\addplot [semithick, sienna1408675, mark=*, mark size=1, mark options={solid}, only marks]
table {%
19 17.5
};
\addplot [semithick, darkorange25512714, mark=*, mark size=1, mark options={solid}, only marks]
table {%
18 17.5
};
\addplot [semithick, goldenrod18818934, mark=*, mark size=1, mark options={solid}, only marks]
table {%
15 18
};

\draw (axis cs:0,0+.3) node[
  scale=1.0,
  anchor=north west,
  text=black,
  rotate=0.0
]{\footnotesize IRV};

\draw (axis cs:1+.3,1) node[
  scale=1.0,
  anchor=west,
  text=steelblue31119180,
  rotate=0.0
]{\footnotesize PR};

\draw (axis cs:3.5+0.3,3.5+.05) node[
  scale=1.0,
  anchor=west,
  text=darkturquoise23190207,
  rotate=0.0
]{\footnotesize Max, RP, Sch, You};

\draw (axis cs:8-.1,8) node[
  scale=1.0,
  anchor=east,
  text=grey127,
  rotate=0.0
]{\footnotesize Bal, Nan, Kem, SD};

\draw (axis cs:8+.4,10+1) node[
  scale=1.0,
  anchor=east,
  text=grey127,
  rotate=0.0
]{\footnotesize Dod};

\draw (axis cs:16,10.5) node[
  scale=1.0,
  anchor=west,
  text=crimson2143940,
  rotate=0.0
]{\footnotesize Buc};

\draw (axis cs:11+.1,12.5) node[
  scale=1.0,
  anchor=east,
  text=mediumpurple148103189,
  rotate=0.0
]{\footnotesize Plu};

\draw (axis cs:13,11+.3) node[
  scale=1.0,
  anchor=north,
  text=orchid227119194,
  rotate=0.0
]{\footnotesize Cop, Sla};

\draw (axis cs:13-.2,15-.2) node[
  scale=1.0,
  anchor=south west,
  text=orchid227119194,
  rotate=0.0
]{\footnotesize Bla};

\draw (axis cs:17,16+.3) node[
  scale=1.0,
  anchor=north,
  text=forestgreen4416044,
  rotate=0.0
]{\footnotesize Bor};

\draw (axis cs:19,16+.3) node[
  scale=1.0,
  anchor=north,
  text=sienna1408675,
  rotate=0.0
]{\footnotesize Vet};

\draw (axis cs:18+.25,19+.2) node[
  scale=1.0,
  anchor=east,
  text=darkorange25512714,
  rotate=0.0
]{\footnotesize KR};

\draw (axis cs:15+.1,18) node[
  scale=1.0,
  anchor=east,
  text=goldenrod18818934,
  rotate=0.0
]{\footnotesize Coo};
\end{axis}

\end{tikzpicture}%

%% file: fairvote_selected_rules_scatter_rank_theta_c_rank_cm_rate_n_c=5_revised.tex
% This file was created with tikzplotlib v0.10.1.
\begin{tikzpicture}

\definecolor{crimson2143940}{RGB}{214,39,40}
\definecolor{darkgrey176}{RGB}{176,176,176}
\definecolor{darkorange25512714}{RGB}{255,127,14}
\definecolor{darkturquoise23190207}{RGB}{23,190,207}
\definecolor{forestgreen4416044}{RGB}{44,160,44}
\definecolor{goldenrod18818934}{RGB}{188,189,34}
\definecolor{grey127}{RGB}{127,127,127}
\definecolor{mediumpurple148103189}{RGB}{148,103,189}
\definecolor{orchid227119194}{RGB}{227,119,194}
\definecolor{sienna1408675}{RGB}{140,86,75}
\definecolor{steelblue31119180}{RGB}{31,119,180}

\begin{axis}[clip=false,
height=\axisHeight,
tick align=outside,
tick pos=left,
width=\axisWidth,
x grid style={darkgrey176},
xlabel={Rank by $\theta_c(f, 5)$},
xmajorgrids,
xmin=0, xmax=20,
xtick distance={5}, xtick style={color=black},
y grid style={darkgrey176},
ymajorgrids,
ymin=0, ymax=20,
ytick distance={5},
yticklabel style={color=white},
ytick style={color=black}
]
\addplot [line width=0.32pt, black]
table {%
0 0
20 20
};
\path [draw=black, thin]
(axis cs:0,0)
--(axis cs:0,0);

\path [draw=steelblue31119180, thin]
(axis cs:1,1)
--(axis cs:1,1);

\path [draw=darkturquoise23190207, thin]
(axis cs:3.35,2)
--(axis cs:3.35,5);

\path [draw=darkturquoise23190207, thin]
(axis cs:3.45,2)
--(axis cs:3.45,5);

\path [draw=darkturquoise23190207, thin]
(axis cs:3.55,2)
--(axis cs:3.55,5);

\path [draw=darkturquoise23190207, thin]
(axis cs:3.65,2)
--(axis cs:3.65,5);

\path [draw=grey127, thin]
(axis cs:7.8,6)
--(axis cs:7.8,10);

\path [draw=grey127, thin]
(axis cs:7.9,6)
--(axis cs:7.9,10);

\path [draw=grey127, thin]
(axis cs:8,6)
--(axis cs:8,10);

\path [draw=grey127, thin]
(axis cs:8.1,6)
--(axis cs:8.1,10);

\path [draw=grey127, thin]
(axis cs:8.2,6)
--(axis cs:8.2,11);

\path [draw=mediumpurple148103189, thin]
(axis cs:11,10)
--(axis cs:11,11);

\path [draw=crimson2143940, thin]
(axis cs:16,12)
--(axis cs:16,12);

\path [draw=orchid227119194, thin]
(axis cs:12.9,13)
--(axis cs:12.9,15);

\path [draw=orchid227119194, thin]
(axis cs:13,13)
--(axis cs:13,15);

\path [draw=orchid227119194, thin]
(axis cs:13.1,13)
--(axis cs:13.1,15);

\path [draw=forestgreen4416044, thin]
(axis cs:17,16)
--(axis cs:17,16);

\path [draw=white, thick] % To hide grid locally
(axis cs:15,17-.3)
--(axis cs:15,17+.3);
\path [draw=goldenrod18818934, thin]
(axis cs:15,17)
--(axis cs:15,17);

\path [draw=darkorange25512714, thin]
(axis cs:18,18)
--(axis cs:18,19);

\path [draw=sienna1408675, thin]
(axis cs:19,18)
--(axis cs:19,19);

\addplot [semithick, black, mark=*, mark size=1, mark options={solid}, only marks]
table {%
0 0
};
\addplot [semithick, steelblue31119180, mark=*, mark size=1, mark options={solid}, only marks]
table {%
1 1
};
\addplot [semithick, darkturquoise23190207, mark=*, mark size=1, mark options={solid}, only marks]
table {%
3.35 3.5
};
\addplot [semithick, darkturquoise23190207, mark=*, mark size=1, mark options={solid}, only marks]
table {%
3.45 3.5
};
\addplot [semithick, darkturquoise23190207, mark=*, mark size=1, mark options={solid}, only marks]
table {%
3.55 3.5
};
\addplot [semithick, darkturquoise23190207, mark=*, mark size=1, mark options={solid}, only marks]
table {%
3.65 3.5
};
\addplot [semithick, grey127, mark=*, mark size=1, mark options={solid}, only marks]
table {%
7.8 8
};
\addplot [semithick, grey127, mark=*, mark size=1, mark options={solid}, only marks]
table {%
7.9 8
};
\addplot [semithick, grey127, mark=*, mark size=1, mark options={solid}, only marks]
table {%
8 8
};
\addplot [semithick, grey127, mark=*, mark size=1, mark options={solid}, only marks]
table {%
8.1 8
};
\addplot [semithick, grey127, mark=*, mark size=1, mark options={solid}, only marks]
table {%
8.2 8.5
};
\addplot [semithick, mediumpurple148103189, mark=*, mark size=1, mark options={solid}, only marks]
table {%
11 10.5
};
\addplot [semithick, crimson2143940, mark=*, mark size=1, mark options={solid}, only marks]
table {%
16 12
};
\addplot [semithick, orchid227119194, mark=*, mark size=1, mark options={solid}, only marks]
table {%
12.9 14
};
\addplot [semithick, orchid227119194, mark=*, mark size=1, mark options={solid}, only marks]
table {%
13 14
};
\addplot [semithick, orchid227119194, mark=*, mark size=1, mark options={solid}, only marks]
table {%
13.1 14
};
\addplot [semithick, forestgreen4416044, mark=*, mark size=1, mark options={solid}, only marks]
table {%
17 16
};
\addplot [semithick, goldenrod18818934, mark=*, mark size=1, mark options={solid}, only marks]
table {%
15 17
};
\addplot [semithick, darkorange25512714, mark=*, mark size=1, mark options={solid}, only marks]
table {%
18 18.5
};
\addplot [semithick, sienna1408675, mark=*, mark size=1, mark options={solid}, only marks]
table {%
19 18.5
};

\draw (axis cs:0,0+.3) node[
  scale=1.0,
  anchor=north west,
  text=black,
  rotate=0.0
]{\footnotesize IRV};

\draw (axis cs:1+.3,1) node[
  scale=1.0,
  anchor=west,
  text=steelblue31119180,
  rotate=0.0
]{\footnotesize PR};

\draw (axis cs:3.5+0.3,3.5+.05) node[
  scale=1.0,
  anchor=west,
  text=darkturquoise23190207,
  rotate=0.0
]{\footnotesize Max, RP, Sch, You};

\draw (axis cs:8-.1,8) node[
  scale=1.0,
  anchor=east,
  text=grey127,
  rotate=0.0
]{\footnotesize Bal, Nan, Kem, SD};

\draw (axis cs:8+.4,10+1) node[
  scale=1.0,
  anchor=east,
  text=grey127,
  rotate=0.0
]{\footnotesize Dod};

\draw (axis cs:11,10.5) node[
  scale=1.0,
  anchor=west,
  text=mediumpurple148103189,
  rotate=0.0
]{\footnotesize Plu};

\draw (axis cs:16,12) node[
  scale=1.0,
  anchor=west,
  text=crimson2143940,
  rotate=0.0
]{\footnotesize Buc};

\draw (axis cs:13-.1,14) node[
  scale=1.0,
  anchor=east,
  text=orchid227119194,
  rotate=0.0
]{\footnotesize Cop, Sla, Bla};

\draw (axis cs:17,16) node[
  scale=1.0,
  anchor=west,
  text=forestgreen4416044,
  rotate=0.0
]{\footnotesize Bor};

\draw (axis cs:15+.1,17) node[
  scale=1.0,
  anchor=east,
  text=goldenrod18818934,
  rotate=0.0
]{\footnotesize Coo};

\draw (axis cs:18+.25,19+.2) node[
  scale=1.0,
  anchor=east,
  text=darkorange25512714,
  rotate=0.0
]{\footnotesize KR};

\draw (axis cs:19,18+.3) node[
  scale=1.0,
  anchor=north,
  text=sienna1408675,
  rotate=0.0
]{\footnotesize Vet};
\end{axis}

\end{tikzpicture}%

%% file: netflix_more_rules_cm_rate_bar_plot.tex
% This file was created with tikzplotlib v0.10.1.
\begin{tikzpicture}

\definecolor{crimson2143940}{RGB}{214,39,40}
\definecolor{darkgrey176}{RGB}{176,176,176}
\definecolor{darkorange25512714}{RGB}{255,127,14}
\definecolor{darkturquoise23190207}{RGB}{23,190,207}
\definecolor{forestgreen4416044}{RGB}{44,160,44}
\definecolor{goldenrod18818934}{RGB}{188,189,34}
\definecolor{grey127}{RGB}{127,127,127}
\definecolor{mediumpurple148103189}{RGB}{148,103,189}
\definecolor{orchid227119194}{RGB}{227,119,194}
\definecolor{sienna1408675}{RGB}{140,86,75}
\definecolor{steelblue31119180}{RGB}{31,119,180}

\begin{axis}[
height=\axisHeight,
tick align=outside,
tick pos=left,
width=\axisWidth,
x grid style={darkgrey176},
xmin=-1, xmax=22,
xtick style={color=black},
xtick={0,1,2,3,4,5,6,7,8,9,10,11,12,13,14,15,16,17,18,19,20,21},
xticklabels={IRV,PR,You,Vie,RP,SC,Max,Sch,Dod,Bal,Nan,Kem,SD,Buc,Plu,KR,Vet,Sla,Cop,Bla,Bor,Coo},
xtick={0, 1, 2, 3, 4, 5, 6, 7, 8, 9, 10, 11, 12, 13, 14, 15, 16, 17, 18, 19, 20, 21},
xticklabels = {IRV, PR, You, Vie, RP, SC, Max, Sch, Dod, Bal, Nan, Kem, SD, Buc, Plu, KR, Vet, Sla, Cop, Bla, Bor, Coo},
xticklabel style={font=\small},
y grid style={darkgrey176},
ylabel={CM rate},
ymajorgrids,
ymin=0, ymax=1.05,
ytick={0.0, 0.1, 0.2, 0.30000000000000004, 0.4, 0.5, 0.6000000000000001, 0.7000000000000001, 0.8, 0.9, 1.0},
ytick style={color=black}
]
\path [draw=orchid227119194, semithick, dash pattern=on 5.55pt off 2.4pt]
(axis cs:-1,0.928756130182791)
--(axis cs:22,0.928756130182791);

\path [draw=darkturquoise23190207, semithick, dash pattern=on 5.55pt off 2.4pt]
(axis cs:-1,0.74079358002675)
--(axis cs:22,0.74079358002675);

\path [draw=grey127, semithick, dash pattern=on 5.55pt off 2.4pt]
(axis cs:-1,0.791440035666518)
--(axis cs:22,0.791440035666518);

\path [draw=black, semithick, dash pattern=on 5.55pt off 2.4pt]
(axis cs:-1,0.0471689701292912)
--(axis cs:22,0.0471689701292912);

\draw[draw=none,fill=black] (axis cs:-0.4,0) rectangle (axis cs:0.4,0.0464556397681676);
\draw[draw=none,fill=steelblue31119180] (axis cs:0.6,0) rectangle (axis cs:1.4,0.434150691038787);
\draw[draw=none,fill=darkturquoise23190207] (axis cs:1.6,0) rectangle (axis cs:2.4,0.736602764155149);
\draw[draw=none,fill=darkturquoise23190207] (axis cs:2.6,0) rectangle (axis cs:3.4,0.738564422648239);
\draw[draw=none,fill=darkturquoise23190207] (axis cs:3.6,0) rectangle (axis cs:4.4,0.739099420419082);
\draw[draw=none,fill=darkturquoise23190207] (axis cs:4.6,0) rectangle (axis cs:5.4,0.739099420419082);
\draw[draw=none,fill=darkturquoise23190207] (axis cs:5.6,0) rectangle (axis cs:6.4,0.739188586714222);
\draw[draw=none,fill=darkturquoise23190207] (axis cs:6.6,0) rectangle (axis cs:7.4,0.739188586714222);
\draw[draw=none,fill=grey127] (axis cs:7.6,0) rectangle (axis cs:8.4,0.788319215336603);
\draw[draw=none,fill=grey127] (axis cs:8.6,0) rectangle (axis cs:9.4,0.789924208649131);
\draw[draw=none,fill=grey127] (axis cs:9.6,0) rectangle (axis cs:10.4,0.790013374944271);
\draw[draw=none,fill=grey127] (axis cs:10.6,0) rectangle (axis cs:11.4,0.790013374944271);
\draw[draw=none,fill=grey127] (axis cs:11.6,0) rectangle (axis cs:12.4,0.790013374944271);
\draw[draw=none,fill=crimson2143940] (axis cs:12.6,0) rectangle (axis cs:13.4,0.802229157378511);
\draw[draw=none,fill=mediumpurple148103189] (axis cs:13.6,0) rectangle (axis cs:14.4,0.832456531431119);
\draw[draw=none,fill=darkorange25512714] (axis cs:14.6,0) rectangle (axis cs:15.4,0.836558181007579);
\draw[draw=none,fill=sienna1408675] (axis cs:15.6,0) rectangle (axis cs:16.4,0.837717342844405);
\draw[draw=none,fill=orchid227119194] (axis cs:16.6,0) rectangle (axis cs:17.4,0.866785555060187);
\draw[draw=none,fill=orchid227119194] (axis cs:17.6,0) rectangle (axis cs:18.4,0.891573785109229);
\draw[draw=none,fill=orchid227119194] (axis cs:18.6,0) rectangle (axis cs:19.4,0.894872938029425);
\draw[draw=none,fill=forestgreen4416044] (axis cs:19.6,0) rectangle (axis cs:20.4,0.911992866696389);
\draw[draw=none,fill=goldenrod18818934] (axis cs:20.6,0) rectangle (axis cs:21.4,0.964422648238966);
\path [draw=black, semithick]
(axis cs:0,0.0464556397681676)
--(axis cs:0,0.0464556397681676);

\path [draw=black, semithick]
(axis cs:1,0.434150691038787)
--(axis cs:1,0.434150691038787);

\path [draw=black, semithick]
(axis cs:2,0.736602764155149)
--(axis cs:2,0.74079358002675);

\path [draw=black, semithick]
(axis cs:3,0.738564422648239)
--(axis cs:3,0.739634418189924);

\path [draw=black, semithick]
(axis cs:4,0.739099420419082)
--(axis cs:4,0.739634418189924);

\path [draw=black, semithick]
(axis cs:5,0.739099420419082)
--(axis cs:5,0.74079358002675);

\path [draw=black, semithick]
(axis cs:6,0.739188586714222)
--(axis cs:6,0.739634418189924);

\path [draw=black, semithick]
(axis cs:7,0.739188586714222)
--(axis cs:7,0.740347748551048);

\path [draw=black, semithick]
(axis cs:8,0.788319215336603)
--(axis cs:8,0.831119037004012);

\path [draw=black, semithick]
(axis cs:9,0.789924208649131)
--(axis cs:9,0.790459206419973);

\path [draw=black, semithick]
(axis cs:10,0.790013374944271)
--(axis cs:10,0.790459206419973);

\path [draw=black, semithick]
(axis cs:11,0.790013374944271)
--(axis cs:11,0.791440035666518);

\path [draw=black, semithick]
(axis cs:12,0.790013374944271)
--(axis cs:12,0.791618368256799);

\path [draw=black, semithick]
(axis cs:13,0.802229157378511)
--(axis cs:13,0.802229157378511);

\path [draw=black, semithick]
(axis cs:14,0.832456531431119)
--(axis cs:14,0.832456531431119);

\path [draw=black, semithick]
(axis cs:15,0.836558181007579)
--(axis cs:15,0.836558181007579);

\path [draw=black, semithick]
(axis cs:16,0.837717342844405)
--(axis cs:16,0.837717342844405);

\path [draw=black, semithick]
(axis cs:17,0.866785555060187)
--(axis cs:17,0.904948729380294);

\path [draw=black, semithick]
(axis cs:18,0.891573785109229)
--(axis cs:18,0.900490414623272);

\path [draw=black, semithick]
(axis cs:19,0.894872938029425)
--(axis cs:19,0.928756130182791);

\path [draw=black, semithick]
(axis cs:20,0.911992866696389)
--(axis cs:20,0.911992866696389);

\path [draw=black, semithick]
(axis cs:21,0.964422648238966)
--(axis cs:21,0.964422648238966);

\draw (axis cs:-0.5-.50,0.918756130182791+.010) node[
  inner sep=1pt,
  fill=white,
  opacity=.6,
  text opacity=1,
  scale=1.0,
  anchor=south west,
  text=orchid227119194,
  rotate=0.0
]{RCW bound};
\draw (axis cs:-0.5-.50,0.73079358002675+.010) node[
  inner sep=1pt,
  text opacity=1,
  scale=1.0,
  anchor=south west,
  text=darkturquoise23190207,
  rotate=0.0
]{PSCW bound};
\draw (axis cs:-0.5-.50,0.781440035666518+.010) node[
  inner sep=1pt,
  fill=white,
  opacity=.6,
  text opacity=1,
  scale=1.0,
  anchor=south west,
  text=grey127,
  rotate=0.0
]{SSCW bound};
\draw (axis cs:-0.5-.50,0.0371689701292912+.010) node[
  inner sep=1pt,
  fill=white,
  opacity=.6,
  text opacity=1,
  scale=1.0,
  anchor=south west,
  text=black,
  rotate=0.0
]{SCW bound};
\end{axis}

\end{tikzpicture}

%% file: fairvote_more_rules_cm_rate_bar_plot.tex
% This file was created with tikzplotlib v0.10.1.
\begin{tikzpicture}

\definecolor{crimson2143940}{RGB}{214,39,40}
\definecolor{darkgrey176}{RGB}{176,176,176}
\definecolor{darkorange25512714}{RGB}{255,127,14}
\definecolor{darkturquoise23190207}{RGB}{23,190,207}
\definecolor{forestgreen4416044}{RGB}{44,160,44}
\definecolor{goldenrod18818934}{RGB}{188,189,34}
\definecolor{grey127}{RGB}{127,127,127}
\definecolor{mediumpurple148103189}{RGB}{148,103,189}
\definecolor{orchid227119194}{RGB}{227,119,194}
\definecolor{sienna1408675}{RGB}{140,86,75}
\definecolor{steelblue31119180}{RGB}{31,119,180}

\begin{axis}[
height=\axisHeight,
tick align=outside,
tick pos=left,
width=\axisWidth,
x grid style={darkgrey176},
xmin=-1, xmax=22,
xtick style={color=black},
xtick={0,1,2,3,4,5,6,7,8,9,10,11,12,13,14,15,16,17,18,19,20,21},
xticklabels={IRV,PR,Vie,You,Max,RP,SC,Sch,Bal,Dod,Nan,Kem,SD,Plu,Sla,Cop,Buc,Bla,Bor,Coo,KR,Vet},
xtick={0, 1, 2, 3, 4, 5, 6, 7, 8, 9, 10, 11, 12, 13, 14, 15, 16, 17, 18, 19, 20, 21},
xticklabels = {IRV, PR, Vie, You, Max, RP, SC, Sch, Bal, Dod, Nan, Kem, SD, Plu, Sla, Cop, Buc, Bla, Bor, Coo, KR, Vet},
xticklabel style={font=\small},
y grid style={darkgrey176},
ylabel={CM rate},
ymajorgrids,
ymin=0, ymax=1.05,
ytick={0.0, 0.1, 0.2, 0.30000000000000004, 0.4, 0.5, 0.6000000000000001, 0.7000000000000001, 0.8, 0.9, 1.0},
ytick style={color=black}
]
\path [draw=orchid227119194, semithick, dash pattern=on 5.55pt off 2.4pt]
(axis cs:-1,0.590900039824771)
--(axis cs:22,0.590900039824771);

\path [draw=darkturquoise23190207, semithick, dash pattern=on 5.55pt off 2.4pt]
(axis cs:-1,0.351851851851852)
--(axis cs:22,0.351851851851852);

\path [draw=grey127, semithick, dash pattern=on 5.55pt off 2.4pt]
(axis cs:-1,0.384010354440462)
--(axis cs:22,0.384010354440462);

\path [draw=black, semithick, dash pattern=on 5.55pt off 2.4pt]
(axis cs:-1,0.0370370370370371)
--(axis cs:22,0.0370370370370371);

\draw[draw=none,fill=black] (axis cs:-0.4,0) rectangle (axis cs:0.4,0.037037037037037);
\draw[draw=none,fill=steelblue31119180] (axis cs:0.6,0) rectangle (axis cs:1.4,0.0716845878136201);
\draw[draw=none,fill=darkturquoise23190207] (axis cs:1.6,0) rectangle (axis cs:2.4,0.348466746316209);
\draw[draw=none,fill=darkturquoise23190207] (axis cs:2.6,0) rectangle (axis cs:3.4,0.351154918359219);
\draw[draw=none,fill=darkturquoise23190207] (axis cs:3.6,0) rectangle (axis cs:4.4,0.351354042214257);
\draw[draw=none,fill=darkturquoise23190207] (axis cs:4.6,0) rectangle (axis cs:5.4,0.351354042214257);
\draw[draw=none,fill=darkturquoise23190207] (axis cs:5.6,0) rectangle (axis cs:6.4,0.351354042214257);
\draw[draw=none,fill=darkturquoise23190207] (axis cs:6.6,0) rectangle (axis cs:7.4,0.351354042214257);
\draw[draw=none,fill=grey127] (axis cs:7.6,0) rectangle (axis cs:8.4,0.383512544802867);
\draw[draw=none,fill=grey127] (axis cs:8.6,0) rectangle (axis cs:9.4,0.383612106730386);
\draw[draw=none,fill=grey127] (axis cs:9.6,0) rectangle (axis cs:10.4,0.383711668657905);
\draw[draw=none,fill=grey127] (axis cs:10.6,0) rectangle (axis cs:11.4,0.383711668657905);
\draw[draw=none,fill=grey127] (axis cs:11.6,0) rectangle (axis cs:12.4,0.383811230585424);
\draw[draw=none,fill=mediumpurple148103189] (axis cs:12.6,0) rectangle (axis cs:13.4,0.409099960175229);
\draw[draw=none,fill=orchid227119194] (axis cs:13.6,0) rectangle (axis cs:14.4,0.500995619275189);
\draw[draw=none,fill=orchid227119194] (axis cs:14.6,0) rectangle (axis cs:15.4,0.518020708880924);
\draw[draw=none,fill=crimson2143940] (axis cs:15.6,0) rectangle (axis cs:16.4,0.550278773397053);
\draw[draw=none,fill=orchid227119194] (axis cs:16.6,0) rectangle (axis cs:17.4,0.569095977698128);
\draw[draw=none,fill=forestgreen4416044] (axis cs:17.6,0) rectangle (axis cs:18.4,0.728992433293508);
\draw[draw=none,fill=goldenrod18818934] (axis cs:18.6,0) rectangle (axis cs:19.4,0.817502986857826);
\draw[draw=none,fill=darkorange25512714] (axis cs:19.6,0) rectangle (axis cs:20.4,0.832536837913182);
\draw[draw=none,fill=sienna1408675] (axis cs:20.6,0) rectangle (axis cs:21.4,0.853942652329749);
\path [draw=black, semithick]
(axis cs:0,0.037037037037037)
--(axis cs:0,0.037037037037037);

\path [draw=black, semithick]
(axis cs:1,0.0716845878136201)
--(axis cs:1,0.0716845878136201);

\path [draw=black, semithick]
(axis cs:2,0.348466746316209)
--(axis cs:2,0.351354042214257);

\path [draw=black, semithick]
(axis cs:3,0.351154918359219)
--(axis cs:3,0.351851851851852);

\path [draw=black, semithick]
(axis cs:4,0.351354042214257)
--(axis cs:4,0.351354042214257);

\path [draw=black, semithick]
(axis cs:5,0.351354042214257)
--(axis cs:5,0.351354042214257);

\path [draw=black, semithick]
(axis cs:6,0.351354042214257)
--(axis cs:6,0.351851851851852);

\path [draw=black, semithick]
(axis cs:7,0.351354042214257)
--(axis cs:7,0.351851851851852);

\path [draw=black, semithick]
(axis cs:8,0.383512544802867)
--(axis cs:8,0.383711668657905);

\path [draw=black, semithick]
(axis cs:9,0.383612106730386)
--(axis cs:9,0.409099960175229);

\path [draw=black, semithick]
(axis cs:10,0.383711668657905)
--(axis cs:10,0.383711668657905);

\path [draw=black, semithick]
(axis cs:11,0.383711668657905)
--(axis cs:11,0.384010354440462);

\path [draw=black, semithick]
(axis cs:12,0.383811230585424)
--(axis cs:12,0.384010354440462);

\path [draw=black, semithick]
(axis cs:13,0.409099960175229)
--(axis cs:13,0.409099960175229);

\path [draw=black, semithick]
(axis cs:14,0.500995619275189)
--(axis cs:14,0.559936280366388);

\path [draw=black, semithick]
(axis cs:15,0.518020708880924)
--(axis cs:15,0.552668259657507);

\path [draw=black, semithick]
(axis cs:16,0.550278773397053)
--(axis cs:16,0.550278773397053);

\path [draw=black, semithick]
(axis cs:17,0.569095977698128)
--(axis cs:17,0.590900039824771);

\path [draw=black, semithick]
(axis cs:18,0.728992433293508)
--(axis cs:18,0.728992433293508);

\path [draw=black, semithick]
(axis cs:19,0.817502986857826)
--(axis cs:19,0.817502986857826);

\path [draw=black, semithick]
(axis cs:20,0.832536837913182)
--(axis cs:20,0.832536837913182);

\path [draw=black, semithick]
(axis cs:21,0.853942652329749)
--(axis cs:21,0.853942652329749);

\draw (axis cs:-0.5-.50,0.580900039824771+.010) node[
  inner sep=1pt,
  fill=white,
  opacity=.6,
  text opacity=1,
  scale=1.0,
  anchor=south west,
  text=orchid227119194,
  rotate=0.0
]{RCW bound};
\draw (axis cs:-0.5-.50,0.341851851851852+.010) node[
  inner sep=1pt,
  fill=white,
  opacity=.6,
  text opacity=1,
  scale=1.0,
  anchor=south west,
  text=darkturquoise23190207,
  rotate=0.0
]{PSCW bound};
\draw (axis cs:-0.5-.50,0.374010354440462+.025) node[
  inner sep=1pt,
  text opacity=1,
  scale=1.0,
  anchor=south west,
  text=grey127,
  rotate=0.0
]{SSCW bound};
\draw (axis cs:-0.5-.50,0.0270370370370371+.010) node[
  inner sep=1pt,
  fill=white,
  opacity=.6,
  text opacity=1,
  scale=1.0,
  anchor=south west,
  text=black,
  rotate=0.0
]{SCW bound};
\end{axis}

\end{tikzpicture}